\documentclass[english,11pt, letter]{article}
\usepackage[T1]{fontenc}
\usepackage[latin9]{inputenc}
\usepackage{geometry}
\usepackage{amsmath}
\usepackage{amsthm}
\usepackage{amssymb}

\makeatletter
\numberwithin{equation}{section}
\numberwithin{figure}{section}
\theoremstyle{plain}
\newtheorem{thm}{\protect\theoremname}
  \theoremstyle{plain}
  \newtheorem{lem}[thm]{\protect\lemmaname}
  \theoremstyle{plain}
  \newtheorem{cor}[thm]{\protect\corollaryname}
  \theoremstyle{definition}
  \newtheorem{defn}[thm]{\protect\definitionname}

\pdfoutput=1

\usepackage{bbm}
\usepackage{bbold}
\usepackage[basic]{complexity}
\usepackage[pdftex,bookmarks,colorlinks]{hyperref}
\usepackage{enumitem}
\usepackage{color}

\usepackage{amsmath}
\usepackage{amsthm}
\usepackage{amsfonts}

\usepackage[lined,boxed,ruled,norelsize,algo2e]{algorithm2e}

\hypersetup{colorlinks=true,citecolor=red}

\makeatother

\usepackage{babel}
  \providecommand{\corollaryname}{Corollary}
  \providecommand{\definitionname}{Definition}
  \providecommand{\lemmaname}{Lemma}
\providecommand{\theoremname}{Theorem}

\begin{document}
\global\long\def\R{\mathcal{{R}}}
 \global\long\def\Z{\mathbb{Z}}

\global\long\def\ellOne{\ell_{1}}
 \global\long\def\ellTwo{\ell_{2}}
 \global\long\def\ellInf{\ell_{\infty}}

\global\long\def\boldVar#1{\mathbf{#1}}
\global\long\def\mvar#1{\boldVar{#1}}
\global\long\def\vvar#1{\vec{#1}}


\global\long\def\defeq{\stackrel{\mathrm{{\scriptscriptstyle def}}}{=}}
\global\long\def\E{\mathbb{E}}
\global\long\def\otilde{\tilde{O}}


\global\long\def\gradient{\bigtriangledown}
 \global\long\def\grad{\gradient}
 \global\long\def\hessian{\gradient^{2}}
 \global\long\def\hess{\hessian}
 \global\long\def\jacobian{\mvar J}

 \global\long\def\setVec#1{\onesVec_{#1}}
 \global\long\def\indicVec#1{\onesVec_{#1}}

\global\long\def\specGeq{\succeq}
 \global\long\def\specLeq{\preceq}
 \global\long\def\specGt{\succ}
 \global\long\def\specLt{\prec}

\global\long\def\innerProduct#1#2{\big\langle#1 , #2 \big\rangle}
 \global\long\def\norm#1{\|#1\|}
\global\long\def\normFull#1{\left\Vert #1\right\Vert }

\global\long\def\opt{\mathrm{opt}}

\global\long\def\fopt{f^{*}}

\global\long\def\va{\vvar a}
 \global\long\def\vb{\vvar b}
 \global\long\def\vc{\vvar c}
 \global\long\def\vd{\vvar d}
 \global\long\def\ve{\vvar e}
 \global\long\def\vf{\vvar f}
 \global\long\def\vg{\vvar g}
 \global\long\def\vh{\vvar h}
 \global\long\def\vl{\vvar l}
 \global\long\def\vm{\vvar m}
 \global\long\def\vn{\vvar n}
 \global\long\def\vo{\vvar o}
 \global\long\def\vp{\vvar p}
 \global\long\def\vq{\vvar q}
 \global\long\def\vr{\vvar r}
 \global\long\def\vs{\vvar s}
 \global\long\def\vu{\vvar u}
 \global\long\def\vv{\vvar v}
 \global\long\def\vw{\vvar w}
 \global\long\def\vx{\vvar x}
 \global\long\def\vy{\vvar y}
 \global\long\def\vz{\vvar z}

\global\long\def\vpi{\vvar{\pi}}
\global\long\def\vxi{\vvar{\xi}}
\global\long\def\vchi{\vvar{\chi}}
 \global\long\def\valpha{\vvar{\alpha}}
 \global\long\def\veta{\vvar{\eta}}
 \global\long\def\vlambda{\vvar{\lambda}}
 \global\long\def\vmu{\vvar{\mu}}
\global\long\def\vdelta{\vvar{\Delta}}
 \global\long\def\vsigma{\vvar{\sigma}}
 \global\long\def\vzero{\vvar 0}
 \global\long\def\vones{\vvar 1}

\global\long\def\xopt{\vvar x^{*}}

\global\long\def\ma{\mvar A}
 \global\long\def\mb{\mvar B}
 \global\long\def\mc{\mvar C}
 \global\long\def\md{\mvar D}
\global\long\def\mE{\mvar E}
 \global\long\def\mf{\mvar F}
 \global\long\def\mg{\mvar G}
 \global\long\def\mh{\mvar H}
\global\long\def\mI{\mvar I}
 \global\long\def\mm{\mvar M}
 \global\long\def\mn{\mathbf{N}}
\global\long\def\mq{\mvar Q}
 \global\long\def\mr{\mvar R}
 \global\long\def\ms{\mvar S}
 \global\long\def\mt{\mvar T}
 \global\long\def\mU{\mvar U}
 \global\long\def\mv{\mvar V}
 \global\long\def\mw{\mvar W}
 \global\long\def\mx{\mvar X}
 \global\long\def\my{\mvar Y}
\global\long\def\mz{\mvar Z}
 \global\long\def\mproj{\mvar P}
 \global\long\def\mSigma{\mvar{\Sigma}}
 \global\long\def\mLambda{\mvar{\Lambda}}
 \global\long\def\mha{\hat{\mvar A}}
 \global\long\def\mzero{\mvar 0}
\global\long\def\mlap{\mvar{\mathcal{L}}}
\global\long\def\mpi{\mvar{\Pi}}

\global\long\def\mdiag{\mvar{diag}}
\global\long\def\diag{\mathrm{diag}}
\global\long\def\sspan{\mathrm{span}}

\global\long\def\tsolve{\mvar{\mathcal{T}_{\textnormal{solve}}}}

\global\long\def\weightVec{\vvar w}
  \global\long\def\oracle{\mathcal{O}}
 \global\long\def\moracle{\mvar O}
 \global\long\def\oracleOf#1{\oracle\left(#1\right)}
 \global\long\def\nSamples{s}
 \global\long\def\simplex{\Delta}

\global\long\def\abs#1{\left|#1\right|}

\global\long\def\capacityMatrix{\mvar U}

\global\long\def\cost{\mathrm{cost}}
\global\long\def\tr{\mathrm{tr}}
\global\long\def\proj{\mathrm{proj}}

\global\long\def\timeNearlyOp{\tilde{\mathcal{O}}}
 \global\long\def\timeNearlyLinear{\timeNearlyOp}

\global\long\def\varFun{f}

\global\long\def\funLip{L}
 \global\long\def\funCon{\mu}

\global\long\def\stepOpt{T}
 \global\long\def\stepOptCoordinate#1{\stepOpt_{(#1)}}

\global\long\def\reff#1{R_{#1}^{\text{eff}}}

\global\long\def\im{\mathrm{im}}

\global\long\def\ceil#1{\left\lceil #1 \right\rceil }

\global\long\def\runtime{\mathcal{T}}
 \global\long\def\timeOf#1{\runtime\left(#1\right)}

\global\long\def\domain{\mathcal{D}}

\global\long\def\argmin{\mathrm{argmin}}
\global\long\def\argmax{\mathrm{argmax}}
\global\long\def\nnz{\mathrm{nnz}}
\global\long\def\vol{\mathrm{vol}}
\global\long\def\supp{\mathrm{supp}}
\global\long\def\dist{\mathcal{D}}


\renewcommand{\dagger}{+}

\global\long\def\mixingtime{t_{mix}}
\global\long\def\lazywalk{\widetilde{\mw}}

\title{Faster Algorithms for Computing the Stationary Distribution, Simulating
Random Walks, and More}

\author{Michael B. Cohen\thanks{This material is based upon work supported by the National Science
Foundation under Grant No. 1111109.}\\
MIT\\
micohen@mit.edu\and Jonathan Kelner\footnotemark[1] \\
MIT\\
kelner@mit.edu\and  John Peebles\thanks{This material is based upon work supported by the National Science
Foundation Graduate Research Fellowship under Grant No. 1122374 and
by the National Science Foundation under Grant No. 1065125.}\\
MIT\\
jpeebles@mit.edu\and  Richard Peng\\
Georgia Tech \thanks{This material is partly based upon work supported by the National
Science Foundation under Grant No. 1637566. Part of this work was
done while at MIT.}\\
rpeng@cc.gatech.edu\and  Aaron Sidford\\
Stanford University\\
sidford@stanford.edu\and  Adrian Vladu\footnotemark[1] \\
MIT\\
avladu@mit.edu}

\date{}
\maketitle
\begin{abstract}
In this paper, we provide faster algorithms for computing various
fundamental quantities associated with random walks on a directed
graph, including the stationary distribution, personalized PageRank
vectors, hitting times, and escape probabilities. In particular, on
a directed graph with $n$ vertices and $m$ edges, we show how to
compute each quantity in time $\tilde{O}(m^{3/4}n+mn^{2/3})$, where
the $\tilde{O}$ notation suppresses polylogarithmic factors in $n$,
the desired accuracy, and the appropriate condition number (i.e. the
mixing time or restart probability). 

Our result improves upon the previous fastest running times for these
problems; previous results either invoke a general purpose linear
system solver on a $n\times n$ matrix with $m$ non-zero entries,
or depend polynomially on the desired error or natural condition number
associated with the problem (i.e. the mixing time or restart probability).
For sparse graphs, we obtain a running time of $\tilde{O}(n^{7/4})$,
breaking the $O(n^{2})$ barrier of the best running time one could
hope to achieve using fast matrix multiplication. 

We achieve our result by providing a similar running time improvement
for solving \emph{directed Laplacian systems}, a natural directed
or asymmetric analog of the well studied symmetric or undirected Laplacian
systems. We show how to solve such systems in time $\tilde{O}(m^{3/4}n+mn^{2/3})$,
and efficiently reduce a broad range of problems to solving $\otilde(1)$
directed Laplacian systems on Eulerian graphs. We hope these results
and our analysis open the door for further study into \emph{directed
spectral graph theory}.
\end{abstract}

\vfill

\pagebreak{}

\section{Introduction}

The application and development of spectral graph theory has been
one of the great algorithmic success stories of the past three decades.
By exploiting the relationship between the combinatorial properties
of a graph, the linear algebraic properties of its Laplacian, and
the probabilistic behavior of the random walks they induce, researchers
have obtained landmark results ranging across multiple areas in the
theory of algorithms, including Markov chain Monte Carlo techniques
for counting \cite{DBLP:journals/iandc/SinclairJ89,karp1989monte,jerrum2004polynomial}
and volume estimation \cite{lovasz1990mixing,dyer1991random,vempala2005geometric,lovasz2006fast,lovasz2006hit,lee2016geodesic},
approximation algorithms for clustering and graph partitioning problems
\cite{alon1985lambda1,SpielmanTengSolver:journal,kannan2004clusterings,AndersenCL06,orecchia2012approximating},
derandomization \cite{hoory2006expander,reingold2008undirected},
error correcting codes \cite{spielman1995linear,sipser1996expander},
and the analysis of random processes \cite{lovasz1993random}, among
others. In addition to their theoretical impact, spectral techniques
have found broad applications in practice, forming the core of Google's
PageRank algorithm, playing a ubiquitous role in practical codes for
machine learning, computer vision, clustering, and graph visualization.
Furthermore they have enabled the computation of fundamental properties
of various Markov chains, such as stationary distributions, escape
probabilities, hitting times, mixing times, and commute times.

More recently, spectral graph theory has been driving an emerging
confluence of algorithmic graph theory, numerical scientific computing,
and convex optimization. This recent line of work began with a sequence
of papers that used combinatorial techniques to accelerate the solution
of linear systems in undirected graph Laplacians, eventually leading
to algorithms that solve these systems in nearly-linear time \cite{SpielmanTengSolver:journal,KoutisMP10:journal,KoutisMP10,KoutisMP11,KelnerOSZ13,lee2013efficient,Cohen:2014:SSL:2591796.2591833,PengS14,LeePS15,KyngLPSS15:arxiv}.
This was followed by an array of papers in the so-called ``Laplacian
Paradigm'' \cite{Teng10}, which either used this nearly-linear-time
algorithm as a primitive or built on the structural properties underlying
it to obtain faster algorithms for problems at the core of algorithmic
graph theory, including finding maximum flows and minimum cuts \cite{christiano2011electrical,lee2013new,sherman2013nearly,kelner2014almost,peng2016approximate},
solving traveling salesman problems \cite{asadpour2010log,anari2015effective},
sampling random trees \cite{kelner2009faster,madry2015fast}, sparsifying
graphs \cite{SpielmanS08,SpielmanS11,allen2015spectral,LeeS15}, computing
multicommodity flows \cite{kelner2012faster,kelner2014almost}, and
approximately solving a wide range of general clustering and partitioning
problems \cite{alon1985lambda1,kannan2004clusterings,AndersenCL06,orecchia2012approximating}.

While these recent algorithmic approaches have been very successful
at obtaining algorithms running in close to linear time for undirected
graphs, the directed case has conspicuously lagged its undirected
counterpart. With a small number of exceptions involving graphs with
particularly nice properties and a line of research in using Laplacian
system solvers inside interior point methods for linear programming
\cite{DBLP:conf/focs/Madry13,LeeS14,DBLP:journals/corr/CohenMSV16},
the results in this line of research have centered almost entirely
on the spectral theory of undirected graphs. While there have been
 interesting results in candidate directed spectral graph theory \cite{ref1,AndersenCL06,GuoM15:arxiv},
their algorithmic ramifications have been less clear.

One problem that particularly well illustrates the discrepancy between
the directed and undirected settings is the computation of the stationary
distribution of a random walk. Computing this is a primary goal in
the analysis of Markov chains, constitutes the main step in the PageRank
algorithm, remains the missing piece in derandomizing randomized log
space computations \cite{4262767}, and is necessary to obtain the
appropriate normalization for any of the theoretical or algorithmic
results in one of the few instantiations of directed spectral graph
theory \cite{ref1,AndersenCL06}. 

In the undirected setting, the stationary distribution is proportional
to the degree of a vertex, so it can be computed trivially. However,
despite extensive study in the mathematics, computer science, operations
research, and numerical scientific computing communities, the best
previously known asymptotic guarantees for this problem are essentially
what one gets by applying general-purpose linear algebra routines.
Given a directed graph with $n$ vertices and $m$ edges these previous
algorithms fall into two broad classes: 
\begin{itemize}
\item \textbf{Iterative Methods:} These aim to compute the stationary distribution
by either simulating the random walk directly or casting it as a linear
system or eigenvector computation and applying either a global or
coordinate-wise iterative method to find it. The running times of
these methods either depend polynomially on the relevant numerical
conditioning property of the instance, which in this case is, up to
polynomial factors, the mixing time of the random process; or they
only compute a distribution that only approximately satisfies the
defining equations of the stationary distribution, with a running
time that is polynomial in $1/\epsilon$. There has been extensive
work on tuning and specializing these methods to efficiently compute
the stationary distribution, particularly in the special case of PageRank.
However, all such methods that we are aware of retain a polynomial
dependence on either the mixing time, which can be arbitrary large
as a function of the number of edges of the graph, or on $1/\epsilon$.\footnote{A possibleexception, is the algorithm that invokes conjugate gradient
in a blackbox manner to solve the requisite linear system to compute
the stationary distribution. At best this analysis would suggest an
$O(mn)$ running time. However, it is not known how to realize even
this running time in the standard word-RAM model of computation.}
\item \textbf{Fast Matrix Multiplication}: By using a direct method based
on fast matrix multiplication, one can find the stationary distribution
in time in time $n^{\omega}$, where $\omega<2.3729$ \cite{DBLP:conf/stoc/Williams12}
is the matrix multiplication exponent. These methods neglect the graph
structure and cannot exploit sparsity. As such, even if one found
a matrix multiplication algorithm matching the lower bound of $\omega=2$,
this cannot give a running time lower than $\Omega(n^{2})$, even
when the graph is sparse.
\end{itemize}
Another problem which well demonstrates the gap between directed and
undirected graph problems is that of solving linear systems involving
graph Laplacians. For undirected graphs, as we have discussed there
are multiple algorithms to solve associated Laplacian systems in nearly
time. However, in the case of directed graphs natural extensions of
solving Laplacian systems are closely related to computing the stationary
distribution, and thus all known algorithms either depend polynomially
on the condition number of the matrix or the desired accuracy or they
require time $\Omega(n^{2})$. Moreover, many of the techniques, constructions,
and properties used to solve undirected Laplacian systems either have
no known analogues for directed graphs or can be explicitly shown
to not exist. This gap in our ability to solve Laplacian systems is
one of the the primary reasons (perhaps \emph{the} primary reason)
that the recent wave of graph algorithms based on the ``Laplacian
Paradigm'' have not produced directed results to match the undirected
ones.

Given the fact that, despite several decades of work on designing
specialized methods for this problem, there are no methods known that
asymptotically improve upon general linear algebra routines, along
with the structural problems in translating the techniques from the
undirected case, it would not be unreasonable to expect that the best
one can hope for is heuristic improvements in special cases, and that
the worst-case asymptotics for graph Laplacians are no better than
the $\min{O(n^{\omega},nm)}\geq\Omega(n^{2})$ that is known for general
matrices.

In this paper, we show that this is not the case by providing an algorithm
that solves directed graph Laplacian systems\textemdash a natural
generalization of undirected graph Laplacian systems\textemdash in
time $\tilde{O}(nm^{3/4}+n^{2/3}m)$ where here and throughout the
paper the $\otilde(\cdot)$ notation hides polylogarithmic factors
in $n$, the desired accuracy, and the natural condition numbers associated
with the problem. Consequently, we obtain the first asymptotic improvement
for these systems over solving general linear systems.\footnote{In follow up work, the authors of this paper in collaboration with
Anup Rao have improved the running time to almost linear in the number
of edges in the graph, meaning the running time is linear if we ignore
contributions to the running time that are smaller than any polynomial.
This paper will be made available online as soon as possible.} In particular, when the graph is sparse, i.e. $m=O(n)$, our algorithm
runs in time $\otilde(n^{7/4})$, breaking the barrier of $O(n^{2})$
that would be achieved by algorithms based on fast matrix multiplication
if $\omega=2$. We then leverage this result to obtain improved running
times for a host of problems in algorithmic graph theory, scientific
computing, and numerical linear algebra, including:
\begin{itemize}
\item \textbf{Computing the Stationary Distribution}: We compute a vector
within $\ell_{2}$ distance $\epsilon$ of the stationary distribution
of a random walk on a strongly connected directed graph in time $\otilde(nm^{3/4}+n^{2/3}m)$,
where the natural condition number of this problem is the mixing time.
(See Section~\ref{subsec:app-cond-est}.)
\item \textbf{Solving Large Classes of Linear Systems}: We provide algorithms
that solve a large class of well-studied linear systems. Compared
with prior algorithms capable of solving this class, ours are the
first that are asymptotically faster than solving general linear systems,
and the first that break the $O(n^{2})$ barrier for sufficiently
sparse instances. Our methods solve directed Laplacian systems and
systems where the matrix is row- or column-diagonally dominant. The
running time is $\tilde{O}\left(nm^{3/4}+n^{2/3}m\right)$. (See Section~\ref{subsec:app-solvers}.)
\item \textbf{Computing Personalized PageRank}: We compute a vector within
$\ell_{2}$ distance $\epsilon$ of the personalized PageRank vector,
for a directed graph with with restart probability $\beta$, in time
$\tilde{O}\left(nm^{3/4}+n^{2/3}m\right)$. Here the natural condition
number is $1/\beta$. In the case of small $\beta$ and $\epsilon$,
this improves upon local methods that take $O(m\beta^{-1}\epsilon^{-1})$
time \cite{PageBMW99,jeh2003scaling,fogaras2004towards,AndersenCL06,andersen2007local,lofgren2014fast}.
(See Section~\ref{subsec:app-pagerank}).
\item \textbf{Simulating Random Walks}: We show how to compute a wide range
of properties of random walks on directed graphs including escape
probabilities, commute times, and hitting times. (See Section~\ref{subsec:app-commute}
and Section~\ref{subsec:Escape-probabilities}.) We also show how
to efficiently estimate the mixing time of a lazy random walk on a
directed graph up to polynomial factors in $n$ and the mixing time.
(See Section~\ref{subsec:app-cond-est}.) The runtime for all these
algorithms is $\tilde{O}\left(nm^{3/4}+n^{2/3}m\right)$.
\item \textbf{Estimating All Commute Times}: We show how to build a $\otilde(n\epsilon^{-2}\log n)$
size data structure in time $\otilde(nm^{3/4}+n^{2/3}m)$ that, when
queried with any two vertices $a$ and $b$, outputs a $1\pm\epsilon$
multiplicative approximation to the expected commute time between
$a$ and $b$, i.e. the expected amount of time for a random walk
starting at $a$ to reach $b$ and return to $a$. Our data structure
is similar to the data structure known for computing all-pairs effective
resistances in undirected graphs \cite{SpielmanS08,SpielmanS11}.
(See Section~\ref{subsec:Sketching-Commute-Times}.)
\end{itemize}

It is important to note that the $\tilde{O}$-notation hides factors
that are polylogarithmic in both the condition number (equivalently,
mixing time) and the ratio of maximum to minimum stationary probability.
As such, the natural parameter regime for our algorithms is when these
quantities are subexponential or polynomial. For all the above problems,
the best prior algorithms had worst case runtimes no better than $O(\min\{n^{\omega},nm\})\geq\Omega(n^{2})$
in this regime. We hope that our results open the door for further
research into directed spectral graph theory, and serve as foundation
for the development of faster algorithms for directed graphs.

\subsection{Approach}

Our approach for solving these problems centers around solving linear
systems in a class of matrices we refer to as \emph{directed (graph)
Laplacians}, a natural generalization of undirected graph Laplacians.
A directed Laplacian, $\mlap\in\R^{n\times n},$ is simply a matrix
with non-positive off-diagonal entries such that each diagonal entry
is equal to the sum of the absolute value of the other off-diagonal
entries in that column, i.e. $\mlap_{ij}\leq0$ for $i\neq j$ and
$\mlap_{ii}=-\sum_{j\neq i}\mlap_{ji}$ (equivalently $\vec{1}^{\top}\mlap=\vec{0}$).
As with undirected Laplacians, every directed Laplacian there is naturally
associated with a directed graph $G=(V,E,w)$, where the vertices
$V$ correspond to the columns of $\mlap$ and there is an edge from
vertex $i$ to vertex $j$ of weight $\alpha$ if and only if $\mlap_{ji}=-\alpha$.

Another close feature of directed and undirected Laplacians is the
close connection between random walks on the associated graph $G$
and solutions to linear systems in $\mlap$. We ultimately show that
solving a small number of directed Laplacian systems suffices to obtain
all of our desired applications (See Section~\ref{sec:solving_strictly_rcdd}
and Section~\ref{sec:applications}). Unfortunately, solving linear
systems in $\mlap$ directly is quite challenging. Particularly troubling
is the fact that we while we know $\mlap$ has a non-trivial kernel
(since $\vec{1}^{\top}\mlap=\vec{0}^{\top})$, we do not have a simple
method to compute it efficiently. Moreover, $\mlap$ is not symmetric,
complicating the analysis of standard iterative algorithms.Furthermore,
the standard approach of multiplying on the left by the transpose,
so that we are solving linear systems in $\mlap^{\top}\mlap$, would
destroy the combinatorial structure of the problem and cause an intolerably
large condition number. A natural idea is to try to work with a symmetrization
of this matrix, $\frac{1}{2}(\mlap+\mlap^{\top})$, but it turns out
that this may not even be positive semidefinite (PSD).\footnote{Consider the directed edge Laplacian $\mlap=\left[\begin{array}{cc}
1 & 0\\
-1 & 0
\end{array}\right]$. Then, $\mlap+\mlap^{\dagger}=\left[\begin{array}{cc}
2 & -1\\
-1 & 0
\end{array}\right]$ has an eigenvector $(\sqrt{2}-1,1)$ with a corresponding eigenvalue
of $(1-\sqrt{2})$.} Consequently, it is not clear a priori how to define an efficient
iterative method for computing the stationary $\mlap$ or solve systems
in it without depending polynomially on the condition number of $\mlap$.

Fortunately, we do know how to characterize the kernel of $\mlap$,
even if computing it is difficult \emph{a priori}. If we let $\md\in\R^{n\times n}$
denote the diagonal matrix consistent with $\mlap$, i.e., $\md_{ii}=\mlap_{ii}$,
then we see that $\mlap\md^{-1}=\mI-\mw$ where $\mI$ is the identity
matrix and $\mw$ is the random walk matrix associated with $G$.
In other words, for any distribution $p$, we have that $\mw p$ is
the resulting distribution of one step of the random walk where, at
a vertex $i\in[n]$, we pick a random outgoing edge with probability
proportional to its weight and follow that edge. The Perron-Frobenius
Theorem implies that as long as the graph is strongly connected there
is some stationary distribution $s\in\R_{>0}$ such that $\mw s=s$.
Consequently, the kernel of $\mlap$ is simply the stationary distribution
of the natural random walk on $G$ multiplied by $\md$.

Consequently, we can show that for every directed Laplacian $\mlap$
that corresponds to a strongly connected graph, there is always a
vector $x\in\R_{>0}^{n}$ such that $\mlap x=0$ (See Lemma~\ref{lem:stationary-equivalence}).
In other words, letting $\mx$ denote the diagonal matrix associated
with $x$ the directed Laplacian $\mlap'=\mlap\mx$ satisfies $\mlap'\mx1=0$.
This says that the total weight of incoming edges to a vertex is the
same as the total weight of outgoing edges from that vertex, i.e.,
that $\mlap'$ corresponds to the Laplacian of an Eulerian graph.
We call such a vector an\emph{ Eulerian scaling} of $\mlap$.

Now, solving systems in an Eulerian Laplacian $\mlap$ (i.e., a Laplacian
corresponding to an Eulerian graph) seems easier than solving an arbitrary
directed Laplacian. In particular, we know the kernel of a $\mlap$,
since it is just the all ones vector. In addition, we have that $\frac{1}{2}(\mlap+\mlap^{\top})$
is symmetric and PSD\textemdash in fact it is just the Laplacian of
an undirected graph! Unfortunately, this does not immediately yield
an algorithm, as it is not known how to use the ability to solve systems
in such a symmetrization to solve systems in the original matrix.

Ultimately, this line of reasoning leaves us with two fundamental
questions: 
\begin{enumerate}
\item Can we solve Eulerian Laplacian systems in time $o(n^{\omega},nm)?$ 
\item Can we use an Eulerian Laplacian system solver for more than solving
Eulerian Laplacian systems? 
\end{enumerate}
The major contribution of this paper is answering both of these questions
in the affirmative. We show the following: 
\begin{itemize}
\item We show that we can solve Eulerian Laplacian systems in time $\tilde{O}\left(nm^{3/4}+n^{2/3}m\right)$. 
\item We show that using Eulerian Laplacian systems we can solve broader
classes of matrices we refer to as RCDD Z-matrices, and $\alpha$
RCDD Z-matrices. 
\item We show that using solvers for $\alpha$ RCDD Z-matrices, we can estimate
an Eulerian scaling of a directed Laplacian. 
\item Putting these components together we achieve our desired applications.
Some of these are applications are straightforward, whereas others
require some significant work. 
\end{itemize}
A serious question that arises throughout these results is the degree
of precision do we need to carry out our arithmetic operations. This
arrises both in using undirected Laplacian system solvers to solving
Eulerian Laplacian systems, and then again in using Eulerian Laplacian
system solvers to derive the rest of our results. These numerical
issues are not merely technicalities\textemdash they crucially affect
the algorithms we can use to solve our problem. In fact, we will see
in Section~\ref{sec:eulerian_solver} that, if we disregarded the
numerics and relied on frequently-quoted assertions relating the behavior
of conjugate gradient to the existence of polynomials with certain
properties, we would actually obtain a better running time, but that
these assertions do not carry over to reasonable finite-precision
setting. 

Given these subtleties, we discuss numerical issues throughout the
paper, and in Appendix~\ref{sec:numerical_stability} we cover particular
details of the stability of our algorithms and the precision they
require, showing that we can achieve all our results in the standard
unit cost RAM model (or any other reasonable model of computation).

In the remainder of this overview we briefly comment on the key technical
ingredients of each of these results.

\subsubsection{Solving Eulerian Laplacian Systems}

To solve a Eulerian Laplacian system $\mlap x=b$, we first precondition,
multiplying both sides by $\mlap^{\top}\mU^{\dagger}$, where $\mU\defeq\frac{1}{2}(\mlap^{\top}+\mlap)$
is a Laplacian of an undirected graph corresponding to $\mlap$, and
$\mU^{\dagger}$ is its Moore-Penrose pseudoinverse. This shows that
it suffices to instead solve, $\mlap^{\top}\mU^{\dagger}\mlap x=\mlap^{\top}\mU^{\dagger}x$.
Now using a nearly-linear-time Laplacian system solver, we can apply
$\mU^{\dagger}$ to a vector efficiently. As such, we simply need
to show that we can efficiently solve systems in the symmetric matrix
$\mlap^{\top}\mU^{\dagger}\mlap$.

Next, we show that the matrix $\mlap^{\top}\mU^{\dagger}\mlap$ is,
in an appropriate sense, approximable by $\mU$. Formally we show
that $\mU$ is smaller in the sense that $\mU\preceq\mlap^{\top}\mU^{\dagger}\mlap$,
and that it is not too much larger in the sense that $\tr(\mU^{\dagger/2}\mlap^{\top}\mU^{\dagger}\mlap\mU^{\dagger/2})=O(n^{2})$.
While the first proof is holds for a broad class of asymmetric matrices,
to prove the second fact we exploit structure of Eulerian Laplacians,
particularly the fact that an Eulerian graph has a decomposition into
simple cycles.

Unfortunately, this property doesn't immediately yield an algorithm
for solving Laplacian systems. The natural approach would be to use
preconditioned Krylov methods, such as the Chebyshev method or conjugate
gradient. These essentially apply a polynomial of $\mU^{\dagger}\mlap^{\top}\mU^{\dagger}\mlap$
to the right hand side. Unfortunately, Chebyshev iterations only yield
a $\Omega(mn)$ time algorithm with this approach. For conjugate gradient,
it can be shown that the trace bound leads to $o(mn)$ time algorithm
in exact arithmetic, but, unfortunately, this analysis does not appear
to be numerically stable, and we do not know how to show it yields
this running time in our computational model.

Instead we implement an approach based on preconditioning and subsampling.
We precondition with $\mlap^{\top}\mU^{\dagger}\mlap+\alpha\mU$ for
a value of $\alpha$ we tune. This reduces the problem to only solving
$\otilde(\sqrt{\alpha})$ linear systems in $\mlap^{\top}\mU^{\dagger}\mlap+\alpha\mU$.
To solve these systems we note that we can write this equivalently
as $\mlap^{\top}\mU^{\dagger}\mU\mU^{\dagger}\mlap+\alpha\mU$ and
using the factorization of $\mU$ into its edges we can subsample
the inner $\mU$ while preserving the matrix. Ultimately, this means
we only need to solve systems in $\alpha\mU$ plus a low rank matrix
which we can do efficiently using the fact that there is an explicitly
formula for low rank updates (i.e. Sherman-Morrison-Woodbury Matrix
Identity). Trading between the number of such systems to solve, the
preprocessing to solve these systems, and the time to solve them gives
us our desired running time for solving such linear systems. We show
in the appendix that we can, with some care, stably implement a preconditioned
Chebyshev method and low rank update formulas. This allows us to circumvent
the issues in using conjugate gradient and achieve our running time
in the desired computational model.

\subsubsection{Solving RCDD Z-matrices}

A row column diagonal dominant (RCDD) matrix is simply a matrix $\mm$
where $\mm_{ii}\geq\sum_{j\neq i}|\mm_{ij}|$ and $\mm_{ii}\geq\sum_{j\neq i}|\mm_{ji}|$
and a Z-matrix is a matrix $\mm$ where the off-diagonal entries are
negative. We show how to solve such matrices by directly reducing
them to solving Eulerian Laplacian systems. Given a RCDD Z-matrix
$\mm$, we add an additional row and column, filling in the entries
in the natural way so that the resulting matrix is an Eulerian Laplacian.
We show that, from the solution to such a linear system, we can immediately
glean the solution to systems in $\mm$. This reduction is analogous
to the reduction from solving symmetric diagonally dominant (SDD)
systems to solving undirected Laplacian systems. In the appendix we
show that this method is stable.

\subsubsection{Computing the Stationary}

Given a RCDD Z-matrix solver we use it to compute the scaling that
makes a directed Laplacian $\mlap$ Eulerian, i.e., we compute the
stationary distribution. To do this, we pick an initial non-negative
diagonal scaling $\mx$ and a initial non-negative diagonal matrix
$\mE$ such that $(\mE+\mlap)\mx$ is $\alpha$-RCDD, that is each
diagonal entry is a $1+\alpha$ larger in absolute value than the
sum of the off-diagonal entries in both the corresponding row and
column.

We then iteratively decrease $\mE$ and update $\mx$ while maintaining
that the matrix is $\alpha$-RCDD. The key to this reduction is the
fact that there is a natural way to perform a rank 1 update of $(\mE+\mlap)\mx$
to obtain an Eulerian Laplacian, and that the stationary distribution
of this Laplacian can be obtained by solving a single system in $(\mE+\mlap)\mx$.
Ultimately, this method yields a sequence of stationary distributions
that, when multiplied together entrywise, yield a good approximate
stationary distribution for $\mlap$. For a more detailed over this
approach and this intuition underlying it, see Section~\ref{sec:computing_stationary:approach}.

\subsubsection{Additional Applications}

Our algorithms for computing personalized page rank vectors, solving
linear systems in arbitrary RCDD matrices, and solving directed Laplacian
linear systems are all proven in a similar fashion. We obtain an approximate
stationary distribution, rescale the system to make it strictly RCDD,
then solve it\textemdash all using algorithms from the previous sections
in a black box fashion. Therefore, the running times for these applications\textemdash and
in fact all our applications\textemdash depend solely (up to polylogarithmic
factors) on the black-box costs of computing the stationary distribution
and solving RCDD matrices.

However, our algorithms must determine how much accuracy to request
when they invoke these two black-box routines. For computing personalized
PageRank, one can determine the accuracy to request based solely on
the restart probability. However, for our other applications, the
accuracy our algorithms request has a dependence on the condition
number $\kappa(\mlap)$ of $\mlap$ and the ratio $\kappa(\ms^{*})$
of max over min stationary probability. In order to get an unconditional
running time\textemdash and out of intrinsic interest\textendash -we
show how to efficiently compute reasonable upper bounds on these quantities.
We use an approach motivated by the close relationship of $\kappa(\mlap)$
and mixing time. Specifically, we formulate a notion of personalized
PageRank mixing time, then get a polynomially good estimate of this
quantity using our ability to solve personalized PageRank. Finally,
we show that $\kappa(\mlap)$ and personalized pagerank mixing time
are equivalent up to factors that are good enough\footnote{They are equivalent up to factors polynomial in $n$ and themselves.
Since these quantities appear only in logs in our runtimes and these
logs already have factors of $n$ in them, this polynomial equivalence
only changes runtimes by a constant factor compared to if we had exact
estimates. A similar relationship also holds between $\kappa(\mlap)$
and the mixing time of lazy random walks on the graph. (See Section
\ref{sec:application_quantities} for details.)} for our purposes. With a reasonable bound on $\kappa(\mlap)$, we
are then able to choose a restart probability that is small enough
in order to guarantee that personalized solving PageRank gives a good
approximation of the stationary distribution.

Our algorithms for computing hitting times, escape probabilities,
and all pairs commute times all start by taking a natural definition
for the quantity of interest and massaging it into an explicit formula
that has an $\mlap^{\dagger}$ in it. Then, they use various methods
to approximately evaluate the formula. In the case of hitting times,
we simply plug everything into the formula and invoke our approximate
solver for $\mlap^{\dagger}$ with appropriately small error. Escape
probabilities are handled similarly, except that there are also two
unknown parameters which we show we can estimate to reasonable accuracy
and plug in.

Perhaps the most sophisticated application is computing all pairs
commute times. We show that the commute time from $u$ to $v$ is
given by the simple formula $(\vec{1}_{u}-\vec{1}_{v})^{\intercal}(\mlap_{b}^{\top}\mU_{b}^{\dagger}\mlap_{b})^{\dagger}(\vec{1}_{u}-\vec{1}_{v})$
where $\mlap_{b}$ is the matrix obtained by performing the diagonal
rescaling of $\mlap$ that turns its diagonal into the stationary
distribution, which also makes the graph Eulerian. An interesting
feature of this formula is that it involves applying the pseudo-inverse
of a matrix of the very same form as the preconditioned system $\mlap^{\top}\mU^{\dagger}\mlap$
that our Eulerian Laplacian solver uses. Another interesting feature
is that when $\mlap$ is symmetric, this formula simplifies to $(\vec{1}_{u}-\vec{1}_{v})^{\intercal}\mU_{b}^{\dagger}(\vec{1}_{u}-\vec{1}_{v})=2m\cdot(\vec{1}_{u}-\vec{1}_{v})^{\intercal}\mU^{\dagger}(\vec{1}_{u}-\vec{1}_{v})$.
Thus, it is a generalization of the well-known characterization of
commute times in terms of effective resistance from undirected graphs.
In undirected graphs, all pairs commute times can be computed efficiently
via Johnson-Lindenstrauss sketching \cite{SpielmanS08}. We show that
a similar approach extends to directed Laplacians as well. While the
general approach is similar, the error analysis is complicated by
the fact that we only have access to an approximate stationary distribution.
If this were used naively, one would have to deal with an approximate
version of $\mlap_{b}$ that, importantly, is only approximately Eulerian.
We bypass this issue by showing how to construct an Eulerian graph
whose stationary is exactly known and whose commute times approximate
the commute times of the original graph. This may be of independent
interest.

The fact that a matrix of the form $\mlap^{\top}\mU^{\dagger}\mlap$
comes up both in solving Eulerian Laplacians and in sketching commute
times indicates that it is an even more natural object than it might
appear at first.

\subsection{Paper Organization}

The rest of our paper is organized as follows:
\begin{itemize}
\item \textbf{Section~\ref{sec:prelim}} - we cover preliminary information 
\item \textbf{Section~\ref{sec:computing_stationary}} - we show how to
compute the stationary distribution
\item \textbf{Section~\ref{sec:eulerian_solver}} - we provide our fast
Eulerian Laplacian system solver 
\item \textbf{Section~\ref{sec:solving_strictly_rcdd}} - we reduce strict
RCDD linear systems to solving Eulerian systems
\item \textbf{Section}~\ref{sec:application_quantities} - we provide condition
number quantities for applications and prove equivalences
\item \textbf{Section~\ref{sec:applications}} - we provide our applications
\item \textbf{Appendix~\ref{sec:numerical_stability}} - we discuss numerical
stability of our solvers
\item \textbf{Appendix~\ref{sec:matrix_facts}} - we provide facts about
matrices we use throughout
\item \textbf{Appendix~}\ref{sec:Deriving-the-Identities} - we derive
identities for hitting times, commute times, and escape probabilities
\end{itemize}

\section{Preliminaries \label{sec:prelim}}

In this section we introduce notation and provide basic machinery
we use throughout this paper.

\subsection{Notation \label{sec:prelim:notation}}

\textbf{Matrices: }We use bold to denote matrices and let $\mI_{n},\mzero_{n}\in\R^{n\times n}$
denote the identity matrix and zero matrix respectively. For symmetric
matrices $\ma,\mb\in\R^{n\times n}$ we use $\ma\preceq\mb$ to denote
the condition that $x^{\top}\ma x\leq x^{\top}\mb x$ and we define
$\succeq$, $\prec$, and $\succ$ analogously. We call a symmetric
matrix $\ma\in\R^{n\times n}$ positive semidefinite if $\ma\succeq\mzero_{n}$
and we let $\norm x_{\ma}\defeq\sqrt{x^{\top}\ma x}$. For any norm
$\norm{\cdot}$ define on vectors in $\R^{n}$ we define the \emph{operator
norm} it induces on $\R^{n\times n}$ by $\norm{\ma}=\max_{x\neq0}\frac{\norm{\ma x}}{\norm x}$
for all $\ma\in\R^{n\times n}$.\\
\\
\textbf{Diagonals}: For $x\in\R^{n}$ we let $\mdiag(x)\in\R^{n\times n}$
denote the diagonal matrix with $\mdiag(x)_{ii}=x_{i}$ and when it
is clear from context we more simply write $\mx\defeq\mdiag(x)$.
For $\ma\in\R^{n\times n}$ we let $\diag(\ma)$ denote the vector
corresponding to the diagonal of $\ma$ and we let $\mdiag(\ma)\defeq\mdiag(\diag(\ma))$,
i.e. $\ma$ with the off-diagonal set to $0$.\\
\\
\textbf{Vectors}: We let $\vec{0}_{n},\vec{1}_{n}\in\R^{n}$ denote
the all zeros and ones vectors respectively. We use $\vec{1}_{i}\in\R^{n}$
to denote the indicator vector for coordinate $i\in[n]$, i.e. $[\vec{1_{i}}]_{j}=0$
for $j\neq i$ and $[\vec{1}_{i}]_{i}=1$. Occasionally we apply scalar
operations to vectors with the interpretation that they should be
applied coordinate-wise, e.g. for $x,y\in\R^{n}$ we let $\max\{x,y\}$
denote the vector $z\in\R^{n}$ with $z_{i}=\max\{x_{i},y_{i}\}$
and we use $x\geq y$ to denote the condition that $x_{i}\geq y_{i}$
for all $i\in[n]$.\\
\\
\textbf{Condition Numbers}: Given a invertible matrix $\ma\in\R^{n\times n}$
we let $\kappa(\ma)\defeq\norm{\ma}_{2}\cdot\norm{\ma^{\dagger}}_{2}$
denote the condition number of $\ma$. Note that if $\mx\in\R^{n\times n}$
is a nonzero diagonal matrix then $\kappa(\mx)=\frac{\max_{i\in[n]}|\mx_{ii}|}{\min_{i\in[n]}|\mx_{ii}|}$.\\
\\
\textbf{Sets}: We let $[n]\defeq\{1,...,n\}$ and $\simplex^{n}\defeq\{x\in\R_{\geq0}^{n}\,|\,\vec{1}_{n}^{\top}x=1\}$,
i.e. the $n$-dimensional simplex.

\subsection{Matrix Classes \label{sec:prelim:matrices}}

\textbf{Diagonal Dominance: }We call a possibly asymmetric matrix
$\ma\in\R^{n\times n}$ \emph{$\alpha$-row diagonally dominant (RDD)
}if $\ma_{ii}\geq(1+\alpha)\sum_{j\neq i}|\ma_{ij}|$ for all $i\in[n]$,
\emph{$\alpha$-column diagonally dominant (CDD) }if $\ma_{ii}\geq(1+\alpha)\sum_{j\neq i}\left|\ma_{ji}\right|$,
and \emph{$\alpha$-RCDD} if it is both $\alpha$-RDD and $\alpha$-CDD.
For brevity, we call $\ma$ RCDD if it is $0$-RCDD and strictly RCDD
if it is $\alpha$-RCDD for $\alpha>0$.\\
\\
\textbf{Z-matrix}: A matrix $\mm\in\R^{n\times n}$ is called a Z-matrix
if $\mm_{ij}\leq0$ for all $i\neq j$, i.e. every off-diagonal entry
is non-positive.\\
\\
\textbf{Directed Laplacian}: A matrix $\mlap\in\R^{n\times n}$ is
called a \emph{directed Laplacian }if it a Z-matrix with $\vec{1}_{n}^{\top}\mlap=\vec{0}_{n}$,
that is $\mlap_{ij}\leq0$ for all $i\neq j$ and $\mlap_{ii}=-\sum_{j\neq i}\mlap_{ji}$
for all $i$. To every directed Laplacian $\mlap$ we associate a
graph $G_{\mlap}=(V,E,w)$ with vertices $V=[n]$ and an edge $(i,j)$
of weight $w_{ij}=-\mlap_{ji}$ for all $i\neq j\in[n]$ with $\mlap_{ji}\neq0$.
Occasionally we write $\mlap=\md-\ma^{\top}$ to denote that we decompose
$\mlap$ into the diagonal matrix $\md$ where $\md_{ii}=\mlap_{ii}$
is the out degree of vertex $i$ in $G_{\mlap}$ and $\ma$ is weighted
adjacency matrix of $G_{\mlap}$ with $\ma_{ij}=w_{ij}$ if $(i,j)\in E$
and $\ma_{ij}=0$ otherwise. We call $\mw=\ma^{\top}\md^{-1}$ the
random walk matrix associated with $G_{\mlap}$. We call $\mlap$
Eulerian if additionally $\mlap\vec{1}_{n}=\vec{0}_{n}$ as in this
case the \emph{associated graph }$G_{\mlap}$ is Eulerian.\\
\\
\textbf{(Symmetric) Laplacian}: A matrix $\mU\in\R^{n\times n}$ is
called a \emph{Symmetric Laplacian }or just a \emph{Laplacian} if
it is symmetric and a Laplacian. This coincides with the standard
definition of Laplacian and in this case note that the associated
graph $G_{\mU}=(V,E,w)$ is symmetric. For a Laplacian we also associate
a matrix $\mb\in\R^{E\times V}$ known as the weighted incidence matrix.
Each row $b^{(i)}$ of $\mb$ corresponds to an edge $\{j,k\}\in E$
and for a canonical orientation ordering of $\{j,k\}$ we have $b_{j}^{(i)}=\sqrt{w_{\{j,k\}}}$,
$b_{k}^{(i)}=-\sqrt{w_{\{j,k\}}}$, and $b_{l}^{(i)}=0$ if $l\notin\{j,k\}$.
Note that $\mU=\mb^{\top}\mb$ and thus $\mlap$ is always PSD.\\
\\
\textbf{Random Walk Matrix}: A matrix $\mw\in\R^{n\times n}$ is called
a \emph{random walk matrix }if $\mw_{ij}\geq0$ for all $i,j\in[n]$
and $\vec{1}_{n}^{\top}\mw=\vec{1}_{n}$. To every random walk matrix
$\mw$ we associated a directed graph $G_{\mw}=(V,E,w)$ with vertices
$V=[n]$ and an edge from $i$ to $j$ of weight $w_{ij}=\mw_{ij}$
for all $i,j\in[n]$ with $\mw_{ij}\neq0$. Note if we say that $\mlap=\mI-\mw$
is a directed Laplacian, then $\mw$ is a random walk matrix and the
directed graphs associated with $\mlap$ and $\mw$ are identical.\\
\\
\textbf{Lazy Random Walk Matrix: }Given a random walk matrix $\mw\in\R^{n\times n}$
the $\alpha$-lazy random walk matrix associated with $\mw$ for $\alpha\in[0,1]$
is given by $\alpha\mI+(1-\alpha)\mw$. When $\alpha=\frac{1}{2}$
we call this a lazy random walk matrix for short and typically denote
it $\lazywalk$.\\
\\
\textbf{Personalized PageRank Matrix}: Given a random walk matrix
$\mw\in\R^{n\times n}$ the personalized PageRank matrix with restart
probability $\beta\in[0,1]$ is given by $\mm_{pp(\beta)}=\beta(\mI-(1-\beta)\mw)^{-1}$.
Given any probability vector $p\in\simplex^{n}$ the personalized
PageRank vector with restart probability $\beta$ and vector $p$
is the vector $x$ which satisfies $\beta p+(1-\beta)\mw x=x$. Rearranging
terms we see that $x=\mm_{\beta}p$ hence justifying our naming of
$\mm_{pp(\beta)}$ as the personalized PageRank matrix. 

\subsection{Directed Laplacians of Strongly Connected Graphs \label{sec:prelim:stationary}}

Here we give properties of a directed Laplacian, $\mlap\in\R^{n\times n}$,
whose associated graph, $G_{\mlap}$, is strongly connected that we
use throughout the paper. Formally, we provide Lemma~\ref{lem:stationary-equivalence}
which shows that $\mlap$ always has a positive Eulerian scaling,
that is a $x\in\R_{>0}$ with $\mlap\mx$ Eulerian, and that this
is given by the stationary distribution of the associated random walk
matrix. Using this we completely characterize the kernel of $\mlap$
and $\mlap^{\top}$.
\begin{lem}
\label{lem:stationary-equivalence} For directed Laplacian $\mlap=\md-\ma^{\top}\in\R^{n\times n}$
whose associated graph is strongly connected there exists a positive
vector $s\in\R_{>0}^{n}$ (unique up to scaling) such that the following
equivalent conditions hold.

\begin{itemize}
\item $\mw s=s$ for the random walk matrix $\mw=\ma^{\top}\md^{-1}$ associated
with $\mlap$.
\item $\mlap\md^{-1}s=0$
\item $\mlap\md^{-1}\ms$ for $\ms=\mdiag(s)$ is an Eulerian Laplacian.
\end{itemize}
If we scale $s$ so that $\norm s_{1}=1$ then we call $s$ the stationary
distribution associated with the random walk on the associated graph
$G_{\mlap}$. We call any vector $x\in\R_{>0}^{n}$ such that $\mlap\mx$
is an Eulerian Laplacian an eulerian scaling for $\mlap$. Furthermore,
$\ker(\mlap)=\mathrm{span}(\md^{-1}s)$ and $\ker(\mlap^{\top})=\mathrm{span}(\vec{1}_{n})$. 

\end{lem}
\begin{proof}
Since the graph associated with $\mlap$ is strongly connected we
have that for any $\alpha>0$ note that $\alpha\mI+\mw$ is positive
irreducible and aperiodic with all columns having $\ellOne$ norm
of $1+\alpha$. Therefore, by the Perron-Frobenius Theorem we know
and therefore therefore there exists a positive vector $s\in\R_{>0}^{n}$
unique up to scaling such that that $(\alpha\mI+\mw)s=(\alpha+1)s$
or equivalently $\mw s=s$. Therefore, by definition $\ma^{\top}\md^{-1}s=s$
and $\mlap\md^{-1}s=(\md-\ma^{\top})\md^{-1}s=0$. Furthermore this
implies that for $\mlap'$=$\mlap$$\md^{-1}\ms$ we have $\mlap'\vec{1}_{n}=\vec{0}_{n}$
and as $\mlap$ is Z-matrix so is $\mlap'$ and therefore $\mlap'$
is a Z-matrix. Furthermore clearly $\mlap^{\top}\md^{-1}\ms$ being
a Eulerian Laplacian implies $\mw s=s$ and therefore we see that
the conditions are all equivalent. Lastly, we know that $\frac{1}{2}(\mlap'+[\mlap']^{\top})$
is a symmetric Laplacian associated with an connected graph. Therefore
it is PSD having only $\vec{1}_{n}$ in there kernel yielding our
characterization of the kernel of $\mlap$ and $\mlap^{\top}$.
\end{proof}
Note that Lemma~\ref{lem:stationary-equivalence} immediately implies
that any random walk matrix $\mw$ associated with a strongly connected
graph has a unique stationary distribution $s$, i.e. $s\in\simplex^{n}$
and $\mw s=s$. Furthermore, we see that all $\alpha$-lazy random
walks for $\alpha\in[0,1)$ associated with $\mw$ have the same stationary
distribution. 

\section{Computing the Stationary Distribution \label{sec:computing_stationary}}

Here we show to compute the stationary distribution given an $\alpha$-RCDD
Z-matrix linear system solver. Throughout this section, we let $\mlap=\md-\ma^{\top}\in\R^{n\times n}$
denote a directed Laplacian and our primary goal in this section is
to compute an approximate stationary vector $s\in\R_{>0}^{V}$ such
that $\mlap\md^{-1}\ms$ is approximately Eulerian. The main result
of this section is the following:
\begin{thm}[Stationary Computation Using a RCDD Solver]
\label{thm:stat-from-dd}  Given $\alpha\in(0,\frac{1}{2})$ and
$\mlap\in\R^{n\times n}$, a directed Laplacian with $m$ nonzero-entries,
in time $O((m+\runtime)\cdot\log\alpha^{-1})$ Algorithm~\ref{alg:stationary_computation}
computes an approximate stationary distribution $s\in\simplex^{n}$
such that $(3\alpha n\cdot\md+\mlap)\md^{-1}\ms$ is $\alpha$-RCDD
where $\md=\mdiag(\mlap)$, $\ms=\mdiag(s)$, and $\runtime$ is the
cost of computing an $\epsilon$-approximate solution to a $n\times n$
$\alpha$-RCDD Z-matrix linear system with $m$-nonzero entries, i.e.
computing $x$ such that $\norm{x-\mm^{-1}b}_{\mdiag(\mm)}\leq\frac{\epsilon}{\alpha}\norm b_{\mdiag(\mm)^{-1}}$
for $\alpha$-RCDD Z-matrix $\mm\in\R^{n\times n}$ with $O(m)$ non-zero
entries, $\epsilon=O(\poly(n/\alpha))$. Furthermore, $\kappa(\ms)\leq\frac{20}{\alpha^{2}}n$,
where $\kappa(\ms)$ is the ratio between the largest and smallest
elements of $s$, and the diagonals of all intermediate RCDD systems
also have condition number \textup{$O(\poly(n/\alpha))$}.
\end{thm}
The remainder of this section is divided as follows: 
\begin{itemize}
\item \textbf{Section~\ref{sec:computing_stationary:approach}}: we provide
an overview of our approach to computing the stationary distribution
\item \textbf{Section~\ref{sec:computing_stationary:algorithm}}: we provide
the details of our algorithm
\item \textbf{Section~\ref{sec:computing_stationary:analysis}}:\textbf{
}we analyze this algorithm proving Theorem~\ref{thm:stat-from-dd}. 
\end{itemize}

\subsection{The Approach\label{sec:computing_stationary:approach}}

Our approach to computing the stationary distribution is broadly related
to the work of Daitch and Spielman \cite{DaitchS08} for computing
the diagonal scaling makes a symmetric $M$-matrix diagonally dominant.
However, the final rescaling that we're trying to calculate is given
by the uniqueness of the stationary distribution, which in turn follows
from the Perron-Frobenius theorem. The fact that we're rescaling and
computing on the same objective leads a more direct iterative process.

As in \cite{DaitchS08} we use an iterative that brings a matrix increasingly
close to being RCDD. However, Our convergence process is through potential
functions instead of through combinatorial entry count; instead of
making parts of $\mlap$ RCDD incrementally we instead start with
a relaxation of $\mlap$ that is RCDD and iteratively bring this matrix
closer to $\mlap$ while maintaining that it is RCDD. We remark that
this scheme can also be adapted to find rescalings of symmetric M-matrices.

Our algorithm hinges on two key insights. The first is that if we
have positive vectors $e,x\in\R_{>0}^{n}$ so that $\mm\defeq(\mE+\mlap)\mx$
is $\alpha$-RCDD, then for any vector $g\in\R_{>0}^{n}$ with $\norm g_{1}=1$
we can compute the stationary distribution of the directed Laplacian
$\mlap'=\mE-ge^{\top}+\mlap$ by solving a single linear system in
$\mm$. Note that $\mm$ is clearly invertible (Lemma~\ref{lem:invertible_B})
whereas $\mlap'$ has a kernel (Lemma~\ref{lem:stationary-equivalence}).
Furthermore, $\mlap'\mx$ is simply a rank-1 update of $\mm$ and
consequently there kernel $\mlap'$ has a closed form, namely $\mx\mlap^{-1}f$
(Lemma~\ref{lem:ker_under_update}). Since the stationary distribution
of $\mlap'$ is given by its kernel (Lemma~\ref{lem:stationary-equivalence})
we obtain the result.

This first insight implies that if we have a RCDD Z-matrix we can
compute the stationary of a related matrix, however a priori it is
unclear how this allows to compute the stationary of $\mlap$. The
second insight is that if we compute the stationary of $\mlap'$,
e.g. we compute some $y\in\R_{>0}^{n}$ such that $\mlap'\my$ is
an Eulerian Laplacian, then $(\mE+\mlap)\my$ is strictly RCDD. Since
$\mlap'\my$ is Eulerian, $(\mE-\mlap-\mlap')\my=ge^{\top}\my$ is
an all positive matrix which is entrywise less than $\mlap'$ in absolute
value. In other words, removing $ge^{\top}\my$ from $\mlap'\my$
strictly decreases the absolute value of the off diagonal entries
and increases the value of the diagonal entries causing $(\mE+\mlap)\my$
to be strictly RCDD. Consequently, given $y$ we can hope to decrease
$e$ to achieve $e'$ to obtain an $\alpha$-RCDD Z-matrix $\mm'=(\mE'+\mlap)\my$
where $e'\leq e$.

Combining these insights naturally yields our algorithm. We compute
an initial$\alpha$-RCDD Z-matrix $\mm=(\mE+\mlap)\mx$ and compute
$x'=\mx\mm^{-1}g$ for $g\in\R_{>0}^{n}$ with $\norm g_{1}=1$ so
that $\mlap'=(\mE-ge^{\top}+\mlap)\mx'$ is Eulerian. We then compute
the smallest $e'\in\R_{>0}^{n}$ so that $(\mE'+\mlap)\mx'$ is again
$\alpha$-RCDD and repeat. What remains to show is that for proper
choice of initial $x$ and $e$ as well as proper choice of $g$ in
each iteration, this algorithm converges quickly provided we use a
sufficiently accurate RCDD linear system solver. 

The only missing ingredient for such analysis, is how to measure progress
of the algorithm. Note that $(\mE+\mlap)\mx$ is $\alpha$-CDD if
and only if $e\geq\alpha d$ (Lemma~\ref{lem:column-alphadd-condtions})
and consequently the smallest $\alpha d$ is the smallest value of
$e$ we can hope for to keep $(\mE+\mlap)\mx$ $\alpha$-RCDD. We
show the each iteration of the procedure outlined above brings $e$
rapidly closer to $\alpha$d. In particular we show that in each iteration
we can ensure that $\norm{\md^{-1}(e-\alpha d)}_{1}$ decreases by
a multiplicative constant up to a small additive factor depending
only on $\alpha$ (Lemma~\ref{lem:e-progress}). Putting this all
together we achieve an algorithm (See Section~\ref{sec:computing_stationary:algorithm})
that in $O(\ln(\frac{1}{\alpha}))$ iterations finds the desired $x,e$.

\subsection{The Algorithm\label{sec:computing_stationary:algorithm}}

Following the approach outlined in the previous section, our iterative
algorithm for computing the stationary is given by Algorithm~\ref{alg:stationary_computation}.
First, we compute positive vectors $e^{(0)},x^{(0)}\in\R_{>0}^{V}$
such that and $\mm^{(0)}\defeq(\mE^{(0)}+\mlap)\mx^{(0)}$ is $\alpha$-DD.
In particular, we let $x^{(0)}\defeq\md^{-1}1$ and we let $e^{(0)}\in\R_{\geq0}^{V}$
be the entry-wise minimal non-negative vector such that $\mm^{(0)}$
is $\alpha$-RCDD. This choice of $x^{(0)}$ allows us to reason about
the size of $e^{(0)}$ easily (See Lemma~\ref{lem:bound-on-e0}). 

\begin{algorithm2e}

\caption{Stationary Computation Algorithm}\label{alg:stationary_computation}

\SetAlgoLined

\textbf{Input: }Directed Laplacian $\mlap\in\R^{n\times n}$ with
diagonal $d\defeq\mdiag(\mlap)$

\textbf{Input: }Restart parameter $\alpha\in[0,\frac{1}{2}]$

Set $x^{(0)}:=\md^{-1}1$ 

Set $e^{(0)}\in\R^{V}$ to the entry-wise minimal with $\mm^{(0)}\defeq(\mE^{(0)}+\mlap)\mx^{(0)}$
$\alpha$-RCDD 

\For{$t=0$ to $k=3\ln\alpha^{-1}$}{

Set $g^{(t)}=\frac{1}{\norm{\md^{-1}e^{(t)}}_{1}}\md^{-1}e^{(t)}$
and let $\mlap^{(t)}\defeq\mE^{(t)}-g^{(t)}(e^{(t)})^{\top}+\mlap$

Let $z_{*}^{(t+1)}=[\mm^{(t)}]^{-1}g^{(t)}$, let $x_{*}^{(t+1)}=\mx^{(t)}[\mm^{(t)}]^{-1}g^{(t)}$,
and let $\epsilon=O(\poly(\frac{\alpha}{n}))$ 

Compute $z^{(t+1)}\approx z_{*}^{(t+1)}$ such that $\norm{z^{(t+1)}-z_{*}^{(t+1)}}_{\mdiag(\mm^{(t)})}\leq\frac{\epsilon}{\alpha}\norm{g^{(t)}}_{\mdiag(\mm^{(t)})^{-1}}$

Compute $x^{(t+1)}\approx x_{*}^{(t+1)}$ such that $\norm{\md(x^{(t+1)}-\mx^{(t)}z^{(t+1)})}_{\infty}\leq\epsilon$ 

Set $e^{(t+1)}\in\R^{n}$ to be entry-wise minimal with $\mm^{(t+1)}\defeq(\mE^{(t+1)}+\mlap)\mx^{(t+1)}$
$\alpha$-RCDD 

}

\textbf{Output}: Approximate stationary $s=\frac{\md x^{(k+1)}}{\norm{\md x^{(k+1)}}_{1}}$
such that $\left(3\alpha n\cdot\md+\mlap\right)\md^{-1}\ms$ is $\alpha$-RCDD

\end{algorithm2e}

Next, we iterate, computing $e^{(t)},x^{(t)}\in\R_{\geq0}^{V}$ such
that $\mm^{(t)}=(\mE^{(t)}+\mlap)\mx^{(t)}$ is $\alpha$-RCDD. We
pick the $g^{(t)}$ discussed in the previous section to be $g^{(t)}=\norm{\md^{-1}e^{(t)}}_{1}^{-1}\md^{-1}e^{(t)}$
as this naturally corresponds to the relative amount of each $e^{(t)}$
we want to remove. We then let $x^{(t+1)}\approx\mx^{(t)}[\mm^{(t)}]^{-1}g^{(t)}$
so that $\mlap^{(t)}\defeq\mE^{(t)}-g^{(t)}[e^{(t)}]^{\top}+\mlap$
is a directed Laplacian where $\mlap^{(t+1)}\mx^{(t+1)}$ is nearly
Eulerian. We let $e^{(t+1)}\in\R_{\geq0}^{V}$ be the entry-wise minimal
vector such that $\mm^{(t+1)}$ is $\alpha$-RCDD and then we repeat.
Outputting the final $x^{(t)}$ computed completes the result.

To complete the specification of the algorithm all that remains is
discuss the precision with which we need to carry out the operations
of the algorithm. There are two places in particular where we might
worry that error in inexact arithmetic could hinder the progress of
the algorithm. The first is the degree of precision to which we solve
the linear system in $\mm^{(t)}$, i.e. compute $z_{*}^{(t+1)}\defeq[\mm^{(t)}]^{-1}g^{(t)}$.
We show that computing instead an approximate $z^{(t+1)}$ such that
$\norm{z^{(t+1)}-z_{*}}_{\mdiag(\mm^{(t)})}\leq\frac{\epsilon}{\alpha}\norm{g^{(t+1)}}_{\mdiag(\mm^{(t)})^{-1}}$
for $\epsilon=O(\poly(\frac{\alpha}{n}))$ suffices. The second is
in computing $x^{(t+1)}=\mx^{(t)}z^{(t+1)}$. Unrolling the iterations
of the algorithm we see that $x^{(t+1)}$ is the entry-wise product
of all previous vectors $[\mm^{(t')}]^{-1}g^{(t')}$ with $t'\leq t$.
Consequently, one might worry that since that errors could accumulate.
We show that computing a bounded precision $x^{(t+1)}$ such that
$\norm{\md(x^{(t+1)}-\mx^{(t)}z^{(t+1)})}_{\infty}\leq\epsilon$ for
$\epsilon=O(\poly(\frac{\alpha}{n}))$ suffices. 

\subsection{The Analysis \label{sec:computing_stationary:analysis}}

Here we prove Theorem~\ref{thm:stat-from-dd} by showing the correctness
of our algorithm for computing the stationary distribution, Algorithm~\ref{alg:stationary_computation}.
We split our proof into multiple parts: 
\begin{itemize}
\item \textbf{Section~\ref{sec:computing_stationary:analysis:prop_x}}:
we show that $x_{*}^{(t+1)}$ is indeed a positive vector that makes
$\mlap^{(t)}\mx_{*}^{(t+1)}$ Eulerian
\item \textbf{Section~\ref{sec:computing_stationary:analysis:prop_e}}:
we provide bounds on $e^{(t)}$
\item \textbf{Section~\ref{sec:computing_stationary:analysis:progress:e}}:
we show how much $\norm{\md^{-1}(e^{(t)}-\alpha d)}_{1}$ decreases
between iterations
\item \textbf{Section~\ref{sec:computing_stationary:analysis:linear_system_accuracy}}:
we bound the error induced by solving linear systems approximately
\item \textbf{Section~\ref{sec:computing_stationary:analysis:progress:e}}:
we put everything together to prove Theorem~\ref{thm:stat-from-dd}. 
\end{itemize}
Note that much of the analysis is done in slightly greater generality
then required. In particular, much of our analysis works so long as
$\alpha\geq0$ however we constrain $\alpha\in(0,\frac{1}{2})$ so
we can easily reason about the stability of the procedure, measure
of the quality of the output stationary distribution, and simplify
our analysis. Also note that minimal effort was taken to control the
actual value of $\epsilon$ beyond to simplify the proof that $\epsilon=O(\poly(\frac{\alpha}{n}))$
suffices. 

\subsubsection{Properties of $x$ \label{sec:computing_stationary:analysis:prop_x}}

Here we prove that $x_{*}^{(t+1)}$ is indeed a positive vector such
that $\mlap^{(t)}\mx_{*}^{(t+1)}$ is Eulerian. First we provide two
general lemmas about the kernel of a matrix after a rank one update,
Lemma~\ref{lem:ker_under_update}, and the invertibility of $\alpha$-DD
matrices Lemma~\ref{lem:invertible_B}. Then, using these lemmas
we prove the main result of this subsection, Lemma~\ref{lem:exact_xt},
which yields our desired properties of $x_{*}^{(t+1)}$.
\begin{lem}
\label{lem:ker_under_update}$\ker\left(\mm+uv^{\top}\right)\subseteq\mathrm{span}\left(\mm^{-1}u\right)$
for all invertible $\mm\in\R^{n\times n}$ and $u,v\in\R^{n}$.
\end{lem}
\begin{proof}
If $\ker(\mm+uv^{\top})=\emptyset$ the claim follows trivially. Otherwise,
there is $x\in\R^{n}$ with $x\neq0$ such that $(\mm+uv^{\top})x=0$.
Since $\mm$ is invertible it follows that $x=-\mm^{-1}u(v^{\top}x)\in\mathrm{span}(\mm^{-1}u)$. 
\end{proof}
\begin{lem}
\label{lem:invertible_B} Every strictly RCDD matrix is invertible.
\end{lem}
\begin{proof}
Let $\mm$ be an arbitrary strictly RCDD matrix. By definition, $\mm$
is $\alpha$-RCDD for some$\alpha>0$. Consequently, there exists
$\epsilon>0$ such that $\mn=\mm-\epsilon\mI$ is $\beta$-RCDD for
some $\beta>0$. Now,
\[
\mm^{\top}\mm=(\epsilon\mI+\mn)^{\top}(\epsilon\mI+\mn)=\epsilon^{2}\mI+\epsilon(\mn+\mn^{\top})+\mn^{\top}\mn
\]
However, $\mn+\mn^{\top}$ is clearly a $\beta$-RCDD symmetric matrix
and therefore is PSD. Consequently $\mm^{\top}\mm\succeq\epsilon^{2}\mI$
and therefore $\mm$ doesn't have a non-trivial kernel and is invertible.
\end{proof}
\begin{lem}
\label{lem:exact_xt} Let $x,e\in\R_{>0}^{n}$ be such that $\mm=(\mE+\mlap)\mx$
is strictly RCDD and let $g\in\R_{>0}^{n}$ with $\norm g_{1}=1$.
Then $x_{*}\defeq\mx\mm^{-1}g$ is all positive and $(\mE-ge^{\top}+\mlap)\mx_{*}$
is an Eulerian Laplacian. 
\end{lem}
\begin{proof}
Since $\mm$ is strictly RCDD by Lemma~\ref{lem:invertible_B} it
is invertible. Furthermore, since $x,e,g\geq0$ by Lemma~\ref{lem:inverse-alphacdd_bound}
we know that $\mm^{-1}g>0$ and $x_{*}>0$. Now $\mlap'=(\mE-ge^{\top}+\mlap)\mx$
is a directed Laplacian and therefore has a non-trivial right kernel.
Consequently, by Lemma~\ref{lem:ker_under_update} and the fact that
$\mlap'-\mm=-ge^{\top}\mx$ we know that $\mm^{-1}g$ is in the kernel
of $\mlap^{'}$. Consequently, $(\mE-ge^{\top}+\mlap)\mx\mm^{-1}\mg1=0$
yielding the result.
\end{proof}

\subsubsection{Properties of $e$ \label{sec:computing_stationary:analysis:prop_e}}

Here we provide bounds on what $e^{(t)}$ needs to be for $\mm^{(t)}$
to be $\alpha$-RCDD. First in Lemma~\ref{lem:column-alphadd-condtions}
we give necessary and sufficient conditions for $(\mE+\mlap)\mx$
to be $\alpha$-CDD. This provides a lower bound on all the $e^{(t)}$.
Then in Lemma~\ref{lem:bound-on-e0} we upper bound $e^{(0)}$. We
conclude with Lemma~\ref{lem:formula_for_ei}, which yields a closed
formula for the $e^{(t)}$.
\begin{lem}[Conditions for $\alpha$-CDD]
 \label{lem:column-alphadd-condtions} For all vectors $x,e\in\R_{>0}^{n}$
and directed Laplacian $\mlap\in\R^{n\times n}$ the matrix $(\mE+\mlap)\mx$
is $\alpha$-CDD if and only if $e\geq\alpha d$.
\end{lem}
\begin{proof}
By definition of $\mlap=\md-\ma$ we have that $(\mE+\mlap)\mx$ is
column $\alpha$-CDD if and only if entrywise
\[
1^{\top}(\mE+\md)\mx\geq(1+\alpha)1^{\top}\ma^{\top}\mx=(1+\alpha)d^{\top}\mx.
\]
Applying $\mx^{-1}$ to each side and then subtracting $d^{\top}$
from each side yields the result. 
\end{proof}
\begin{lem}[Bound on $e^{(0)}$]
\label{lem:bound-on-e0} $\norm{\md^{-1}e^{(0)}}_{1}\leq(1+2\alpha)n$
\end{lem}
\begin{proof}
Since $e^{(0)}$ is the entry-wise minimal vector such that $\mm^{(0)}=(\mE^{(0)}+\mlap)\mx^{(0)}$
is $\alpha$-RCDD and since $\mx^{(0)}=\md^{-1}$ we have 
\[
\md^{-1}(e^{(0)}+d)=(1+\alpha)\max\left\{ \ma^{\top}\md^{-1}1\,,\,\md^{-1}\ma1\right\} \leq(1+\alpha)[\ma^{\top}\md^{-1}1+\md^{-1}\ma1].
\]
However, since $\ma1=d$ we have that 
\begin{align*}
\norm{\md^{-1}e^{(0)}}_{1} & \leq\norm{(1+\alpha)\ma^{\top}\md^{-1}1+(1+\alpha)\md^{-1}d-\md^{-1}d}_{1}\\
 & \leq(1+\alpha)\norm{\ma^{\top}\md^{-1}1}_{1}+\alpha\norm 1_{1}\leq(1+2\alpha)|V|.
\end{align*}
Where we used that $\norm{\ma^{\top}\md^{-1}1}_{1}=1^{\top}\ma^{\top}\md^{-1}1=d^{\top}\md^{-1}1=n$.
\end{proof}
\begin{lem}[Formula for $e$]
\label{lem:formula_for_ei} For $e,x\in\R_{>0}^{n}$ and $g\in\simplex^{n}$
let $v\defeq(\mE-ge^{\top}+\mlap)\mx1$ for directed Laplacian $\mlap=\md-\ma^{\top}$.
Then, $f\in\R^{n}$, the entry-wise minimal vector such that $(\mf+\mlap)\mx$
is $\alpha$-RCDD is given by
\[
f=\alpha d+(1+\alpha)\max\left\{ e-\mx^{-1}v-(e^{\top}x)\mx^{-1}g\,,\,0\right\} \,.
\]
\end{lem}
\begin{proof}
By Lemma~\ref{lem:column-alphadd-condtions} we know that $(\mf+\mlap)\mx$
is $\alpha$-CDD if and only if $f\geq\alpha d$. Furthermore, $(\mf+\mlap)\mx$
is $\alpha$-RDD if and only if $(\mf+\md)\mx1\geq(1+\alpha)\ma^{\top}\mx1$
which happens if and only if
\begin{align*}
f & \geq-d+(1+\alpha)\mx^{-1}\ma^{\top}\mx1\\
 & \geq-d+(1+\alpha)\mx^{-1}[-v+\left(\mE-ge^{\top}+\md\right)\mx1]\\
 & =\alpha d+(1+\alpha)[-\mx^{-1}v+e-(e^{\top}x)\mx^{-1}g].
\end{align*}
Taking the maximum of the two lower bounds on $f$ yields the result.
\end{proof}

\subsubsection{Progress from Updating $e$ \label{sec:computing_stationary:analysis:progress:e}}

Here we bound how much progress we make by updating $e$. We first
give a general technical lemma, Lemma~\ref{lem:gen_progress}, and
then in Lemma~\ref{lem:e-progress} we bound how $\norm{\md^{-1}(e^{(t)}-\alpha d)}_{1}$
decreases in each iteration. 
\begin{lem}
\label{lem:gen_progress} For $a,b\in\R_{\geq0}^{n}$ and $z\in\R_{>0}^{n}$
with $b_{i}=\max\left\{ a_{i}-\frac{a^{\top}z}{\norm a_{1}}\cdot\frac{a_{i}}{z_{i}}\,,\,0\right\} $
we have $\norm b_{1}\leq\frac{1}{2}\norm a_{1}$.
\end{lem}
\begin{proof}
Let $T\defeq\left\{ i\in[n]\,:\,\frac{a^{\top}z}{\norm a_{1}}\cdot\frac{a_{i}}{z_{i}}\leq a_{i}\right\} $.
Then we have that
\begin{equation}
\norm a_{1}-\norm b_{1}=\sum_{i\in[n]}(a_{i}-b_{i})=\sum_{i\in T}\frac{a^{\top}z}{\norm a_{1}}\cdot\frac{a_{i}}{z_{i}}+\sum_{i\in[n]\setminus T}a_{i}\,.\label{eq:scale_bound_1}
\end{equation}
We can bound $\norm a_{1}^{2}$ trivially by 
\begin{align}
\norm a_{1}^{2} & =\left(\sum_{i\in T}a_{i}+\sum_{i\in[n]\setminus T}a_{i}\right)^{2}\leq2\left(\sum_{i\in T}a_{i}\right)^{2}+2\left(\sum_{i\in[n]\setminus T}a_{i}\right)^{2}\,.\label{eq:scale_bound_2}
\end{align}
We can bound the first term by Cauchy-Schwarz 
\begin{equation}
\left(\sum_{i\in T}a_{i}\right)^{2}=\left(\sum_{i\in T}\frac{\sqrt{a_{i}}}{\sqrt{z_{i}}}\cdot\sqrt{z_{i}a_{i}}\right)^{2}\leq\sum_{i\in T}\frac{a_{i}}{z_{i}}\sum_{i\in T}a_{i}z_{i}\leq\sum_{i\in T}\frac{a_{i}}{z_{i}}(a^{\top}z)\,.\label{eq:scale_bound_3}
\end{equation}
and the second term trivially by 
\begin{equation}
\left(\sum_{i\in[n]\setminus T}a_{i}\right)^{2}\leq\left(\sum_{i\in[n]\setminus T}a_{i}\right)\sum_{i\in[n]}a_{i}\leq\sum_{i\in[n]\setminus T}a_{i}\norm a_{1}\,.\label{eq:scale_bound_4}
\end{equation}
Combining (\ref{eq:scale_bound_1}), (\ref{eq:scale_bound_2}), (\ref{eq:scale_bound_3}),
and (\ref{eq:scale_bound_4}) yields the result 
\[
\frac{1}{2}\norm a_{1}\leq\sum_{i\in T}\frac{a^{\top}z}{\norm a_{1}}\cdot\frac{a_{i}}{z_{i}}+\sum_{i\in[n]\setminus T}a_{i}=\norm a_{1}-\norm b_{1}\,.
\]
\end{proof}
\begin{lem}[Formula for $e_{i}^{(t)}$]
\label{lem:e-progress} For $e,x\in\R_{>0}^{n}$ and $g\defeq\norm{\md^{-1}e}_{1}^{-1}\md^{-1}e$
let $v\defeq(\mE-ge^{\top}+\mlap)\mx1$. Then setting $f\in\R^{n}$
to the entry-wise minimal vector such that $(\mf+\mlap)\mx$ is $\alpha$-RCDD
gives:
\[
\norm{\md^{-1}(f-\alpha d)}_{1}\leq\frac{1+\alpha}{2}\norm{\md^{-1}(e-\alpha d)}_{1}+(1+\alpha)\norm{\md^{-1}\mx^{-1}v}_{1}+\left(\frac{1+\alpha}{2}\right)\frac{\alpha}{2}n\,.
\]
\end{lem}
\begin{proof}
By Lemma~\ref{lem:formula_for_ei} we know $f=\alpha d+(1+\alpha)\max\left\{ e-\mx^{-1}v-(e^{\top}x)\mx^{-1}g\,,\,0\right\} $.
Consequently
\begin{align*}
\norm{\md^{-1}(f-\alpha d)}_{1} & =(1+\alpha)\norm{\max\{\md^{-1}e-\md^{-1}\mx^{-1}v-(e^{\top}x)\md^{-1}\mx^{-1}g\,,\,0\}}_{1}\\
 & \leq(1+\alpha)\norm{\md^{-1}\mx^{-1}v}_{1}+(1+\alpha)\norm{\max\{\md^{-1}e-(e^{\top}x)\md^{-1}\mx^{-1}g\,,\,0\}}_{1}\,.
\end{align*}
Applying Lemma~\ref{lem:gen_progress} with $a=\md^{-1}e$ and $z=\md x$
we see that
\begin{align*}
\normFull{\max\left\{ \md^{-1}e-(e^{\top}x)\md^{-1}\mx^{-1}g\,,\,0\right\} }_{1} & =\normFull{\max\left\{ a-\frac{a^{\top}z}{\norm a_{1}}\mz^{-1}a\,,\,0\right\} }_{1}\leq\frac{1}{2}\norm a_{1}\\
 & =\frac{1}{2}\norm{\md^{-1}e}_{1}\leq\frac{1}{2}\norm{\md^{-1}(e-\alpha d)}_{1}+\frac{\alpha}{2}n\,.
\end{align*}
Combining yields the result.
\end{proof}

\subsubsection{Linear System Solving Accuracy \label{sec:computing_stationary:analysis:linear_system_accuracy}}

Here we show how to deal with approximate computation.
\begin{lem}
\label{lem:stationary_stable_compute} Let $\mlap=\md-\ma\in\R^{n\times n}$
be a directed Laplacian and let $e,x\in\R_{>0}^{n}$ be such that
$\mm=(\mE+\mlap)\mx$ is $\alpha$-RCDD for $\alpha\in(0,\frac{1}{2})$.
Let $g\defeq\norm{\md^{-1}e}_{1}^{-1}\md^{-1}e$ and let $z\in\R^{n}$
be an approximate solution to $\mm z=g$ in the sense that for $z_{*}\defeq\mm^{-1}g$
we have 
\begin{equation}
\norm{z-z_{*}}_{\mdiag(\mm)}\leq\frac{\epsilon}{\alpha}\norm g_{\mdiag(\mm)^{-1}}\,.\label{eq:lem:stationary_stable_compute:1}
\end{equation}
Furthermore, let $y$ be approximately $\mx z$ in the sense that
$\norm{\md(y-\mx z)}_{\infty}\leq\epsilon\norm{\md^{-1}e}_{\infty}$.
Then if 
\[
\epsilon\leq\epsilon'\cdot\frac{\alpha^{2}}{100}\cdot\frac{1}{\norm{\md^{-1}e}_{1}}\cdot\frac{1}{(n+\norm{\md^{-1}e}_{1})}\cdot\frac{1}{\sqrt{n\cdot(1+\norm{\md^{-1}e}_{\infty})\kappa(\md\mx)}}
\]
for $\epsilon'\in(0,1)$ we have that $y>0$ and $\norm{(\my_{*})^{-1}(y-y_{*})}_{\infty}\leq\epsilon'$
where $y_{*}=\mx\mm^{-1}\md^{-1}e$. Consequently, for $\mlap'=\mE-ge^{\top}+\mlap$
we have $\norm{\my^{-1}\md^{-1}\mlap'y}_{1}\leq\epsilon'$ and $\kappa(\md\my)\leq\frac{10}{\alpha^{2}}\norm{\md^{-1}e}_{1}$.
\end{lem}
\begin{proof}
By Lemma~\ref{lem:inverse-alphacdd_bound} and the fact that $\norm g_{1}=1$
by construction we have
\begin{equation}
\mdiag(\mm)^{-1}g\leq z_{*}\leq\mdiag(\mm)^{-1}g+\frac{1}{\alpha}\mdiag(\mm)^{-1}1\,.\label{eq:lem:stationary_stable_compute:2}
\end{equation}
Consequently, by (\ref{eq:lem:stationary_stable_compute:1}) 
\begin{align*}
\norm{\mz_{*}^{-1}(z-z_{*})}_{\infty} & \leq\norm{\mdiag(\mm)\mg^{-1}(z-z_{*})}_{2}\leq\norm{\mg^{-1}}_{2}\cdot\sqrt{\norm{\mdiag(\mm)}_{2}}\cdot\norm{z-z_{*}}_{\mdiag(\mm)}\\
 & \leq\norm{\mg^{-1}}_{2}\cdot\sqrt{\norm{\mdiag(\mm)}_{2}}\cdot\frac{\epsilon}{\alpha}\norm g_{\mdiag(\mm)^{-1}}\\
 & \leq\frac{\epsilon}{\alpha}\norm{\mg^{-1}}_{2}\cdot\sqrt{\norm{\mdiag(\mm)}_{2}\cdot\norm{\mdiag(\mm)^{-1}}_{2}}\cdot\norm g_{2}\\
 & \leq\frac{\epsilon}{\alpha}\norm{\mg^{-1}}_{2}\cdot\sqrt{\kappa(\mdiag(\mm))}\cdot\sqrt{n}\norm g_{\infty}\\
 & \leq\frac{\epsilon}{\alpha}\kappa(\mg)\sqrt{n\kappa(\mdiag(\mm))}
\end{align*}
Now, $\kappa(\mdiag(\mm))\leq(1+\norm{\md^{-1}e}_{\infty})\kappa(\md\mx)$
and $\norm{\mg}_{2}\leq1$. Furthermore, $\norm{\mg^{-1}}_{2}\leq\frac{1}{\alpha}\norm{\md^{-1}e}_{1}$
since $\mE\succeq\alpha\md$ by Lemma~\ref{lem:column-alphadd-condtions}
and the fact that $\mm$ is $\alpha$-RCDD. Consequently, $\kappa(\mg)\leq\frac{1}{\alpha}\norm{\md^{-1}e}_{1}$
and by our choice of $\epsilon$ we have
\[
\norm{\mz_{*}^{-1}(z-z_{*})}_{\infty}\leq\frac{\epsilon}{\alpha^{2}}\norm{\md^{-1}e}_{1}\sqrt{n\cdot(1+\norm{\md^{-1}e}_{\infty})\kappa(\md\mx)}\defeq c_{1}
\]
This in turn implies that 
\begin{align*}
\norm{\my_{*}^{-1}(y-y_{*})}_{\infty} & =\norm{\mz_{*}^{-1}\mx^{-1}(y-\mx z_{*})}_{\infty}\leq\norm{\mz_{*}^{-1}\mx^{-1}(y-\mx z)}_{\infty}+\norm{\mz_{*}^{-1}\mx^{-1}(\mx z_{*}-\mx z)}_{\infty}\\
 & \leq\norm{\mdiag(\mm)\mg^{-1}\mx^{-1}(y-\mx z_{*})}_{\infty}+c_{1}\\
 & \leq\norm{\mdiag(\mm)\mg^{-1}\mx^{-1}\md^{-1}}_{2}\cdot\norm{\md(y-\mx z_{*})}_{\infty}+c_{1}\\
 & =\norm{\md^{-1}e}_{1}\cdot\norm{\mI+\md\mE^{-1}}_{2}\cdot\frac{\epsilon}{\alpha}+c_{1}\leq\frac{2}{\alpha}\norm{\md^{-1}e}_{1}\cdot\frac{\epsilon}{\alpha}+c_{1}\\
 & \leq\frac{3\epsilon}{\alpha^{2}}\norm{\md^{-1}e}_{1}\sqrt{n\cdot(1+\norm{\md^{-1}e}_{\infty})\kappa(\md\mx)}\defeq c_{2}
\end{align*}
where we used that $\norm{\md\mE^{-1}}_{2}\leq\frac{1}{\alpha}$ by
Lemma~\ref{lem:column-alphadd-condtions} and that $1\leq\frac{1}{\alpha}$
by assumption that $\alpha<\frac{1}{2}$. Consequently, we see that
$y$ is within a small multiplicative constant of $y_{*}$, i.e. $\frac{1}{2}y_{*}\leq(1-c_{2})y_{*}\leq y\leq(1+c_{2})y_{*}\leq1y_{*}$
since by our choice of $\epsilon$ $c_{2}\leq\frac{1}{2}$. Therefore,
applying Lemma~\ref{lem:exact_xt} we know that $\mlap'y_{*}=0$
and 
\begin{align*}
\norm{\my^{-1}\md^{-1}\mlap'y}_{1} & =\norm{\my^{-1}\md^{-1}\mlap'(y-y_{*})}_{1}\leq2\norm{\my_{*}^{-1}\md^{-1}\mlap'\my_{*}\my_{*}^{-1}(y-y_{*})}_{1}\\
 & \leq4\norm{\diag(\my_{*}^{-1}\md^{-1}\mlap'\my_{*})}_{1}\norm{\my_{*}^{-1}(y-y_{*})}_{\infty}\\
 & \leq4\norm{\diag(\mI+\md^{-1}\mE)}_{1}\cdot c_{2}\leq4(n+\norm{\md^{-1}e}_{1})\cdot c_{2}\\
 & \leq\frac{12\epsilon}{\alpha^{2}}\norm{\md^{-1}e}_{1}(n+\norm{\md^{-1}e}_{1})\sqrt{n\cdot(1+\norm{\md^{-1}e}_{\infty})\kappa(\md\mx)}\,.
\end{align*}
Where we used that $\mlap'\my_{*}$ is a Eulerian Laplacian so $\my_{*}^{-1}\md^{-1}\mlap'\my_{*}$is
RDD and we can apply Lemma~\ref{lem:lap-onenorm-bound}. Finally,
we note that by (\ref{eq:lem:stationary_stable_compute:2})
\begin{align*}
\kappa(\md\my_{*}) & =\norm{\md\mz_{*}\mx}_{2}\cdot\norm{\md^{-1}\mz_{*}^{-1}\mx^{-1}}_{2}\\
 & \leq\norm{(\mdiag(\mm)^{-1}\mg+\frac{1}{\alpha}\mdiag(\mm)^{-1})\md\mx}_{2}\cdot\norm{\md^{-1}\mx^{-1}\mdiag(\mm)\mg^{-1}}_{2}\\
 & \leq\norm{(\md^{-1}\mE+\mI)^{-1}(\mg+\alpha^{-1}\mI)}_{2}\cdot\norm{(\md^{-1}\mE+\mI)\mg^{-1}}_{2}\\
 & \leq\frac{2}{\alpha}\cdot\norm{\md^{-1}e}_{1}\norm{(\mI+\md\mE^{-1})}_{2}\leq\frac{4}{\alpha^{2}}\norm{\md^{-1}e}_{1}
\end{align*}
Where we used $\mg\preceq\mI\preceq\frac{1}{\alpha}\mI$ and $\md^{-1}\mE\succeq\alpha\mI$.
Our bound on $\norm{\my_{*}^{-1}(y-y_{*})}$ completes the result.
\end{proof}

\subsubsection{Putting it All Together\label{sec:computing_stationary:analysis:main_theorem}}

Here we show how to put together the results of the previous subsections
to prove Theorem~\ref{thm:stat-from-dd}. 
\begin{proof}[Proof of Theorem~\ref{thm:stat-from-dd}]
 The running time is immediate from the definition of the algorithm
and the fact that for exact $x^{(t)}$ computation it is the case
that $v^{(t)}=0$ by Lemma~\ref{lem:exact_xt}. What remains is to
prove that $\left(3\alpha|V|\cdot\md+\mlap\right)\md^{-1}\ms$ is
an $\alpha$-RCDD. Now, by Lemma~\ref{lem:column-alphadd-condtions}
and Lemma~\ref{lem:bound-on-e0} we know that
\[
\norm{\md^{-1}(e^{(0)}-\alpha d)}_{1}\leq\norm{\md^{-1}e^{(0)}}_{1}-\alpha\norm{\md^{-1}d}_{1}\leq(1+\alpha)\cdot n\leq2n
\]
and clearly $\kappa(\md\mx^{(0)})=1$. Now suppose that for some $t\geq0$
we have that 
\[
\norm{\md^{-1}(e^{(t)}-\alpha d)}_{1}\leq\max\left\{ \left(\frac{7}{8}\right)^{t}2n\,,\,2\alpha n\right\} 
\]
and $\kappa(\md\mx^{(t)})\leq\frac{20}{\alpha^{2}}n$. Then since
this implies $\norm{\md^{-1}e^{(t)}}_{1}\leq2n$ by Lemma~\ref{lem:stationary_stable_compute}
we have that for any absolute constant $\epsilon'$ it is the case
that $x^{(t+1)}>0$ and $\norm{[\mx^{(t+1)}]^{-1}\md^{-1}\mlap^{(t)}x^{(t+1)}}_{1}\leq\epsilon'$
and $\kappa(\md\mx^{(t+1)})\leq\frac{10}{\alpha^{2}}\norm{\md^{-1}e^{(t)}}_{1}\leq\frac{20}{\alpha^{2}}n$.
Furthermore, this implies by Lemma~\ref{lem:e-progress} that for
$\epsilon'\leq\frac{1}{24}$ 
\begin{align*}
\norm{\md^{-1}(e^{(t+1)}-\alpha d)}_{1} & \leq\frac{1+\alpha}{2}\norm{\md^{-1}(e^{(t)}-\alpha d)}_{1}+(1+\alpha)\epsilon'+\left(\frac{1+\alpha}{2}\right)\frac{\alpha}{2}n\\
 & \leq\frac{3}{4}\norm{\md^{-1}(e^{(t)}-\alpha d)}_{1}+\frac{1}{4}\alpha n\leq\max\left\{ \left(\frac{7}{8}\right)^{t+1}2n\,,\,2\alpha n\right\} \,,
\end{align*}
where in the last line we used that if $\norm{\md^{-1}(e^{(t)}-\alpha d)}_{1}\geq2\alpha n$,
then decreasing it by $\frac{3}{4}=(1-\frac{1}{4})$ is the same as
decreasing by $\frac{7}{8}=(1-\frac{1}{8})$ and then subtracting
$\frac{1}{8}$ fraction is subtracting at least $\frac{1}{4}\alpha n$,
which cancels the additive term. 

Consequently, by induction we have that $\norm{\md^{-1}(e^{(t)}-\alpha d)}_{1}\leq\max\left\{ \left(\frac{7}{8}\right)^{t}2n\,,\,2\alpha n\right\} $
for all $t$. Therefore, for $k=8\ln\alpha^{-1}$ we have 
\[
\left(\frac{7}{8}\right)^{t}2n=\left(1-\frac{1}{8}\right)^{t}2n\leq e^{-\frac{t}{8}}2n\leq2n\alpha
\]
and $\norm{\md^{-1}(e^{(k+1)}-\alpha d)}_{\infty}\leq\norm{\md^{-1}(e^{(k+1)}-\alpha d)}_{1}\leq2n\alpha$;
therefore $e^{(t)}\leq3\alpha n\cdot d$. Furthermore, since $(\mE^{(k+1)}+\mlap)\mx^{(k+1)}$
is $\alpha$-RCDD by construction this implies that $(3n\alpha\md+\mlap)\mx^{(k)}$
is $\alpha$-RCDD, yielding the result.
\end{proof}

\section{Eulerian Laplacian Solver \label{sec:eulerian_solver}}

Throughout this section, let $\mlap$ denote an Eulerian directed
Laplacian with $n$ vertices and $m$ edges, and let $\mU$ denote
the associated undirected Laplacian: $\mU\defeq\frac{1}{2}(\mlap+\mlap^{\top})$.
Let $\mb$ denote a weighted edge-vertex incidence matrix for $\mU$,
so that $\mU=\mb^{\top}\mb$. We define $\tsolve\defeq(nm^{3/4}+n^{2/3}m)(\log n)^{3}$
to simplify the statements of our runtime bounds.

The goal of this section is to prove the following, Theorem~\ref{thm:euleriansolver},
showing that we can efficiently solve linear systems in $\mlap$.
Note that we make no attempt to minimize the number of logarithmic
factors in this presentation (including e.g. in parameter balancing);
however, we suspect that with careful use of the recent results \cite{LeeS15}
and \cite{LeePS15} the $\log n$ factors can all be eliminated.
\begin{thm}
\label{thm:euleriansolver}Let $b$ be an $n$-dimensional vector
in the image of $\mlap$, and let $x$ be the solution to $\mlap x=b$
. Then for any $0<\epsilon\leq\frac{1}{2}$, one can compute in $O(\tsolve\log(1/\epsilon))$
time, a vector $x'$ which with high probability is an approximate
solution to the linear system in that $\|x'-x\|_{\mU}\leq\epsilon\|b\|_{\mU^{\dagger}}$. 
\end{thm}
Our proof of Theorem~\ref{thm:euleriansolver} is crucially based
on the symmetric matrix $\mx\defeq\mlap^{\top}\mU^{\dagger}\mlap$.
We begin by noting that $\mx$ is somewhat well approximated by $\mU$:
\begin{lem}
\label{lem:lapbounds}$\mx\succeq\mU$, while $\tr(\mx\mU^{\dagger})=\tr(\mlap^{\top}\mU^{\dagger}\mlap\mU^{\dagger})\leq2(n-1)^{2}$.
\end{lem}
\begin{proof}
We first show that $\mU\preceq\mx$. Let $\mv$ denote the antisymmetric
matrix $\frac{1}{2}(\mlap-\mlap^{\top})$ so that $\mlap=\mU+V$.
With this notation we have
\begin{align*}
\mx & =(\mU+\mv)^{\top}\mU^{\dagger}(\mU+\mv)=(\mU-\mv)\mU^{\dagger}(\mU+\mv)\\
 & =\mU-\mv+\mv-\mv\mU^{\dagger}\mv=\mU+\mv^{\top}\mU^{\dagger}\mv\succeq\mU\,.
\end{align*}

Next, for the trace bound, we first note as $\mlap$ is an Eulerian
Laplacian it can be decomposed as a sum over simple cycles (multiplied
by a scalar weight) $\mlap_{i}$, i.e. 
\[
\mlap=\sum_{i}\mlap_{i}\enspace\text{ and }\enspace\mU=\sum_{i}\mU_{i}
\]
where $\mU_{i}\defeq\frac{1}{2}(\mlap_{i}+\mlap_{i}^{\top})$. 

Now we claim $\mx\preceq\sum_{i}\mlap_{i}^{\top}\mU_{i}^{\dagger}\mlap_{i}.$This
can be shown by simply directly bounding the quadratic form for each
input, showing that

\[
v^{\top}\mx v\leq\sum_{i}(\mlap_{i}v)^{\top}\mU_{i}^{\dagger}(\mlap_{i}v).
\]
This in turn follows from the general statement that for any two vectors
$u$ and $v$, and positive definite matrices $\ma$ and $\mb$,

\[
(u+v)^{\top}(\ma+\mb)^{-1}(u+v)\leq u^{\top}\ma^{-1}u+v^{\top}\mb^{-1}v\,.
\]
Consequently, our claim $\mx\preceq\sum_{i}\mlap_{i}^{\top}\mU_{i}^{\dagger}\mlap_{i}$
holds and therefore $\tr(\mx\mU^{\dagger})\leq\sum_{i}\tr(\mlap_{i}^{\top}\mU_{i}^{\dagger}\mlap_{i}\mU^{\dagger})$.

Next, we note that if $\mlap_{i}$ is a simple cycle of length $k_{i}$
and weight $w_{i}$ through the vertices $V_{i}\subseteq V$, then
$\mlap_{i}^{\top}\mU_{i}^{\dagger}\mlap_{i}$ is precisely the clique
on those $k_{i}$ vertices, multiplied by $\frac{2w_{i}}{k_{i}}$.
This can be verified, for instance, using an eigenvector decomposition
(the Fourier vectors simultaneously diagonalize $\mlap_{i}$, $\mlap_{i}^{\top}$,
and $\mU_{i}$). We then can express $\tr(\mlap_{i}^{\top}\mU_{i}^{\dagger}\mlap_{i}\mU^{\dagger})$
as $\frac{2w_{i}}{k_{i}}$ times the sum over all pairs of vertices
$u,v\in V_{i}$ of the effective resistance between $u$ and $v$
in $\mU$. Effective resistances form a metric; since any two vertices
in $V_{i}$ are connected by some path through the cycle, their effective
resistance must always be at most the sum of the effective resistances
between the edges in the cycle, which in turn is $\frac{2}{w_{i}}\tr(\mU_{i}\mU^{\dagger})$.
With $\frac{k_{i}(k_{i}-1)}{2}$ pairs of vertices, we have
\[
\tr(\mlap_{i}^{\top}\mU_{i}^{\dagger}\mlap_{i}\mU^{\dagger})\leq2(k_{i}-1)\tr(\mU_{i}\mU^{\dagger})\leq2(n-1)\tr(\mU_{i}\mU^{\dagger})
\]
and
\[
\sum_{i}\tr(\mlap_{i}^{\top}\mU_{i}^{\dagger}\mlap_{i}\mU^{\dagger})\leq2(n-1)\sum_{i}\tr(\mU_{i}\mU^{\dagger})\leq2(n-1)\tr(\mU\mU^{\dagger})=2(n-1)^{2}\,.
\]
\end{proof}
Lemma~\ref{lem:lapbounds} immediately suggests our approach: we
essentially solve systems in $\mlap$ by solving systems in $\mx$
(technically, we use a slightly different matrix $\widetilde{\mx}$
which we define later), and use preconditioning schemes for symmetric
positive definite matrices based on the structure of $\mx$ and $\mU$. 

A simple preconditioning strategy would be to directly use $\mU$
as a preconditioner for $\mx$ in an algorithm like preconditioned
Chebyshev iteration. This depends on the \emph{relative condition
number} of $\mU$ and $\mx$; the relative condition number of two
symmetric positive semidefinite matrices $\ma$ and $\mb$ is defined
to be infinite if they don't have the same null space, and otherwise
to be $\frac{\lambda_{max}(\ma\mb^{\dagger})}{\lambda_{min}(\ma\mb^{\dagger})}.$
Since the relative condition number of $\mx$ and $\mU$ is at most
$O(n^{2})$ by Lemma~\ref{lem:lapbounds} and iterative algorithms
like preconditioned Chebyshev iteration have an iteration count proportional
to the square root of the relative condition number, this would require
about $n$ iterations for a total runtime of about $nm$ (using fast
Laplacian solvers to solve systems in $\mU$). 

One might hope to improve this by plugging the \emph{trace} bound
into a result from \cite{SpielmanW09}, which says that if one uses
preconditioned \emph{conjugate gradient} in exact arithmetic, the
iteration count scales only with the cube root of the trace. That
would give a runtime of about $n^{2/3}m$. However, it is unknown
how to obtain this result without using exact arithmetic (or a polynomial
number of bits of precision). Since we are aiming for algorithms which
only need a logarithmic number of bits of precision (see Appendix~\ref{sec:numerical_stability})
it is unclear how to leverage preconditioned conjugate gradient directly
to improve the running time and this remains an interesting open question. 

Instead use a technique similar to the method of transmogrification
used in Laplacian solvers, in particular \cite{KoutisMP11}. However,
our method deviates from much of this previous work in that it is
non-recursive and uses only the Woodbury matrix identity (instead
of partial Cholesky factorization) to solve the preconditioner. This
allows us to obtain a preconditioner with a better relative condition
number than $\mU$ itself, which can still be applied nearly as efficiently.
First, we state the ``matrix Chernoff bound'' that we will use to
construct the preconditioner:
\begin{lem}
\label{lem:matrixchernoff}Let $\epsilon,\delta\in(0,1/2)$. Let $\mm=\sum_{i}\mm_{i}$
be the sum of independent random symmetric $n$-dimensional matrices
$\mm_{i}$, where each $\mm_{i}$ is positive definite and satisfies
$\|\mm_{i}\|_{2}\leq\frac{\epsilon^{2}}{3\log(2n/\delta)}$ with probability
1. Furthermore, let $\mathrm{E}[\mm]\preceq\mI$. Then $\|\mm-\mathrm{E}[\mm]\|\leq\epsilon$
with probability at least $1-\delta$.
\end{lem}
\begin{proof}
Apply Theorem 1.4 (``matrix Bernstein'') from \cite{Tropp12} to
$\sum_{i}(\mm_{i}-\mathrm{E}[\mm_{i}])$, and let $R'=\frac{\epsilon^{2}}{3\log(2n/\delta)}$.
We have unconditionally, for all $i$, $\|\mm_{i}-\mathrm{E}[\mm_{i}]\|\leq\|\mm_{i}\|+\|\mathrm{E}[\mm_{i}]\|\leq2R'$.
Furthermore, we have
\begin{align*}
\sum_{i}\mathrm{E}[(\mm_{i}-\mathrm{E}[\mm_{i}])^{2}] & =\sum_{i}(\mathrm{E}[\mm_{i}^{2}]-(\mathrm{E}[\mm_{i}])^{2})\preceq\sum_{i}\mathrm{E}[\mm_{i}^{2}]\preceq R'\sum_{i}\mathrm{E}[\mm_{i}]\preceq R'\mI\,.
\end{align*}

We thus have $R=2R'$ and $\sigma^{2}\leq R'$ in that theorem statement.
The probability that the max eigenvalue of $\mm-\mathrm{E}[\mm]$
is greater than $\epsilon$ is then at most

\begin{align*}
n\exp\left(-\frac{\epsilon^{2}/2}{R'+2R'\epsilon/3}\right) & =n\exp\left(-\frac{\epsilon^{2}}{2R'(1+2\epsilon/3)}\right)\leq n\exp\left(-\frac{\epsilon^{2}}{3R'}\right)\\
 & \leq n\exp(-\log(2n/\delta))=\frac{\delta}{2}\,.
\end{align*}
Applying the same argument to the negation $\mathrm{E}[\mm]-\mm$
and union bounding gives the result.
\end{proof}
For notational convenience (to deal with inexact Laplacian system
solvers), for the next statements, we let $\widetilde{\mU}$ be any
symmetric matrix satisfying $\frac{1}{2}\mU\preceq\widetilde{\mU}\preceq\mU,$
and let $\widetilde{\mx}\defeq\mlap^{\top}\widetilde{\mU}^{\dagger}\mlap$.
The next statement gives us a method to construct preconditioners
for $\mx$; it is inspired by the ``ultrasparsifiers'' or ``incremental
sparsifiers'' from \cite{KoutisMP11}.
\begin{lem}
\label{lem:ultrasparsify} For any $k\in[0,n^{2}]$, there exists
a set of nonnegative weights $w_{i}$, with at most $O(\frac{n^{2}\log n}{k})$
of them nonzero, such that
\[
\frac{1}{2}\left(\widetilde{\mU}+\frac{1}{k}\widetilde{\mx}\right)\preceq\widetilde{\mU}+\mlap^{\top}\widetilde{\mU}^{\dagger}\mb^{\top}\mw\mb\widetilde{\mU}^{\dagger}\mlap\preceq2\left(\widetilde{\mU}+\frac{1}{k}\widetilde{\mx}\right).
\]
Furthermore, if $\widetilde{\mU}^{\dagger}$can be applied to a vector
in $O(m\log n)$ time, such weights can be computed in time $O(m(\log n)^{2})$
with high probability. Note that no access to $\widetilde{\mU}$ itself
is needed; only its pseudoinverse.
\end{lem}
\begin{proof}
Define $l_{i}\defeq\|\mlap^{\top}\widetilde{\mU}^{\dagger}b_{i}\|_{\widetilde{\mU}^{\dagger}}$
where $b_{i}$ is the $i$th row of $\mb$ for all $i\in[n]$. Then
there exists some universal constant $C$ (depending on the exponent
in the ``high probability'') such that if $p_{i}\geq Cl_{i}\log n$,
we can independently pick $\frac{\sum_{i}p_{i}}{k}$ rows of $\mb$,
with replacement, with probability of picking row $i$ each time being
$\frac{p_{i}}{\sum_{i}p_{i}}$, then add $\frac{1}{p_{i}}$ to $w_{i}$
each time $i$ is picked. Here, we define $\mv=\widetilde{\mU}+\frac{1}{k}\widetilde{\mx}$
and $\widetilde{\mv}=\widetilde{\mU}+\mlap^{\top}\widetilde{\mU}^{\dagger}\mb^{\top}\mw\mb\widetilde{\mU}^{\dagger}\mlap$.
Then we apply Lemma~\ref{lem:matrixchernoff} to $\mv^{-\frac{1}{2}}(\widetilde{\mv}-\widetilde{\mU})\mv^{-\frac{1}{2}}$,
which is a sum of independent random matrices, with expectation$\preceq\mI$,
with each term having spectral norm at most $\frac{1}{C\log n}$.
If $C$ is set sufficiently large this guarantees the desired result.

We have
\begin{align*}
\sum_{i\in[n]}l_{i} & =\sum_{i\in[n]}b_{i}^{\top}\widetilde{\mU}^{\dagger}\mlap\widetilde{\mU}^{\dagger}\mlap^{\top}\widetilde{\mU}^{\dagger}b_{i}=\sum_{i\in[n]}\tr(\mlap^{\top}\widetilde{\mU}^{\dagger}b_{i}b_{i}^{\top}\widetilde{\mU}^{\dagger}\mlap\widetilde{\mU}^{\dagger})\leq\tr(\mlap^{\top}(4\mU^{\dagger})\mlap(2\mU^{\dagger}))\\
 & =8\tr(\mx\mU^{\dagger})=O(n^{2})\,.
\end{align*}
where the last line used Lemma~\ref{lem:lapbounds}. Thus, if we
can compute a constant-factor approximation to the $l_{i}$, we can
obtain $p_{i}\geq Cl_{i}\log(n)$ with $\sum_{i}p_{i}=O(n^{2}\log n)$,
and the random sample will be of size $O(k^{-1}n^{2}\log n)$ as desired.
Next, we note that $l_{i}$ is within a factor of 2 of $\|\mb\widetilde{\mU}^{\dagger}\mlap\widetilde{\mU}^{\dagger}b_{i}\|_{2}$,
since $\widetilde{\mU}^{\dagger}\mb^{\top}\mb\widetilde{\mU}^{\dagger}=\mU^{\dagger}\mU\mU^{\dagger}$
is within a factor of two of $\widetilde{\mU}^{\dagger}=\widetilde{\mU}^{\dagger}\widetilde{\mU}\widetilde{\mU}^{\dagger}$.
Thus a constant factor approximation to the latter expression is sufficient.
Finally, we note that we can compute all of these constant factor
approximations using Johnson-Lindenstrauss (analogous to the construction
in \cite{SpielmanS08}). In particular, we may let $\mpi$ be a Gaussian
random matrix with $r=O(\log n)$ rows and $m$ columns, rescaled
by $\frac{1}{\sqrt{r}}$. Then the Johnson-Lindenstrauss lemma implies
that with high probability, $\|\mpi\mb\widetilde{\mU}^{\dagger}\mlap\widetilde{\mU}^{\dagger}b_{i}\|_{2}$
will be within a constant factor of $\|\mb\widetilde{\mU}^{\dagger}\mlap\widetilde{\mU}^{\dagger}b_{i}\|_{2}$
for all $i$. We may compute each row of $\mpi\mb\widetilde{\mU}^{\dagger}\mlap\widetilde{\mU}^{\dagger}$
in $O(m\log n)$ time by successively multiplying the rows of $\mpi$
by $m$-sparse matrices and $\widetilde{\mU}^{\dagger},$ each of
which can be done in $O(m\log n)$ time. Thus computing it for all
rows takes time $O(m(\log n)^{2}).$ Finally, since $\mpi\mb\widetilde{\mU}^{\dagger}\mlap\widetilde{\mU}^{\dagger}$
has $O(\log n)$ rows and $b_{i}$ is 2-sparse, we can compute each
$\mpi\mb\widetilde{\mU}^{\dagger}\mlap\widetilde{\mU}^{\dagger}b_{i}$
in $O(\log n)$ time, taking $O(m\log n)$ time in total.
\end{proof}
Lemma~\ref{lem:ultrasparsify} says that for large $k$, we can approximate
$\widetilde{U}+\frac{1}{k}\widetilde{X}$ with a low-rank update to
$\widetilde{\mU}$, i.e. rank $O(k^{-1}n^{2}\log n)$ to be precise.
In the next lemma we show that with a little precomputation, we can
apply the pseudoinverse of this approximation efficiently.
\begin{lem}
\label{lem:ShermanMorrison} Consider the matrix $\mz=\widetilde{\mU}+\mlap^{\top}\widetilde{\mU}^{\dagger}\mb^{\top}\mw\mb\widetilde{\mU}^{\dagger}\mlap$.
If $\widetilde{\mU}^{\dagger}$ can be applied to a vector in time
$T$, and $\mw$ has $r$ non-zero entries, then after precomputation
taking $O((m+T)r+r^{3})$ time, we can apply $\mathbf{Z}^{\dagger}$
to a vector in $O(m+T+r^{2})$ time. 
\end{lem}
\begin{proof}
We define $\mr$ to be the restriction of $\mb\widetilde{\mU}^{\dagger}\mlap$
to the rows where $\mw$ is nonzero, so that $\mz=\widetilde{\mU}+\mr^{\top}\mw_{S}\mr$
(where $\mw_{S}$ is $\mw$ similarly restricted). Then we can apply
the Woodbury matrix identity to get that
\[
\mz^{\dagger}=\widetilde{\mU}^{\dagger}-\widetilde{\mU}^{\dagger}\mr^{\top}(\mw_{S}^{-1}+\mr\widetilde{\mU}^{\dagger}\mr^{\top})^{-1}\mr\widetilde{\mU}^{\dagger}.
\]

First, we define the precomputation stage, which computes the entries
of the matrix $\mm=(\mw_{S}^{-1}+\mr\widetilde{\mU}^{\dagger}\mr^{\top})^{-1}$.
This has two logical steps: computing the entries of $\mr\widetilde{\mU}^{\dagger}\mr^{\top}$
and then computing the matrix inverse. We compute $\mr\widetilde{\mU}^{\dagger}\mr^{\top}$column
by column, simply applying $\mr^{\top}$, $\widetilde{\mU}^{\dagger}$,
and $\mr$ in succession to each of the standard basis vectors. The
applications of $\mr$ and $\mr^{\top}$ are done by successively
applying $\mb$, $\widetilde{\mU}^{\dagger}$,$\mlap$, and their
transposes; each of these is either an $O(m)$-sparse matrix or one
which we are promised can be applied in time $T$. Thus, each of these
columns can be computed in time at most $O(m+T)$; getting all $r$
columns then takes $O((m+T)r)$ time. Adding this to $\mw_{S}^{-1}$
(which is diagonal) and then inverting the $r\times r$ matrix takes
time at most $O(r^{3})$ even using a naive algorithm such as Gaussian
elimination (which already suffices to achieve our desired result).

Now, once this $\mm$ has been precomputed, we can apply the operator
$\mz^{\dagger}$ to a vector just by applying $\widetilde{\mU}^{\dagger}$
and applying $\widetilde{\mU}^{\dagger}\mr^{\top}\mm\mr\widetilde{\mU}^{\dagger}$
by applying the individual factors successively. Applying $\mm$ costs
at most $O(r^{2})$ since that is the size of the matrix, while the
other applications take $O(m+T);$ thus the total time to apply is
$O(m+T+r^{2}).$
\end{proof}
\begin{lem}
\label{lem:complicated-matrix-solver}If $\widetilde{\mU}^{\dagger}$
can be applied to a vector in $O(m\log n)$ time, we can, with high
probability, apply a linear operator $\my^{\dagger}$ satisfying $(1-\epsilon)\widetilde{\mx}^{\dagger}\preceq\my^{\dagger}\preceq(1+\epsilon)\widetilde{\mx}^{\dagger}$
in $O((nm^{3/4}+n^{2/3}m)(\log n)^{3}\log(1/\epsilon))=O(\tsolve\log(1/\epsilon))$
time. 
\end{lem}
\begin{proof}
Invoke \ref{lem:ultrasparsify} for $k=\max(n^{2}/\sqrt{m},n^{4/3})$,
and define $\mz=\widetilde{\mU}+\mlap^{\top}\widetilde{\mU}^{\dagger}\mb^{\top}\mw\mb\widetilde{\mU}^{\dagger}\mlap$. 

The number of nonzero entries in $\mw$ will then be at most $O(\min(\sqrt{m},n^{2/3})\log n)$.
Now, apply \ref{lem:ShermanMorrison}. The precomputation takes time
$O((m\log n)r+r^{3})\leq O(n^{2/3}m(\log n)^{2}+\min(m^{3/2},n^{2})(\log n)^{3}).$
Since $\min(m^{3/2},n^{2})\leq nm^{3/4}+n^{2/3}m$, this is $O((nm^{3/4}+n^{2/3}m)(\log n)^{3}$.
The time per application is at most $O(m+m\log n+(\sqrt{m}\log n)^{2})=O(m(\log n)^{2}).$

Now, we use $\mz^{+}$ as a preconditioner in an iterative method:
preconditioned Chebyshev iteration. Because the relative condition
number between $\mz$ and $\widetilde{\mx}$ is $O(k)$, the number
of iterations required is $O(\sqrt{k}\log(1/\epsilon))=O((n/m^{1/4}+n^{2/3})\log(1/\epsilon))$.
Each iteration requires one application of $\widetilde{\mx}$ and
one application of $\mz^{\dagger}$ to vectors, each of which takes
$O(m(\log n)^{2})$ time. The total running time of this solve phase
is then at most $O((nm^{3/4}+n^{2/3}m)(\log n)^{2}\log(1/\epsilon))$.
\end{proof}
Now, we turn Lemma~\ref{lem:complicated-matrix-solver} into a solver
for $\mlap$.
\begin{lem}
\label{lem:weirdnormsolver} Let $b$ be an $n$-dimensional vector
in the image of $\mlap$, and let $x$ be the solution to $\mlap x=b$.
Then for any $0<\epsilon\leq\frac{1}{2}$, one can compute in $O(\tsolve\log(1/\epsilon))$
time a vector $x'$ which with high probability is an approximate
solution to the linear system in that $\|x'-x\|_{\mx}\leq\epsilon\|b\|_{\mU^{\dagger}}$.
Furthermore, a given choice of random bits produces a correct result
for all $b$ simultaneously, and makes $x'$ linear in $b$.
\end{lem}
\begin{proof}
First note that we can actually choose a $\widetilde{\mU}^{\dagger}$
satisfying the needed properties: a linear-operator based graph Laplacian
solver, such as \cite{KoutisMP11}. With this $\widetilde{\mU}^{\dagger}$,
we apply Lemma~\ref{lem:complicated-matrix-solver}, which produces
a linear operator $\my^{\dagger}.$ We simply return $x'=\my^{\dagger}\mlap^{\top}\widetilde{\mU}^{\dagger}b$\textendash corresponding
to approximately solving the linear system $\mlap^{\top}\widetilde{\mU}^{\dagger}\mlap x=\mlap^{\top}\widetilde{\mU}^{\dagger}b$.
The error bound follows from
\begin{align*}
\|x'-x\|_{\mx} & \leq\|x'-x\|_{\widetilde{\mx}}=\|(\my^{\dagger}-\widetilde{\mx}^{\dagger})\mlap^{\top}\widetilde{\mU}^{\dagger}b\|_{\widetilde{X}}=\|\mlap^{\top}\widetilde{\mU}^{\dagger}b\|_{(\my^{\dagger}-\widetilde{\mx}^{\dagger})\widetilde{\mx}(\my^{\dagger}-\widetilde{\mx}^{\dagger})}\\
 & \leq\epsilon\|\mlap^{\top}\widetilde{\mU}^{\dagger}b\|_{\widetilde{\mx}^{\dagger}}=\epsilon\|\widetilde{\mU}^{\dagger}b\|_{\mlap\widetilde{\mx}^{\dagger}\mlap^{\top}}=\epsilon\|\widetilde{\mU}^{\dagger}b\|_{\widetilde{\mU}}\\
 & =\epsilon\|b\|_{\widetilde{\mU}^{\dagger}}\leq\epsilon\|b\|_{\mU^{\dagger}}.
\end{align*}
\end{proof}
Theorem~\ref{thm:euleriansolver} follows as an immediate corollary:
\begin{proof}[Proof of Theorem \ref{thm:euleriansolver}]
Apply Lemma \ref{lem:weirdnormsolver} and use the fact (from Lemma
\ref{lem:lapbounds}) that $\mU\preceq\mx$:
\begin{align*}
\|x'-x\|_{\mU} & \leq\|x'-x\|_{\mx}\leq\epsilon\|b\|_{\mU^{\dagger}}.
\end{align*}
\end{proof}

\section{Solving Strictly RCDD Systems \label{sec:solving_strictly_rcdd}}

Here we show how to reduce solving strictly RCDD systems using an
Eulerian Laplacian solver. First we provide an prove Theorem~\ref{thm:basic-rcdd}
which achieves precisely this goal by using the Eulerian Laplacian
system solver presented in Section~\ref{sec:solving_strictly_rcdd}.
Then we provide Corollary~\ref{cor:alpha-rcdd-solver} which specializes
this result to the case of solving $\alpha$-RCDD systems. We make
extensive use of this corollary to compute the the Eulerian scaling
of a directed Laplacian in Section~\ref{sec:computing_stationary}.
\begin{thm}
\label{thm:basic-rcdd} Let $\ma\in\R^{n\times n}$ be a strictly
RCDD $Z$-matrix, let $b\in\R^{n}$, let $x$ be the solution to $\ma x=b$,
and let $0<\epsilon\leq\frac{1}{2}$. Then in $O(\tsolve\log(1/\epsilon))$
time we can compute a vector $x'$ that satisfies $\|x'-x\|_{\frac{1}{2}(\ma+\ma^{\top})}\leq\epsilon\|b\|_{(\frac{1}{2}(\ma+\ma^{\top}))^{-1}}$
with high probability.
\end{thm}
\begin{proof}
Define the $(n+1)\times n$ matrix $\mc=$$\begin{pmatrix}\mI_{n}\\
-\vones_{n}^{\top}
\end{pmatrix}$. Now define a $(n+1)\times(n+1)$ matrix $\mlap=\mc\ma\mc^{\top}$.
The top left $n\times n$ entries of $\mlap$ are a copy of $\ma$,
and its rows and columns sum to 0 since the columns of $\mc$ sum
to 0. Since $\ma$ is a diagonally dominant Z-matrix, the new off-diagonal
entries are negative and the new diagonal entry is then positive.
Thus $\mlap$ is a Z-matrix for which all rows and columns sum to
to 0, so it is an Eulerian Laplacian.

Now, if $\ma x=b$, we must have $\mlap\begin{pmatrix}x\\
0
\end{pmatrix}=\begin{pmatrix}b\\
y
\end{pmatrix}$ for some scalar $y$, since the upper left block is $\ma$. Furthermore,
this must satisfy $\vones_{n}^{\top}b+y=0$, since the image of a
Laplacian is orthogonal to the all-ones vector. Consequently, $\begin{pmatrix}b\\
y
\end{pmatrix}=\mc x$. Now, for a connected Eulerian Laplacian, the null space is precisely
the all-ones vector. Furthermore, $\mlap$ is clearly connected as
the $n+1$ row must be non-zero by the fact that $\ma$ is strictly
RCDD. Consequently, any solution to $\mlap z=\begin{pmatrix}b\\
y
\end{pmatrix}$ must be off by only by a shift of the all-ones vector from $\begin{pmatrix}x\\
0
\end{pmatrix}$. In particular, this implies that $x=\mc^{\top}\mlap^{\dagger}\mc b$
since the $\mc^{\top}$ adds the shift of the all-ones vector to zero
out the auxiliary variable. This means that $\ma^{-1}=\mc^{\top}\mlap^{\dagger}\mc$,
and by the same argument, $(\frac{1}{2}(\ma+\ma^{\top}))^{-1}=\mc^{\top}(\frac{1}{2}(\mlap+\mlap^{\top}))^{\dagger}\mc=\mc^{\top}\mU^{\dagger}\mc$.
This effectively reduces solving RCDD systems to solving Eulerian
Laplacian systems.

Specifically, we obtain $x'$ by calling Theorem~\ref{thm:euleriansolver}
on $\mc b$ with $\mlap$ and $\epsilon$ to get a vector $z'$ (with
$z$ referring to $\mlap^{\dagger}\mc b$), then returning $\mc^{\top}z'$.
The desired error follows from

\begin{align*}
\|x'-x\|_{\frac{1}{2}(\ma+\ma^{\top})} & =\|\mc^{\top}(z'-z)\|_{\frac{1}{2}(\ma+\ma^{\top})}=\|z'-z\|_{\frac{1}{2}\mc(\ma+\ma^{\top})\mc^{\top}}=\|z'-z\|_{\mU}\\
 & \leq\epsilon\|\mc b\|_{\mU^{\dagger}}=\epsilon\|b\|_{\mc^{\top}\mU^{\dagger}\mc}=\epsilon\|b\|_{(\frac{1}{2}(\ma+\ma^{\top}))^{-1}}.
\end{align*}
\end{proof}
\begin{cor}
\label{cor:alpha-rcdd-solver} Let $\ma\in\R^{n\times n}$ be a $\alpha$-RCDD
$Z$-matrix, let $\md$ represent the matrix with $\md_{ij}=0$ for
$i\neq j$ and $\md_{ii}=\frac{A_{ii}}{1+\alpha}$, let $x$ be the
solution to $\ma x=b$ for $b\in\R^{n},$ and let $0<\epsilon\leq\frac{1}{2}$.
Then in $O(\tsolve\log(1/\epsilon))$ time we can compute a vector
$x'$ satisfying $\|x'-x\|_{\md}\leq\frac{\epsilon}{\alpha}\|b\|_{\md^{-1}}$
with high probability.
\end{cor}
\begin{proof}
We apply Theorem~\ref{thm:basic-rcdd} to our input. We can express
$\frac{1}{2}(\ma+\ma^{\top})$ as $\alpha\md+\mm$, where $\mm$ is
a symmetric diagonally dominant, and thus positive definite, matrix.
In other words we have

\[
\frac{1}{2}(\ma+\ma^{\top})\succeq\alpha\md
\]
meaning that the $\frac{1}{2}(\ma+\ma^{\top})$ norm dominates the
$\alpha\md$ norm while the $(\frac{1}{2}(\ma+\ma^{\top}))^{-1}$
norm is dominated by the $\alpha^{-1}\md^{-1}$ norm. Thus we have

\begin{align*}
\|x'-x\|_{\md} & =\alpha^{-\frac{1}{2}}\|x'-x\|_{\alpha\md}\leq\alpha^{-\frac{1}{2}}\|x'-x\|_{\frac{1}{2}(\ma+\ma^{\top})}\leq\alpha^{-\frac{1}{2}}\epsilon\|b\|_{(\frac{1}{2}(\ma+\ma^{\top}))^{-1}}\\
 & \leq\alpha^{-\frac{1}{2}}\epsilon\|b\|_{\alpha^{-1}\md^{-1}}=\frac{\epsilon}{\alpha}\|b\|_{\md^{-1}}.
\end{align*}
\end{proof}

\section{\label{sec:application_quantities}Condition Number Bounds}

For the various applications of this paper (Section~\ref{sec:applications})
we use three quantities to measure the degeneracy or condition number
of the matrices involved. Here we define these quantities and prove
that they are equivalent up to polynomial factors. Since the running
times for all algorithms provided in this paper depend on these quantities
only poly-logarithmically, this allows us to use whichever quantity
simplifies the analysis while only mildly affecting performance. 

We define these quantities for a directed Laplacian, $\mlap=\mI-\mw\in\R^{n\times n}$,
where $\mw\in\R_{\geq0}^{n\times n}$ is the random walk matrix associated
with a strongly connected graph. Note that for any Laplacian, $\mlap=\md-\ma^{\top}$,
we have $\mlap\md^{-1}=\mI-\mw$; thus it is easy to relate these
quantities to those of an arbitrary Laplacian, since $\kappa(\md-\ma)\leq\kappa(\mI-\ma\md^{-1})\cdot\kappa(\md)$.\\
\\
\textbf{Inverse of Second Smallest Singular Value}: This quantity
is $\norm{\mlap^{\dagger}}_{2}$. Note that since $\norm{\mlap}_{2}\leq2n$,
by Lemma~\ref{lem:lap-onenorm-bound} we have $\kappa(\mlap)\leq2n\norm{\mlap^{\dagger}}_{2}$.
Thus, up to a polynomial dependence on $n$, this quantity bounds
the condition number of $\mlap$.\\
\\
\textbf{(Lazy) Mixing Time}: This quantity is denoted $t_{mix}$.
It is defined as the smallest $k$ such that $\norm{\lazywalk^{k}p-s}_{1}\leq\frac{1}{2}$,
for all $p\in\simplex^{n}$, where $\lazywalk\defeq\frac{1}{2}(\mI+\mw)$
is the lazy random walk matrix associated with $\mw$, and $s\in\simplex^{n}$
is the stationary distribution associated with $\mw$.\\
\\
\textbf{Personalized PageRank Mixing Time}: This quantity is denoted
$t_{pp}$. It is defined as the smallest $k\geq0$ such that, setting
$\beta=\frac{1}{k}$, one has $\norm{\mm_{pp(\beta)}p-s}_{1}\leq\frac{1}{2}$,
for all $p\in\simplex^{n}$; the matrix $\mm_{pp(\beta)}=\beta\left(\mI-(1-\beta)\mw\right)^{-1}$
is the personalized PageRank matrix with restart probability $\beta$
(See Section~\ref{sec:prelim}), and $s$ is the stationary distribution
associated with $\mw$. Note that $\mm_{pp(\beta)}p$ is the personalized
PageRank vector with restart probability $\beta$ and vector $p$.\\
\\
Each of these quantities is required for our applications, and they
are involved in the running time guarantees we provide. The inverse
of the second smallest singular value is related to the condition
number of $\mlap$, and is used to bound the accuracy to which we
need to solve various linear systems. The mixing time is related to
many applications involving directed Markov chains. In order to obtain
bounds on these quantities, we use the personalized PageRank mixing
time, which can be computed accurately enough for our purposes (See
Section~\ref{subsec:app-cond-est}). 

The main result of this section is to prove the following theorem
showing how these different quantities are related. Rather than relating
many of the quantities to $\norm{\mlap^{\dagger}}_{2}$ directly,
we instead obtain cleaner bounds by relating them to the maximum column
norm of $\mlap$, i.e. $\norm{\mlap^{\dagger}}_{1}$. Since for $x\in\R^{n}$
we have $\frac{1}{\sqrt{n}}\norm x_{1}\leq\norm x_{2}\leq\norm x_{1}$,
it is easy to see that $\frac{1}{\sqrt{n}}\norm{\mlap^{\dagger}}_{1}\leq\norm{\mlap^{\dagger}}_{2}\leq\sqrt{n}\norm{\mlap^{\dagger}}_{1}$.
\begin{thm}
\label{thm:mixingTimesSingular} Let $\mlap=\mI-\mw\in\R^{n\times n}$
be a directed Laplacian, where $\mw\in\R_{\geq0}^{n\times n}$ is
a random walk matrix associated with a strongly connected graph. The
inverse of the second smallest singular value of $\mlap$, $\norm{\mlap^{\dagger}}_{2}$,
the mixing time, $t_{mix}$, the personalized PageRank mixing time,
$t_{pp}$, and the largest column norm of $\mlap^{\dagger},$ $\norm{\mlap^{\dagger}}_{1},$
obey the following:

\begin{itemize}
\item \textbf{Mixing Time vs. Column Norm}: $\frac{1}{16}\sqrt{t_{mix}}\leq\norm{\mlap^{\dagger}}_{1}\leq t_{mix}\cdot4\sqrt{n}\log_{2}n$
\item \textbf{Personalized PageRank Mixing vs. Singular Values}: $\frac{1}{8}t_{pp}\leq\norm{\mlap^{\dagger}}_{1}\leq t_{pp}\cdot16\sqrt{n}\log_{2}n$
\item \textbf{Mixing Time vs. Personalized PageRank Mixing} Time: $t_{pp}\leq36\cdot t_{mix}$
\item \textbf{Column Norm vs. Singular Values}: $\frac{1}{\sqrt{n}}\norm{\mlap^{\dagger}}_{1}\leq\norm{\mlap^{\dagger}}_{2}\leq\sqrt{n}\norm{\mlap^{\dagger}}_{1}$
\end{itemize}
\end{thm}
Note that the analysis of our applications only rely on the relationship
between $t_{pp}$ and $\norm{\mlap^{\dagger}}_{1}$ that we provide
in this section. However, we provide this full set of relationships
both for completeness and to admit clean and more easily understood
running times.

Our proof of this theorem is split into several parts. First, in Section~\ref{subsec:Bounds-Between-Mixing}
we provide an operator view of mixing times, and use this to prove
the relationship between personalized PageRank mixing time and mixing
time. Then in Section~\ref{subsec:Singular-Value-Lower} we prove
upper bounds on $\norm{\mlap^{\dagger}}_{1}$, and in Section~\ref{subsec:Condition-Number-Upper}
we prove lower bounds on $\norm{\mlap^{\dagger}}_{1}$. Along the
way we provide a more fine grained understanding of personalized PageRank
mixing time (Lemma~\ref{lem:ppLeft} and Lemma~\ref{lem:cond_to_tppr})
that may be of further use. 

\subsection{Bounds Between Mixing Times\label{subsec:Bounds-Between-Mixing}}

Here we provide bounds between different possible notions of mixing
times. First we provide Lemma~\ref{lem:equiv_helper_lemma}, a general
mathematical lemma which implies that mixing time conditions are equivalent
to certain operator norm statements on matrices, e.g. $\norm{\mw^{k}p-s}_{1}\leq\frac{1}{2}$
for all $p\in\simplex$ implies $\norm{\mw-s\vec{1}^{\top}}_{1}\leq\frac{1}{2}$.
Using this we prove Lemma~\ref{lem:mix-c} which shows how the mixing
amount of a random walk can be amplified, e.g. given $t_{mix}$ we
can compute bounds on the smallest $k$ for which $\norm{\mw^{k}p-s}_{1}\leq\epsilon$,
for any small enough $\epsilon$. Finally, we conclude this section
with Lemma~\ref{lem:superMixing}, which relates $t_{mix}$ to $t_{pp}$.
\begin{lem}
\label{lem:equiv_helper_lemma} For all $\ma\in\R^{m\times n}$, $b\in\R^{n}$,
and $\alpha\geq0$, one has that 
\[
\norm{\ma p-b}_{1}\leq\alpha,\textnormal{ for all }p\in\Delta^{n},
\]
if and only if
\[
\norm{\ma-b\vec{1}^{\top}}_{1}\leq\alpha\,.
\]
\end{lem}
\begin{proof}
Suppose $\norm{\ma-b\vec{1}^{\top}}_{1}\leq\alpha$. Then, for all
$p\in\simplex^{n}$,
\[
\norm{\ma p-b}_{1}=\normFull{(\ma-b\vec{1}^{\top})p}_{1}\leq\alpha\norm p_{1}=\alpha
\]
yielding one direction of the claim.

On the other hand, suppose $\norm{\ma p-b}_{1}\leq\alpha$ for all
$p\in\simplex^{n}$. Let $x\in\R^{n}$ be arbitrary. Decompose $x$
into its positive and negative coordinates, i.e. define $x_{+},x_{-}\in\R_{\geq0}^{n}$
as the unique vectors such that $x=x_{+}-x_{-}$, and at most one
of $x_{+}$ and $x_{-}$ is non-zero in a coordinate. Now, if $\norm{x_{+}}_{1}=0$
or $\norm{x_{-}}_{1}=0$ then for some $s\in\{-1,1\}$ we have $s\cdot x/\norm x_{1}\in\simplex^{n}$
and 
\[
\normFull{\left(\ma-b\vec{1}^{\top}\right)x}_{1}=\norm x_{1}\cdot\normFull{\ma\frac{x}{\norm x_{1}}-b}_{1}\leq\norm x_{1}\cdot\alpha\,.
\]
Otherwise, by triangle inequality 
\begin{align*}
\normFull{\left(\ma-b\vec{1}^{\top}\right)x}_{1} & \leq\normFull{\left(\ma-b\vec{1}^{\top}\right)x_{+}}_{1}+\normFull{\left(\ma-b\vec{1}^{\top}\right)x_{-}}_{1}\\
 & =\norm{x_{+}}_{1}\cdot\normFull{\ma\frac{x_{+}}{\norm{x_{+}}}-b}_{1}+\norm{x_{-}}_{1}\cdot\normFull{\ma\frac{x_{-}}{\norm{x_{-}}}-b}_{1}\,.
\end{align*}

Since $\frac{x_{+}}{\norm{x_{+}}}\in\simplex^{n}$ and $\frac{x_{-}}{\norm{x_{-}}}\in\simplex^{n}$,
we have $\normFull{\ma\frac{x_{+}}{\norm{x_{+}}}-b}_{1}\leq\alpha$
and $\normFull{\ma\frac{x_{-}}{\norm{x_{-}}}-b}_{1}\leq\alpha$. Furthermore,
since $\norm x_{1}=\norm{x_{+}}_{1}+\norm{x_{-}}_{1}$ we have that
$\normFull{\left(\ma-b\vec{1}^{\top}\right)x}_{1}\leq\norm x_{1}\cdot\alpha$.
In either case, the claim holds.
\end{proof}
\begin{lem}
\label{lem:mix-c} Let $\mw\in\R_{\geq0}^{n\times n}$ be a random
walk matrix associated with a strongly connected graph. Let $s$ be
its stationary distribution, and let $t_{mix}$ be its mixing time.
Then for all $k\geq t_{mix}$ and $x\perp s$ we have 
\[
\norm{\mw^{k}x-s}_{1}\leq\left(\frac{1}{2}\right)^{\left\lfloor \frac{k}{t_{mix}}\right\rfloor }\enspace\text{ and }\enspace\norm{\mw^{k}-s\vec{1}^{\top}}_{1}\leq\left(\frac{1}{2}\right)^{\left\lfloor \frac{k}{t_{mix}}\right\rfloor }\,.
\]
\end{lem}
\begin{proof}
First note that, since $\mw s=s$, it is the case that $\mw^{k}x-s=\mw^{k}(x-s)$.
Furthermore, since $\norm{\mw}_{1}=1$, we have that 
\[
\norm{\mw^{k+1}x-s}_{1}=\norm{\mw(\mw^{k}x-s)}_{1}\leq\norm{\mw^{k}x-s}_{1}\,.
\]
Consequently, $\norm{\mw^{k}x-s}_{1}$ decreases monotonically with
$k$ and by Lemma~\ref{lem:equiv_helper_lemma} it therefore suffices
to prove the claim for $k=c\cdot t_{mix}$ for positive integers.

Next, by the definition of $t_{mix}$ and Lemma~\ref{lem:equiv_helper_lemma}
we know that $\norm{\mw-s\vec{1}^{\top}}_{1}\leq\frac{1}{2}$. Furthermore,
since $\vec{1}^{\top}\mw=\vec{1}^{\top}$ and since $\mw s=s$, for
all non-negative integers $\alpha$, we have that 
\[
(\mw^{\alpha}-s\vec{1}^{\top})(\mw-s\vec{1})=\mw^{\alpha+1}-s\vec{1}^{\top}\mw-\mw^{\alpha}s\vec{1}+s\vec{1}^{\top}s\vec{1}=\mw^{\alpha+1}-s\vec{1}^{\top}\,.
\]
By induction, we know that for all non-negative integers $c$ it is
the case that 
\[
\norm{\mw^{c\cdot t_{mix}}-s\vec{1}^{\top}}_{1}=\normFull{\left(\mw^{t_{mix}}-s\vec{1}^{\top}\right)^{c}}_{1}\leq\left(\frac{1}{2}\right)^{c}\,.
\]
Applying Lemma~\ref{lem:equiv_helper_lemma} again yields the result.
\end{proof}
We also show the following bound based on amplifying mixed walks:
\begin{lem}
\label{lem:superMixing} Let $\mw\in\R_{\geq0}^{n\times n}$ be a
random walk matrix associated with a strongly connected graph, and
let $s$ and be its stationary distribution, and let $t_{mix}$ be
its mixing time. Let $\mm_{pp(\beta)}$ be the personalized PageRank
matrix with restart probability $\beta$ associated with $\mw$. For
all $k\geq1$, $\beta\leq(k\cdot t_{mix})^{-1}$, $x\in\simplex^{n}$,
we have $\norm{\mm_{pp(\beta)}x-s}_{1}\leq3k^{-1/2}$. Consequently,
the personalized PageRank mixing time $t_{pp}$ obeys $t_{pp}\leq36\cdot t_{mix}$.
\end{lem}
\begin{proof}
First, a little algebraic manipulation reveals that for $\gamma=\frac{\beta}{2-\beta}$,
\[
\mI-(1-\beta)\mw=(2-\beta)\mI-2(1-\beta)\lazywalk=(2-\beta)\cdot\left(\mI-2\left(\frac{1-\beta}{2-\beta}\right)\lazywalk\right)=\frac{\beta}{\gamma}\left(\mI-(1-\gamma)\lazywalk\right)\,.
\]
Consequently, since $(\mI-(1-\gamma)\lazywalk)\frac{1}{\gamma}s=s$
we have 
\begin{align*}
\mm_{pp(\beta)}x-s & =\mbox{\ensuremath{\gamma\left(\mI-(1-\gamma)\lazywalk\right)}}^{-1}(x-s)=\gamma\sum_{i=0}^{\infty}(1-\gamma)^{i}\lazywalk^{i}(x-s)\,.
\end{align*}
Now, by triangle inequality, we know that
\[
\norm{\mm_{pp(\beta)}x-s}_{1}\leq\gamma\sum_{i=0}^{\infty}(1-\gamma)^{i}\norm{\mw^{i}(x-s)}_{1}\,.
\]
Since, $\norm{\mw^{i}}_{1}\leq1$ and $\norm{x-s}_{1}\leq2$, we have
that the first $\lfloor\sqrt{k}t_{mix}\rfloor$ terms sum to at most
$2\gamma\sqrt{k}t_{mix}\leq2k^{-1/2}$. Furthermore, Lemma~\ref{lem:mix-c}
implies that for all $i>\lceil\sqrt{k}t_{mix}\rceil$ we have 
\begin{eqnarray*}
\norm{\mw^{k}x-s}_{1} & \leq & \left(\frac{1}{2}\right)^{\lfloor\frac{i}{t_{mix}}\rfloor}\leq2^{-\lfloor\sqrt{k}\rfloor}\leq\frac{1}{\sqrt{k}}
\end{eqnarray*}
for $k\geq1$, and therefore the remaining terms sum to at most $k^{-1/2}$. 
\end{proof}

\subsection{Singular Value Lower Bounds \label{subsec:Singular-Value-Lower}}

Here we show how to lower bound $\norm{\mlap^{\dagger}}_{1}$ in terms
of the mixing time and personalized PageRank mixing time. First we
provide a small technical lemma, Lemma~\ref{lem:tech-large-simplex-diff},
that will be used to reason about vectors orthogonal to stationary
distribution. Then, using this result, in Lemma~\ref{lem:mixLeft}
we lower bound $\norm{\mlap^{\dagger}}_{1}$ by the mixing time, and
in Lemma~\ref{lem:ppLeft} we lower bound $\norm{\mlap^{\dagger}}_{1}$
by the personalized PageRank mixing time. 
\begin{lem}
\label{lem:tech-large-simplex-diff} For all $a,b\in\simplex^{n}$
and $\alpha\in\R$, we have that
\[
\norm{a-\alpha b}_{1}\geq\frac{1}{2}\norm{a-b}_{1}\,.
\]
\end{lem}
\begin{proof}
If $\left|1-\alpha\right|\geq\frac{1}{2}\norm{a-b}_{1}$ then 
\[
\norm{a-\alpha b}_{1}\geq\left|\sum_{i\in[n]}a_{i}-\alpha b_{i}\right|=\left|1-\alpha\right|\geq\frac{1}{2}\norm{a-b}_{1}\,.
\]
On the other hand if $\left|1-\alpha\right|\leq\frac{1}{2}\norm{a-b}_{1}$
then by reverse triangle inequality we have
\[
\norm{a-\alpha b}_{1}\geq\norm{a-b}_{1}-\norm{(\alpha-1)b}_{1}\geq\norm{a-b}_{1}-\left|1-\alpha\right|\geq\frac{1}{2}\norm{a-b}_{1}\,.
\]
\end{proof}
\begin{lem}
\label{lem:mixLeft} For a directed Laplacian $\mlap=\mI-\mw$ associated
with a strongly connected graph, the mixing time of the random walk
matrix $\mw$ satisfies
\[
\sqrt{t_{mix}}\leq16\norm{\mlap^{\dagger}}_{1}\,.
\]
\end{lem}
\begin{proof}
Consider some vector $x\in\simplex^{n}$ that has not mixed at step
$k=t_{mix}-1$, i.e. $\norm{\mw^{k}x-s}_{1}>\frac{1}{2}$. Let
\[
y\defeq\mw^{k}x-s\cdot\frac{s^{\top}\mw^{k}x}{\norm s_{2}^{2}}\,.
\]
By construction, we have that $s^{\top}y=0$ and therefore 
\begin{equation}
\frac{\norm y_{1}}{\norm{\mlap y}_{1}}\leq\max_{z\perp s\,:\,z\neq\vec{0}}\frac{\norm z_{1}}{\norm{\mlap z}_{1}}=\max_{z\neq\vec{0}}\frac{\norm{\mlap^{\dagger}z}_{1}}{\norm z_{1}}=\norm{\mlap^{\dagger}}_{1}\,.\label{eq:tmix_upbound_1}
\end{equation}
Now by Lemma~\ref{lem:tech-large-simplex-diff} and our assumption
on $x$, we know that $\norm y_{1}\geq\frac{1}{4}$, so all that remains
is to upper bound $\norm{\mlap y}_{1}$. However, since $\mlap s=\vec{0}$
we have 
\begin{align*}
\norm{\mlap y}_{1} & =\norm{(\mI-\mw)\lazywalk x}_{1}=\normFull{(\mI-\mw)\left(\frac{\mI+\mw}{2}\right)^{k}x}_{1}=\normFull{\frac{\mI-\mw}{2^{k}}\sum_{i=0}^{k}{k \choose i}\mw^{i}x}_{1}\\
 & =\frac{1}{2^{k}}\normFull{\left(\mI-\mw^{k+1}+\sum_{i=1}^{k}\left({k \choose i-1}-{k \choose i}\right)\mw^{i}\right)x}_{1}\\
 & \leq\frac{1}{2^{k}}\left(2+\sum_{i=1}^{k}\left|{k \choose i-1}-{k \choose i}\right|\right)\,.
\end{align*}
In the last line we used the fact that $\norm{\mw}_{1}=1$, and thus
$\norm{\mw^{i}x}_{1}\leq1$ for all $i\geq0$. 

Now, consider the function $f(i)\defeq{k \choose i}$. Clearly $f(i)\geq1$,
and for all $1\leq i\leq$k, $f(i)$ increases monotonically with
$i$ until it achieves its maximum value and then decreases monotonically,
and $f(i)=f(k-i)$. Consequently, 
\[
\sum_{i=1}^{k}\left|{k \choose i-1}-{k \choose i}\right|=-2+2\cdot\max_{1\leq i\leq k}\frac{k!}{i!(k-i)!}=-2+2\frac{k!}{\lfloor\frac{k}{2}\rfloor!\lceil\frac{k}{2}\rceil!}\,.
\]
Using Stirling's approximation $\sqrt{2\pi}\cdot n^{n+\frac{1}{2}}e^{-n}\leq n!\leq e\cdot n^{n+\frac{1}{2}}e^{-n}$
we have that if $k$ is even then 
\begin{align*}
\frac{k!}{\lfloor\frac{k}{2}\rfloor!\lceil\frac{k}{2}\rceil!} & \leq\frac{e\cdot k^{k+\frac{1}{2}}e^{-k}}{2\pi\cdot\left(\frac{k}{2}\right)^{k+1}e^{-k}}=\frac{e\cdot2^{k+1}}{2\pi\sqrt{k}}\,.
\end{align*}
Since this maximum binomial coefficient increases as $k$ increases,
we have that this inequality holds for all $k$ if we replace $k$
in the right hand side with $k+1$, and we have

\[
\norm{\mlap y}_{1}\leq\frac{1}{2^{k}}\left(2\cdot\frac{e\cdot2^{k+2}}{2\pi\sqrt{k+1}}\right)=\frac{4e}{\pi}\cdot\frac{1}{\sqrt{k+1}}\leq\frac{4}{\sqrt{k+1}}\,.
\]
Combining this with (\ref{eq:tmix_upbound_1}) and the facts that
$\norm y_{1}\geq\frac{1}{4}$ and $k=t_{mix}-1$ yields the result.
\end{proof}
\begin{lem}
\label{lem:ppLeft} Let $\mlap=\mI-\mw$ be a directed Laplacian associated
with a strongly connected graph, and let $\mm_{pp(\beta)}$ denote
the personalized PageRank matrix with restart probability $\beta\in[0,1]$
associated with $\mw$. For all $x\in\simplex^{n}$ we have
\[
\frac{1}{4\beta}\norm{\mm_{pp(\beta)}x-s}_{1}\leq\norm{\mlap^{\dagger}}_{1}
\]
and consequently the personalized PageRank mixing time for $\mw$,
denoted $t_{pp}$, satisfies. 
\[
\frac{t_{pp}}{8}<\norm{\mlap^{\dagger}}_{1}\enspace\text{ and }\enspace\frac{1}{4\beta}\norm{\mm_{pp(\beta)}-s\vec{1}^{\top}}_{1}\leq\norm{\mlap^{\dagger}}_{1}\,.
\]
\end{lem}
\begin{proof}
Let $y=\mm_{pp(\beta)}x$. Recall that this implies $y=\beta(\mI-(1-\beta)\mw)^{-1}x=\beta\sum_{i=0}^{\infty}(1-\beta)^{i}\mw^{i}x$,
and since $x\in\simplex^{n}$ clearly $y\in\simplex^{n}$. Now let
$z=y-\alpha s$ where $\alpha=s^{\top}y/\norm s_{2}^{2}$ so that
$s^{\top}z=0$.. Since $\norm{y-s}_{1}>\epsilon$ by Lemma~\ref{lem:tech-large-simplex-diff},
we have that $\norm z_{1}>\frac{\epsilon}{2}$. Furthermore, since
$(\mI-(1-\beta)\mw)y=\beta x$, we have $\mlap y=\beta(x-y)$. Also,
since $\mlap s=0$ we have
\[
\norm{\mlap z}_{1}=\norm{\mlap y}_{1}=\beta\norm{x-y}_{1}\leq\beta\cdot(\norm x_{1}+\norm y_{1})=2\beta.
\]
Consequently,
\[
\norm{\mlap^{\dagger}}_{1}=\max_{x\neq0}\frac{\norm{\mlap^{\dagger}x}_{1}}{\norm x_{1}}=\max_{x\perp s,x\neq0}\frac{\norm x_{1}}{\norm{\mlap x}_{1}}\geq\frac{\norm z_{1}}{\norm{\mlap z}_{1}}>\frac{\epsilon}{4\beta}\,.
\]
Using the definition of $t_{pp}$ and applying Lemma~\ref{lem:equiv_helper_lemma}
completes the results.
\end{proof}

\subsection{Condition Number Upper Bounds \label{subsec:Condition-Number-Upper}}

Here we provide various upper bounds on $\norm{\mlap^{\dagger}}_{1}$
for a directed Laplacian $\mlap=\mI-\mw$ associated with a strongly
connected directed graph. In particular we show how to upper bound
$\norm{\mlap^{\dagger}}_{1}$in terms of $t_{mix}$ and $t_{pp}$.
We achieve this by providing an even stronger lower bound on the mixing
time of random walks. First, in Lemma~\ref{lem:condition-num-upp-bound}
we show how to lower bound how well any distribution of random walk
matrices of bounded length mixes. We then show how to amplify this
bound in Lemma~\ref{lem:upper_bound_helper}, providing stronger
lower bounds. As a near immediate corollary of this lemma we upper
bound $\norm{\mlap^{\dagger}}_{1}$ in terms of the lazy random walk
mixing time (Lemma~\ref{lem:cond_to_tmix}) and the personalized
page rank mixing time (Lemma~\ref{lem:cond_to_tppr}). 
\begin{lem}
\label{lem:condition-num-upp-bound} Let $\mw\in\R_{\geq0}^{n\times n}$
be a random walk matrix associated with a strongly connected graph,
and let $\mm$ be a distribution of the walks of $\mw$ of length
at most $k$, i.e. $\mm=\sum_{i=0}^{k}\alpha_{i}\mw^{k}$ where $\alpha_{i}\geq0$,
$\sum_{i=0}^{k}\alpha_{i}=1$. If $k<\frac{\norm{\mlap^{\dagger}}}{2\sqrt{n}}$,
then$\norm{\mm-s\vec{1}^{\top}}_{1}>\frac{1}{2\sqrt{n}}$.
\end{lem}
\begin{proof}
Recall that 
\[
\norm{\mlap^{\dagger}}_{1}=\max_{x\neq0}\frac{\norm{\mlap x}_{1}}{\norm x_{1}}=\max_{x\perp s,x\neq0}\frac{\norm x_{1}}{\norm{\mlap x}_{1}}=\max_{\norm x_{1}=1,x\perp s}\norm{\mlap x}_{1}^{-1}\,.
\]
Consequently, we let $x$ be such that $\norm x_{1}=1$, $x\perp s$,
and $\norm{\mlap x}_{1}^{-1}=\norm{\mlap^{\dagger}}_{1}$. 

Now, by reverse triangle inequality and the definition of operator
norm we have
\[
\norm{\mm-s\vec{1}{}^{\top}}_{1}\geq\norm{(\mm-s\vec{1}{}^{\top})x}_{1}\geq\norm{x-s\vec{1}^{\top}x}_{1}-\norm{\mm x-x}_{1}\,.
\]
However, since $\norm x_{1}=\norm s_{1}=1$ we know that 
\[
\norm{x-s\vec{1}^{\top}x}_{1}\geq\norm{x-s\vec{1}^{\top}x}_{2}\geq\norm x_{2}\geq\frac{1}{\sqrt{n}}\norm x_{1}=\frac{1}{\sqrt{n}}.
\]
Furthermore we have that 
\[
\norm{\mm x-x}_{1}=\normFull{\sum_{i=0}^{k}\alpha_{i}(\mw^{i}x-x)}_{1}\leq\sum_{i=0}^{k}\alpha_{i}\norm{\mw^{i}x-x}_{1}\leq\max_{i\in[k]}\norm{\mw^{i}-x}_{1}
\]
and by triangle inequality and the fact that $\norm{\mw}_{1}=1$ we
have that
\[
\|\mw^{i}x-x\|_{1}\leq\|\mw x-x\|_{1}+\|\mw^{2}x-\mw x\|_{1}+\cdots+\|\mw^{k}x-\mw^{k-1}x\|_{1}\leq k\|\mw x-x\|_{1}=k\|\mlap x\|_{1}\,.
\]
Using, that $k<\frac{\norm{\mlap^{\dagger}}_{1}}{2\sqrt{n}}$ and
$\norm{\mlap x}_{1}=\norm{\mlap^{\dagger}}_{1}^{-1}$ we have

\[
\norm{\mm-s\vec{1}{}^{\top}}_{1}\geq\frac{1}{\sqrt{n}}-k\|\mlap x\|_{1}\geq\frac{1}{\sqrt{n}}-\frac{1}{2\sqrt{n}}=\frac{1}{2\sqrt{n}}\,.
\]
\end{proof}
\begin{lem}
\label{lem:upper_bound_helper} Let $\mw\in\R_{\geq0}^{n\times n}$
be a random walk matrix associated with a strongly connected graph,
and let $\mm$ be a distribution of the walks of $\mw$, i.e. $\mm=\sum_{i=0}^{\infty}\alpha_{i}\mw^{i}$
where $\alpha_{i}\geq0$ and $\sum_{i=0}^{\infty}\alpha_{i}=1$. If
$\norm{\mm-s\vec{1}^{\top}}_{1}\leq\epsilon<1$, then $\alpha_{i}>0$
for some $i\geq k$ where
\[
k=\frac{\norm{\mlap^{\dagger}}_{1}}{\log_{\epsilon}\left(\frac{1}{2\sqrt{n}}\right)\cdot2\sqrt{n}}\,.
\]
\end{lem}
\begin{proof}
Proceed by contradiction, and suppose that $\alpha_{i}=0$ for all
$i\geq k$. Then we have that $\mm=\sum_{i=0}^{k'}\alpha_{i}\mw^{i}$,
where $\sum_{i=0}^{k'}\alpha_{i}=1$ and $k'<k$. Now clearly for
all integers $a>0$, since $\mm s=s$, we have $\norm{\mm^{a}-s\vec{1}^{\top}}_{1}=\norm{(\mm-s\vec{1}^{\top})^{a}}_{1}\leq\epsilon^{a}$.
Furthermore, picking $a=\log_{\epsilon}(1/(2\sqrt{n}))$ we have that
$\norm{\mm^{a}-s\vec{1}^{\top}}\leq\frac{1}{2\sqrt{n}}$. However,
we also have $\mm^{a}=\sum_{i=0}^{ak'}\beta_{i}\mm^{i}$, where $\beta_{i}\geq0$
and $\sum_{i=0}^{ak'}\beta_{i}=1$. Consequently, by Lemma~\ref{lem:condition-num-upp-bound}
we have $ak'\geq\norm{\mlap^{\dagger}}_{1}/2\sqrt{n}$ . Since $k=\norm{\mlap^{\dagger}}_{1}/(a2\sqrt{n})$
we obtain $k'\geq k$, contradicting $k'<k$.
\end{proof}
\begin{lem}
\label{lem:cond_to_tmix} Let $\mlap=\mI-\mw\in\R^{n\times n}$ be
a directed Laplacian associated with a strongly connected graph and
let $t_{mix}$ be its mixing time, then $\norm{\mlap^{\dagger}}_{1}\leq t_{mix}\cdot4\sqrt{n}\log_{2}n$.
\end{lem}
\begin{proof}
By Lemma~\ref{lem:equiv_helper_lemma} we know $\norm{\lazywalk^{t_{mix}}-s\vec{1}^{\top}}_{1}\leq\frac{1}{2}$.
Since $\lazywalk^{t_{mix}}$ is a distribution of the walks of $\mw$
of length at most $k$, by Lemma~\ref{lem:upper_bound_helper} we
know that $t_{mix}\geq k$ where 
\[
k=\frac{\norm{\mlap^{\dagger}}_{1}}{\log_{1/2}\left(\frac{1}{2\sqrt{n}}\right)\cdot2\sqrt{n}}=\frac{\norm{\mlap^{\dagger}}_{1}}{(1+\log_{2}\sqrt{n})\cdot2\sqrt{n}}\,.
\]
Since when $n=1$ we have $\mlap=0$ and $\mw-s\vec{1}^{\top}=0$
making the lemma trivially true and when $n\geq2$ we have that $1+\log_{2}(\sqrt{n})\leq2\log_{2}n$
the result follows.
\end{proof}
\begin{lem}
\label{lem:cond_to_tppr} Let $\mlap=\mI-\mw\in\R^{n\times n}$ be
a directed Laplacian associated with a strongly connected graph, and
let $\mm_{pp(\beta)}$ be the personalized PageRank matrix with restart
probability $\beta$ associated with $\mw$. If $\norm{\mm_{pp(\beta)}x-s}_{1}\leq\frac{1}{2}$
for all $x\in\simplex$, then $\norm{\mlap^{\dagger}}_{1}\leq\frac{16\sqrt{n}\log_{2}n}{\beta}$;
consequently, $\norm{\mlap^{\dagger}}_{1}\leq t_{pp}\cdot16\sqrt{n}\log_{2}n$.
\end{lem}
\begin{proof}
By Lemma~\ref{lem:equiv_helper_lemma} we know that $\norm{\mm_{pp(\beta)}x-s\vec{1}^{\top}}_{1}\leq\frac{1}{2}$.
Now since $\mw^{k}-s\vec{1}^{\top}=(\mw-s\vec{1}^{\top})^{k}$ as
$\mw s=s$, we have
\[
\mm_{pp(\beta)}-s\vec{1}^{\top}=\beta\left(\mI-(1-\beta)\mw\right)^{-1}-s\vec{1}^{\top}=\beta\sum_{i=0}^{\infty}(1-\beta)^{i}\left(\mw-s\vec{1}^{\top}\right)^{i}
\]
and, for all $k$ and $x$, we have
\[
\norm{(\mm_{pp(\beta)}-s\vec{1}^{\top})x}_{1}\geq\normFull{\sum_{i=0}^{k}\beta(1-\beta)^{i}\left(\mw-s\vec{1}^{\top}\right)^{i}x}_{1}-\normFull{\sum_{i=k+1}^{\infty}\beta(1-\beta)^{i}\left(\mw-s\vec{1}^{\top}\right)^{i}x}_{1}
\]
Since $\norm{\mw}_{1}=1$, we have
\[
\normFull{\sum_{i=k+1}^{\infty}\beta(1-\beta)^{i}\left(\mw-s\vec{1}^{\top}\right)^{i}x}_{1}\leq(1-\beta)^{k+1}\normFull{(\mm_{pp(\beta)}-s\vec{1}^{\top})x}_{1}
\]
and, since $\norm{\mm_{pp(\beta)}x-s\vec{1}^{\top}}_{1}\leq\frac{1}{2}$,
we have
\[
\normFull{\sum_{i=0}^{k}\beta(1-\beta)^{i}\left(\mw-s\vec{1}^{\top}\right)^{i}x}_{1}\leq\left(1+(1-\beta)^{k+1}\right)\frac{1}{2}
\]
Since the geometric sum satisfies the identity, $\sum_{i=0}^{k}(1-\beta)^{i}=\frac{1-(1-\beta)^{k+1}}{\beta}$.
We have that $\mm=\frac{1}{1-(1-\beta)^{k+1}}\sum_{i=0}^{k}\beta(1-\beta)^{i}\mw^{i}$
is a distribution of walks of $\mw$ of length at most $k$ with
\[
\norm{\mm x-s\vec{1}^{\top}}_{1}\leq\frac{1+(1-\beta)^{k+1}}{1-(1-\beta)^{k+1}}\cdot\frac{1}{2}\leq\frac{1+e^{-\beta k}}{1-e^{-\beta k}}\cdot\frac{1}{2}\,.
\]
If we pick$k=\frac{1}{\beta}\ln5$ we have that $\norm{\mm-s\vec{1}^{\top}}_{1}\leq\frac{3}{4}$
and therefore, by Lemma~\ref{lem:upper_bound_helper}:
\[
\frac{1}{\beta}\ln5=k\geq\frac{\norm{\mlap^{\dagger}}_{1}}{\log_{\frac{3}{4}}\left(\frac{1}{2\sqrt{n}}\right)\cdot2\sqrt{n}}=\frac{\log_{2}\left(\frac{4}{3}\right)\norm{\mlap^{\dagger}}_{1}}{(1+\log_{2}(2\sqrt{n}))\cdot\sqrt{n}}\,.
\]
Since when $n=1$, we have $\mlap=0$ and $\mw-s\vec{1}^{\top}=0$,
making the lemma trivially true. When $n\geq2$ we have that $1+\log_{2}\sqrt{n}\leq2\log_{2}n$;
combining this with $\log_{2}\frac{4}{3}/\ln5\geq\frac{1}{4}$ completes
the result.
\end{proof}

\section{\label{sec:applications}Applications}

In this section, we use the algorithms from the previous sections\textemdash in
particular the stationary computation algorithm described in Theorem~\ref{thm:stat-from-dd}
and the solver for RCDD-linear systems from Corollary~\ref{cor:alpha-rcdd-solver}\textemdash to
efficiently solve several of problems of interest. Most of these problems
are related to computing random walk-related quantities of directed
graphs. We emphasize that unlike all prior work for directed graphs,
our results have only a polylogarithmic dependence on the condition
number\textemdash or equivalently\textemdash the mixing time. We show
how to efficiently compute 
\begin{itemize}
\item $\ma^{\dagger}b$ where $\ma$ is any row- or column-diagonally dominant
matrix and $b$ is orthogonal to the null space of $\ma$, without
any requirement that $\ma$ being strictly row-column-diagonally-dominant
or a $\mz$-matrix (\ref{subsec:app-solvers})
\item $\mlap^{\dagger}b$ where $\mlap$ is a directed Laplacian matrix
and $b$ is any vector in $\R^{n}$, including those not in the image
of $\mlap$ (\ref{subsec:app-solvers})
\item personalized PageRank (\ref{subsec:app-pagerank})
\item the mixing time of a directed graph\textemdash up to various polynomial
factors\textemdash and its stationary distribution (\ref{subsec:app-cond-est})
\item hitting times for any particular pair of vertices (\ref{subsec:app-commute})
\item escape probabilities for any particular triple of vertices (\ref{subsec:Escape-probabilities})
\item commute times between all pairs of vertices using Johnson-Lindenstrauss
sketching (\ref{subsec:app-commute})
\end{itemize}
As with Theorem~\ref{thm:stat-from-dd}, all of our routines will
utilize a solver for a RCDD linear system in a black-box manner. Therefore,
the running time bounds will also involve black-box bounds on the
costs of such a solver.

\subsection{\label{subsec:app-pagerank}Computing Personalized PageRank Vectors}

Here we show how to apply our RCDD solver to compute personalized
PageRank vectors. The algorithm we provide and analyze in this section
serves as the basis for many of the applications we provide throughout
Section~\ref{sec:applications}.

Recall that given a restart probability $\beta>0$ and a vector of
teleportation probabilities $p$, the personalized PageRank vector
of a Markov chain is defined as the vector $x$ satisfying
\[
\beta\cdot p+(1-\beta)\mw x=x\,.
\]
Rearranging terms we have that this condition is equivalent to finding
the $x$ such that

\[
\mm_{\beta}x=p_{\beta}\enspace\text{ where }\enspace\mm_{\beta}=\frac{\beta}{1-\beta}\mI+\left(\mI-\mw\right)\enspace\text{ and }\enspace p_{\beta}=\frac{\beta}{1-\beta}\cdot p\,.
\]
Consequently it suffices to solve this linear system $\mm_{\beta}x=p_{\beta}$. 

We solve this system in three steps. First we use Theorem~\ref{thm:stat-from-dd}
with Corollary~\ref{cor:alpha-rcdd-solver} to compute an approximate
Eulerian scaling of $\mI-\mw$, which we denote $s$. Under this scaling
we show that that $\mm_{\beta}\ms$ is $\alpha$-RCDD for a sufficiently
small $\alpha$. Since $\mm_{\beta}\ms$ is $\alpha$-RCDD we can
solve linear systems in it, i.e. compute a $y\approx y_{*}$where
$\mm_{\beta}\ms y_{*}=p_{\beta}$, by Corollary~\ref{cor:alpha-rcdd-solver}.
Final, we output $x=\ms y$ as our approximate personalized PageRank
vector. This algorithm is stated formally in Algorithm~\ref{alg:stationary_ppr}
and in Lemma~\ref{lem:pers-page-rank} we bound the running time
of this algorithm and the quality of $x$ as an approximation personal
PageRank vector.

\begin{algorithm2e}

\caption{Personalized PageRank Computation}\label{alg:stationary_ppr}

\SetAlgoLined

\textbf{Input: }Random walk matrix $\mw\in\R_{\geq0}^{n\times n}$,
teleportation probability $\beta\in(0,1)$, accuracy $\epsilon\in(0,1)$ 

\textbf{Notation}: Let $\mm_{\beta}=\frac{\beta}{1-\beta}\mI+\left(\mI-\mw\right)$
and let $p_{\beta}=\frac{\beta}{1-\beta}\cdot p$.

Call Theorem~\ref{thm:stat-from-dd} with $\alpha=\frac{\beta}{1-\beta}\cdot\frac{1}{3n}$
on $\mlap=\mI-\mw$ with Corollary~\ref{cor:alpha-rcdd-solver} as
the $\alpha$-RCDD Z-matrix solver.

Let $s\in\simplex^{n}$ be the approximate Eulerian scaling resulting
from invoking Theorem~\ref{thm:stat-from-dd}.

Call Corollary~\ref{cor:alpha-rcdd-solver} with $\ma=\mm_{\beta}\ms$,
$b=p_{\beta}$, and $\epsilon'=\frac{\epsilon\beta}{45\sqrt{n}}$. 

Let $y$ be the approximation to $\ma^{-1}b$ that resulted from invoking
Corollary~\ref{cor:alpha-rcdd-solver}

\textbf{Output}: $x=\ms y$ 

\end{algorithm2e}
\begin{lem}
\label{lem:pers-page-rank} (Personalized PageRank Computation) Let
$\mw\in\R_{\geq0}^{n\times n}$ be a random walk matrix with $m$
non-zero entries, let $\beta\in(0,1)$ be a restart probability, let
$p\in\R_{\geq0}^{n}$ be teleportation constants, and let $x_{*}$
be the corresponding personalized PageRank vector. Then for any $\epsilon\in(0,1)$,
Algorithm~\ref{alg:stationary_ppr} computes a vector $x$ such that
$\norm{x-x_{*}}_{2}\leq\epsilon\norm p_{2}$ with high probability
in $n$ in time $O(\tsolve\cdot(\log^{2}\frac{n}{\beta}+\log\frac{1}{\epsilon}))$.
\end{lem}
\begin{proof}
Here we analyze Algorithm~\ref{alg:stationary_ppr} step by step.
First, consider the invocation Theorem~\ref{thm:stat-from-dd} with
$\alpha=\frac{\beta}{1-\beta}\cdot\frac{1}{3n}$ on matrix $\mI-\mw$.
By Theorem~\ref{thm:stat-from-dd} we know that the approximate Eulerian
scaling $s\in\simplex^{n}$ is such that $\left(3\alpha n\cdot\mI+\left(\mI-\mw\right)\right)\ms$
is $\alpha$-RCDD and $\kappa(\ms)\leq\frac{20}{\alpha^{2}}n$. Furthermore,
since we invoke Theorem~\ref{thm:stat-from-dd} using Corollary~\ref{cor:alpha-rcdd-solver}
as the $\alpha$-RCDD Z-matrix solve we know that $\runtime$ in Theorem~\ref{thm:stat-from-dd}
is $O(\tsolve\log(\frac{n}{\alpha}))$. Since $\tsolve=\Omega(m)$
invoking Theorem~\ref{thm:stat-from-dd} takes time $O((m+\tsolve\log(\frac{n}{\alpha}))\log\alpha^{-1})=O(\tsolve\log^{2}(\frac{n}{\beta}))$.

Next, consider the invocation of Corollary~\ref{cor:alpha-rcdd-solver}
on an $\ma$ and $b$. Note that $\alpha$ was chosen precisely so
that $\left(3\alpha n\cdot\mI+\left(\mI-\mw\right)\right)\ms=\mm_{\beta}\ms$
and therefore by the reasoning in the preceding paragraph we know
that $\mm_{\beta}\ms$ is $\alpha$-RCDD and therefore we can invoke
Corollary~\ref{cor:alpha-rcdd-solver}. By Corollary~\ref{cor:alpha-rcdd-solver}
we know that $\norm{y-y_{*}}_{\md}\leq\frac{\epsilon'}{\alpha}\norm{p_{\beta}}_{\md^{-1}}$
where $y_{*}=\ma^{-1}b=\ms^{-1}\mm_{\beta}^{-1}p_{\beta}$ and $\md=(\frac{\beta}{1-\beta}\mI+\mI)\ms=\frac{1}{1-\beta}\ms$.
Moreover, we compute this approximation in time $O(\tsolve\log(1/\epsilon'))=O(\tsolve\log(\frac{n}{\epsilon\beta}))$
giving the overall running time.

Finally, consider the vector $x=\ms y$ returned by the algorithm.
Since $x_{*}=\ms y_{*}$we have
\begin{align*}
\norm{x-x_{*}}_{2} & =\norm{\ms y-\ms y_{*}}_{2}\leq\sqrt{\max_{i\in[n]}\ms_{ii}}\cdot\norm{y-y_{*}}_{\ms}=\sqrt{(1-\beta)\cdot\max_{i\in[n]}\ms_{ii}}\cdot\norm{y-y_{*}}_{\md}\\
 & =\frac{\epsilon'}{\alpha}\sqrt{(1-\beta)\cdot\max_{i\in[n]}\ms_{ii}}\cdot\norm{p_{\beta}}_{\md^{-1}}\leq\frac{\epsilon'}{\alpha}(1-\beta)\sqrt{\kappa(\ms)}\norm{p_{\beta}}_{2}=\frac{\epsilon'}{\alpha}\beta\sqrt{\kappa(\ms)}\norm p_{2}\,.
\end{align*}
Now, using that $\alpha\geq\frac{\beta}{3n}$, $\kappa(\ms)\leq\frac{20}{\alpha^{2}}n$,
and our choice of $\epsilon'$ the result follows. 
\end{proof}
With this routine, we can also compute a good approximation of the
stationary distribution via a sufficiently large $\beta$.

\subsection{Condition Number and Stationary Distribution Estimation \label{subsec:app-cond-est}}

Here we show how to use the PageRank computation algorithm from Section~\ref{subsec:app-pagerank}
to estimate various ``condition number'' like quantities associated
with a directed Laplacian or a random walk matrix. We use these quantities
to parameterize the running times of the algorithms presented throughout
the remainder of Section~\ref{sec:applications}. In particular,
we use the relationships between the smallest singular value of $\mlap$,
the mixing time, and the personalized PageRank mixing time proved
in Section~\ref{sec:application_quantities} to estimate each up
to polynomial factors in the quantities and $n$. Furthermore, we
use the bounds of Section~\ref{sec:application_quantities} to obtain
precise estimates to the stationary distribution. The key results
of this section are summarized in the following theorem, whose proof
we defer to the end of this Section~\ref{subsec:app-cond-est}.
\begin{thm}
\label{thm:cond_and_stationary_est} Let $\mlap=\mI-\mw\in\R^{n\times n}$
be directed Laplacian associated with a strongly connected graph.
Let $\kappa$ be any of the mixing time of $\mw$, $t_{mix}$, the
personalized PageRank mixing time of $\mw$, and $\kappa(\mlap)$.
Then in time $O(\tsolve\cdot\log^{2}(n\cdot\kappa))$ with high probability
in $n$ we can compute $\tilde{\kappa}$ such that
\[
t_{pp}\leq\tilde{\kappa}\leq400n^{2}\norm{\mlap^{\dagger}}_{1}\enspace\textnormal{\textit{and}}\enspace n^{-1/2}\log_{2}^{-1/2}(n)\norm{\mlap^{\dagger}}_{1}\leq\tilde{\kappa}\leq400n^{2}\norm{\mlap^{\dagger}}_{1},
\]

and
\[
t_{pp}\leq\tilde{\kappa}\leq6400n^{2.5}\log_{2}(n)\cdot t_{pp}\,\enspace\text{\textit{and}}\enspace\frac{1}{16}n^{-1/2}\log_{2}^{-1/2}(n)\sqrt{t_{\text{mix}}}\leq\tilde{\kappa}\leq1600n^{2.5}\log_{2}(n)t_{\text{mix}}\,.
\]
Furthermore, for any $\epsilon\in(0,1)$ in time $O(\tsolve\cdot\log^{2}(n\cdot\epsilon^{-1})$
we can compute $s'$ that approximates the stationary distribution
$s$ of $\mw$, in the sense that $\norm{s'-s}_{2}\leq\epsilon\cdot\norm{\mlap^{\dagger}}_{1}$.
Consequently, if instead we spend time $O(\tsolve\cdot\log^{2}(n\cdot\kappa\cdot\epsilon^{-1})$
we can obtain $\norm{s'-s}_{2}\leq\epsilon$ and if we spend time
$O(\tsolve\cdot\log^{2}(n\cdot\kappa\cdot\kappa(\ms)\cdot\epsilon^{-1}))$
we can obtain $(1-\epsilon)\ms\preceq\ms'\preceq(1+\epsilon)\ms$.
\end{thm}
Note that we believe that significantly better estimates of the condition
number quantities are possible. However, as our algorithms depend
logarithmically on these quantities, the estimates that this theorem
provides are more than sufficient to achieve our running time bounds.

Our primary tools for proving Theorem~\ref{thm:cond_and_stationary_est}
are the PageRank computation algorithm from Section~\ref{subsec:app-pagerank}.
Lemma~\ref{lem:pers-page-rank}, and the structural results of Section~\ref{sec:application_quantities}
allow us to relate the personalized PageRank mixing time to the quantities
we wish to compute. Recall from Section~\ref{sec:application_quantities}
that for a random walk matrix $\mw\in\R^{n\times n}$ associated with
a strongly connected graph, the personalized PageRank mixing time,
$t_{pp}$ is the smallest $k$ such that for all probability vectors
$x\in\simplex^{n}$,

\[
\norm{\mm_{pp(1/k)}x-s}_{1}\leq\frac{1}{2}\enspace\text{ where }\enspace\mm_{pp(1/k)}\defeq(1/k)(\mI-(1-(1/k))\mw)^{-1}\,.
\]
Our approach is to roughly estimate this quantity and use it to compute
approximations to the other condition number quantities. Unfortunately
computing this quantity directly requires knowing $s$. We circumvent
this using the linearity of the PageRank matrix, specifically we use
that 
\[
dist(\beta)\defeq\max_{x\in\R^{n}\,:\,x\perp\vec{1},\norm x_{1}\leq1}\norm{\mm_{pp(\beta)}x}_{1}
\]
is a constant factor approximation as shown in the following lemma.
\begin{lem}
\label{lem:simplerOperator} For any $\beta\in(0,1)$,
\[
dist(\beta)\leq\max_{x\in\simplex^{n}}\norm{\mm_{pp(\beta)}x-s}_{1}\leq2\cdot dist(\beta).
\]
\end{lem}
\begin{proof}
Note that by Lemma~\ref{lem:equiv_helper_lemma} we have that 
\[
\max_{x\in\simplex}\norm{\mm_{pp(\beta)}x-s}_{1}=\norm{\mm_{pp(\beta)}x-s\vec{1}^{\top}x}_{1}=\max_{x\in\R^{n}:\,\norm x_{1}\leq1}\norm{\mm_{pp(\beta)}x-s\vec{1}^{\top}x}_{1}.
\]
This maximization can then be restricted to the set of vectors $x\perp1$,
giving (again by by Lemma~\ref{lem:equiv_helper_lemma}):

\[
\leq\max_{x\in\R^{n}:\,x\perp1\,\norm x_{1}\leq1}\norm{\mm_{pp(\beta)}x-s\vec{1}^{\top}x}_{1}=\max_{x\in\R^{n}:\,x\perp1\,\norm x_{1}\leq1}\norm{\mm_{pp(\beta)}x}_{1}=dist(\beta).
\]

To show the other direction, let $x\in\simplex^{n}$ be a maximizer
of $\norm{\mm_{pp(\beta)}x-s}_{1}$. Let $y=\frac{1}{2}(x-s)$. We
know that 
\[
\mm_{pp(\beta)}s=\beta\left(\mI-(1-\beta)\mw\right)^{-1}s=\beta\sum_{i=0}^{\infty}(1-\beta)^{i}\mw^{i}s=s
\]
and therefore $\norm{\mm_{pp(\beta)}y}_{1}=\norm{\mm_{pp(\beta)}x-s}_{1}$.
Furthermore, $\vec{1}^{\top}y=0$ and by triangle inequality $\norm y_{1}\leq1$.
Consequently, $y$ can be used in the maximization problem for $dist(\beta)$
and we have $dist(\beta)\geq\frac{1}{2}\norm{\mm_{pp(\beta)}x-s}_{1}$.
\end{proof}
Therefore the key to estimating $t_{pp}$ becomes estimating the $\ell_{1}$
operator norm of $\mm_{pp(1/k)}$ in the space orthogonal to $\vec{1}$.
In the following lemma we show how to do this through random projections.
\begin{lem}
\label{lem:randomProjection} Given a matrix $\mm\in\R^{n\times n}$,
we can determine w.h.p. in $n$ whether

\[
\max_{x\in\R^{n}:x\perp1\,\norm x_{1}\leq1}\norm{\mm x}_{1}\geq\frac{1}{2}\enspace\text{ or }\enspace\max_{x\in\R^{n}:\,x\perp1\,\norm x_{1}\leq1}\norm{\mm x}_{1}<\frac{1}{2n^{2}}
\]
by:
\end{lem}
\begin{enumerate}
\item Evaluating $\norm{\mm x}_{2}$ to additive error $\frac{1}{4n}$ for
$O(\log n)$ random unit vectors $x$ projected to be orthogonal to
the all 1-s vector, and
\item Return with the first case ($M$ having large $1\rightarrow1$ norm)
if the result is at least $\frac{1}{4n}$ for any of the vectors and
the second otherwise.
\end{enumerate}
\begin{proof}
The fact that $\norm y_{2}\leq\norm y_{1}\leq n^{1/2}\norm y_{2}$
means we have: 
\[
\max_{x\in\R^{n}:\,x\perp1\,\norm x_{1}\leq1}\norm{\mm x}_{1}\leq n^{1/2}\max_{x\in\R^{n}:\,x\perp1\,\norm x_{1}\leq1}\norm{\mm x}_{2}\leq n^{1/2}\max_{x\in\R^{n}:\,x\perp1\,\norm x_{2}\leq1}\norm{\mm x}_{2},
\]
and
\[
\max_{x\in\R^{n}:\,x\perp1\,\norm x_{1}\leq1}\norm{\mm x}_{1}\geq\max_{x\in\R^{n}:\,x\perp1\,\norm x_{1}\leq1}\norm{\mm x}_{2}\geq n^{-1/2}\max_{x\in\R^{n}:\,x\perp1\,\norm x_{1}\leq1}\norm{\mm x}_{2}
\]
Consequently, approximating the $\ell_{2}$-matrix norm gives an $n$-approximation
to the desired quantity.

Furthermore, we have that for a random unit vector $x$ orthogonal
to the all-1s vector, with constant probability its mass in the top
eigenspace is at least $1/\mathrm{rank}(\mm)\geq1/n$ (see e.g. the
proof of Theorem~7.2. of \cite{SpielmanTengSolver:journal}) and
$\norm{\mm x}_{2}$ is at least $1/\sqrt{n}$ of the maximum. So taking
$O(\log n)$ such vectors $x$ randomly and evaluating $\norm{\mm x}_{2}$
to additive error $1/4n$, and checking if $\norm{\mm x}_{2}\geq\frac{1}{4n}$
for any yields the result.
\end{proof}
Combining this estimation routine with a binary-search-like routine
leads to an estimation routine whose pseudocode is shown in Algorithm~\ref{alg:approx_tpp}.

\begin{algorithm2e}

\caption{Approximating Page Rank Mixing Time}\label{alg:approx_tpp}

\SetAlgoLined

\textbf{Input: }Random walk matrix $\mw\in\R_{\geq0}^{n\times n}$ 

Initialize $k$ = 1

\Repeat{$\widetilde{dist}\leq\frac{1}{10}n^{-3/2}$}{

Initialize estimate $\widetilde{dist}=0$.

\For{$O(\log n)$ steps}{

Pick a random unit vector $g$ orthogonal to the all-1's vector.

Use Lemma~\ref{lem:pers-page-rank} with $\epsilon=\frac{1}{100}n^{-3/2}$
to compute approximation $\widetilde{z}$ to $\mm_{pp(1/k)}g$.

Update estimate, $\widetilde{dist}=\max\{\widetilde{dist},\norm{\widetilde{z}}_{2}\}.$

}

}

\textbf{Output}: k

\end{algorithm2e}

Note that since mixing is not monotonic in $k$, this is not binary
search in the strictest sense. However, the relatively large multiplicative
error that we allow means we can still invoke Theorem~\ref{lem:superMixing}
to provide guarantees for this process.
\begin{lem}
\label{lem:conditionNumberEstimate} In $O(\tsolve\cdot\log^{2}(n\cdot\norm{\mlap^{\dagger}}_{1}))$
time Algorithm~\ref{alg:approx_tpp} computes a quantity $\tilde{\kappa}$
such that
\[
t_{pp}\leq\tilde{\kappa}\leq400n^{2}\norm{\mlap^{\dagger}}_{1}.
\]
\end{lem}
\begin{proof}
Since $\norm g_{2}\leq1,$ the errors from the approximate personalized
PageRank calls to Lemma~\ref{lem:pers-page-rank} affect the results
by at most an additive $\frac{1}{100}n^{-3/2}$. Lemma~\ref{lem:randomProjection}
gives that if $\widetilde{dist}\leq\frac{1}{10}n^{-3/2}$, we have
with high probability
\[
\max_{x\in\R^{n}\,:\,x\perp\vec{1},\norm x_{1}\leq1}\norm{\mm_{pp(1/k)}x}_{1}\leq\frac{1}{4},
\]
which by Lemma~\ref{lem:simplerOperator} implies $\max_{x\in\simplex^{n}}\norm{\mm_{pp(\beta)}x-s}_{1}\leq1/2$
with high probability, and therefore $k\geq t_{pp}$. 

On the other hand, Lemma~\ref{lem:ppLeft} gives that for any $k\geq400n^{2}\norm{\mlap^{\dagger}}_{1}$,
we have for all $x\in\Delta^{n}$,
\[
\norm{\mm_{pp(1/k)}x-s}_{1}\leq\frac{4}{k}\norm{\mlap^{\dagger}}_{1}\leq\frac{1}{100}n^{-2}.
\]
which by Lemma~\ref{lem:simplerOperator} implies $dist\leq\frac{1}{50}n^{-2}.$
\ref{lem:randomProjection} then gives that we will have $\widetilde{dist}\leq\frac{1}{10}n^{-2}$with
probability. This gives the desired approximation ratio of the output
$\tilde{k}$. Also, the fact that we only run the approximate Personalized
PageRank algorithm from Lemma~\ref{lem:pers-page-rank} with $k\leq O(n^{2}\norm{\mlap^{\dagger}}_{1})$
gives the running time bound.
\end{proof}
We now have everything we need to prove Theorem~\ref{thm:cond_and_stationary_est}.
Our bounds on the other condition number quantities come from Theorem~\ref{thm:mixingTimesSingular}
and by applying Lemma~\ref{lem:pers-page-rank} and Lemma~\ref{lem:ppLeft}
we obtain our stationary distribution estimation bounds. More precisely
we compute a personalized PageRank vector on the uniform distribution
with the restart probability small enough so that we can use Lemma~\ref{lem:ppLeft}
to reason about the error in this computation.
\begin{proof}[Proof of Theorem~\ref{thm:cond_and_stationary_est}]
 Computing the condition number estimate, $\tilde{\kappa}$, using
Lemma~\ref{lem:conditionNumberEstimate} and then applying Theorem~\ref{thm:mixingTimesSingular}
yields our bounds on $\tilde{\kappa}$. The form of the running time
follows again from the equivalence of condition numbers from Lemma~\ref{lem:conditionNumberEstimate}
and the fact that $\norm{\mlap^{\dagger}}_{1}=O(\mathrm{poly}(n,\kappa(\mlap))$.
(See Section~\ref{sec:application_quantities}). All that remains
is to compute the estimates $s'$ to the stationary distribution.
To prove the results on stationary computation, we simply compute
$y\approx\mm_{pp(\epsilon)}(\frac{1}{n}\vec{1})$ to $\epsilon$ accuracy
using Lemma~\ref{lem:pers-page-rank}. Since $\norm{\frac{1}{n}\vec{1}}_{2}=1/\sqrt{n}$,
Lemma~\ref{lem:pers-page-rank} yields that $\norm{y-\mm_{pp(\epsilon)}(\frac{1}{n}\vec{1})}_{2}\leq\epsilon/\sqrt{n}$
and computes this in time $O(\tsolve\cdot\log^{2}(\frac{n}{\epsilon})$.
However, by Lemma~\ref{lem:ppLeft} we have $\norm{\mm_{pp(\epsilon)}(\frac{1}{n}\vec{1})-s}_{1}\leq4\epsilon\norm{\mlap^{\dagger}}_{1}$
and consequently
\[
\norm{y-s}_{2}\leq\normFull{y-\mm_{pp(\epsilon)}(\frac{1}{n}\vec{1})}_{2}+\normFull{\mm_{pp(\epsilon)}(\frac{1}{n}\vec{1})-s}_{2}\leq\epsilon\left(1+4\norm{\mlap^{\dagger}}_{1}\right).
\]
Since $\norm{\mlap^{\dagger}}_{1}$ is larger than some absolute constant
we see that by scaling down $\epsilon$ by a constant we obtain our
first result on computing an approximation to $s$. The second follows
from estimating $\norm{\mlap^{\dagger}}_{1}$ from the first part
of the theorem and scaling down $\epsilon$ by that. The third follows
from the fact that if we have $\norm{y-s}_{\infty}\leq\epsilon\cdot\min_{i\in[n]}s_{i}$
then $(1-\epsilon)\ms\preceq\my\preceq(1+\epsilon)\ms$ and noting
that we can make sure this hold by seeing if $\norm{y-s}_{2}\leq\frac{\epsilon}{2}\cdot\min_{i\in[y]}y_{i}$
which must happen when $\norm{y-s}_{2}\leq\frac{\epsilon}{8}\cdot\min_{i\in[n]}s_{i}.$
The form of the running time follows from the fact that $\max_{i\in[n]}s_{i}\geq\frac{1}{n}$
and thus $1/\min_{i\in[n]}s_{ii}\leq n\cdot\kappa(\ms)$.
\end{proof}

\subsection{\label{subsec:app-solvers}Solving More General Linear Systems}

Corollary~\ref{cor:alpha-rcdd-solver} allows one to solve linear
systems involving RCDD matrices, but requires the matrix to be strictly
RCDD and have non-positive off-diagonals. Furthermore, it may return
solutions with components in the null space of the matrix, while several
of our subsequent applications rely on matrix-vector multiplications
involving the pseudoinverse, specifically, computing vectors of the
form $\mlap^{\dagger}b$.

In this section we give removes these restrictions, at the cost of
introducing a dependence on the condition number of the matrix. These
proofs can be viewed as more involved analogs of the double-cover
reduction from solving SDDM matrices to M matrices, and projecting
away against the null space. The main complications are the much more
complicated error propagation behaviors in directed settings. The
following two Lemmas are crucial for handling these errors.
\begin{lem}
\label{lem:breakup-ineq} Let $u,v$ be any vectors and $\alpha,\beta$
any scalars. Then $\norm{\alpha u-\beta v}_{2}\leq|\alpha-\beta|\cdot\norm u_{2}+\norm{u-v}_{2}\cdot|\beta|$.
\end{lem}
\begin{proof}
We have $\norm{\alpha u-\beta v}_{2}=\norm{(\alpha-\beta)u+(u-v)\beta}_{2}\leq|\alpha-\beta|\cdot\norm u_{2}+\norm{u-v}_{2}\cdot|\beta|$
\end{proof}
\begin{lem}
\label{lem:comboUpper} Let $\mm=\alpha\mb+\ma$ be an invertible
matrix with $\mm y=b$ and $x=\ma^{\dagger}b$. Suppose $b\in\im(\ma)$.
Then, $\norm{y-x}_{2}\leq\alpha\cdot\|\mm^{-1}\|_{2}\cdot\|\mb\|_{2}\cdot\|x\|_{2}$.
In particular, this implies $\|y\|_{2}\leq(1+\alpha\cdot\|\mm^{-1}\|_{2}\cdot\|\mb\|_{2})\cdot\|x\|_{2}\leq(1+\alpha\cdot\|\mm^{-1}\|_{2}\cdot\|\mb\|_{2})\cdot\norm{\ma^{\dagger}}_{2}\cdot\|b\|_{2}$.
\end{lem}
\begin{proof}
Expanding yields that 
\[
\norm{y-x}_{2}=\norm{\mm^{-1}b-\ma^{\dagger}b}_{2}=\norm{\mm^{-1}(\mI-\mm\ma^{\dagger})b}_{2}=\alpha\norm{\mm^{-1}\mb\ma^{\dagger}b}_{2}=\alpha\norm{\mm^{-1}\mb x}_{2}
\]
where in the second-to-last step, we used that $\mm=\alpha\mb+\ma$
and $\ma\ma^{\dagger}b=b$ for $b\in\im(\ma)$ .
\end{proof}
\begin{thm}
(Arbitrary RCDD system solver)\label{thm:0-rcdd-solver} Let $\ma$
be an invertible $n\times n$ RCDD matrix. Let $b$ be an $n$-dimensional
vector, and $x$ the solution to $\ma x=b$. Let $r,u$ be the smallest
and largest entries, respectively, on the diagonal of $\ma$ so $\kappa(\md)=u/r$.
Then for any $\epsilon>0$, one can compute, with high probability
and in $O\left(\tsolve\log\left(\frac{\kappa(\ma)\cdot\kappa(\md)}{\epsilon}\right)\right)$
time, a vector $x'$ satisfying $\|x'-x\|_{2}\leq\epsilon\|x\|_{2}$.
\end{thm}
\begin{proof}
We first handle the case where the off-diagonal entries can be positive.
For an $\alpha$-RCDD matrix $\mm$, let $\mm_{+}$ denote the matrix
of positive off-diagonal entries of $\mm$, with zeros elsewhere.
Let $\mm_{-}$ be the matrix of negative off-diagonals. Consider the
matrix 
\[
\mz=\left[\begin{array}{cc}
\mdiag(\mm)+\mm_{-} & -\mm_{+}\\
-\mm_{+} & \mdiag(\mm)+\mm_{-}
\end{array}\right]
\]
and the system 
\[
\mz\left[\begin{array}{c}
x_{1}\\
x_{2}
\end{array}\right]=\left[\begin{array}{c}
b\\
-b
\end{array}\right].
\]
Note that $\mz$ is RCDD and a Z-matrix. Let $y$ be the solution
to $\mm y=b$. Then the solution $x_{1}=-x_{2}=x$ satisfies this
new system. Similarly, let $(x_{1};x_{2})$ be a solution to this
new system, where the semicolon denotes vertical concatenation. Then
$x=(x_{1}-x_{2})/2$ is a solution to the original system. Now, suppose
we have an approximate solution $(x'_{1};x'_{2})$ to the new system
in the sense of Corollary~\ref{cor:alpha-rcdd-solver}. Define $(e_{1};e_{2})=(x'_{1}-x_{1};x'_{2}-x_{2})$.
Let $\md=\mdiag(\mm)/(1+\alpha)$. Note that 
\[
2\mdiag(\mz)-\left[\begin{array}{c}
\mI\\
-\mI
\end{array}\right]\md\left[\begin{array}{cc}
\mI & -\mI\end{array}\right]=\left[\begin{array}{c}
\mI\\
\mI
\end{array}\right]\md\left[\begin{array}{cc}
\mI & \mI\end{array}\right]\succeq\mzero
\]
and consequently, the approximate solution $x'=\frac{x'_{1}-x'_{2}}{2}$
to $\ma x=b$ satisfies
\begin{align*}
\|x'-x\|_{\md} & =\frac{1}{2}\cdot\left[\begin{array}{cc}
e_{1}^{\top} & e_{2}^{\top}\end{array}\right]\left[\begin{array}{c}
\mI\\
-\mI
\end{array}\right]\md\left[\begin{array}{cc}
\mI & -\mI\end{array}\right]\left[\begin{array}{c}
e_{1}\\
e_{2}
\end{array}\right]\leq\left[\begin{array}{cc}
e_{1}^{\top} & e_{2}^{\top}\end{array}\right]\mdiag(\mz)\left[\begin{array}{c}
e_{1}\\
e_{2}
\end{array}\right]\\
 & =\|(e_{1};\,e_{2})\|_{\mdiag(\mz)}\leq\|(b;\,-b)\|_{\mdiag(\mz)^{-1}}=\|b\|_{\md^{-1},}.
\end{align*}
which is, up to a constant factor, exactly the error guarantee of
Corollary~\ref{cor:alpha-rcdd-solver}. 

With the reduction now complete, let $\md=\mdiag(\ma)$. Set $\mm=\alpha\md+\ma$
for $\alpha=\frac{\epsilon r}{4\cdot\|\ma^{-1}\|_{2}^{2}\cdot u^{2}}$.
Note that $\mm$ is $\alpha$-RCDD. Since $\mm$ is strictly RCDD,
it is invertible by Lemma~\ref{lem:invertible_B}. Thus, let $y$
be the vector satisfying $\mm y=b$. By Corollary~\ref{cor:alpha-rcdd-solver},
we can compute, with high probability and in $O\left(\tsolve\log\frac{1}{\alpha\epsilon r}\right)=O\left(\tsolve\log\left(\frac{\epsilon^{2}r^{2}}{\|\ma^{-1}\|_{2}^{2}\cdot u^{2}}\right)^{-1}\right)$
time, a vector $x'$ satisfying

\[
\|x'-y\|_{2}\|x'-y\|_{\md}\leq\frac{\alpha\epsilon r}{2\alpha}\|b\|_{\md^{-1}}\implies\|x'-y\|_{2}\leq\frac{\epsilon}{2}\|b\|_{2}.
\]

Let $x$ be the solution to $\ma x=b$. Then we have
\[
\|\ma x-\ma y\|_{2}=\|b-(b-\alpha\md y)\|_{2}=\alpha\|\md y\|_{2}\leq\alpha\cdot u\cdot\|y\|_{2}=\alpha\cdot u\cdot\|\mm^{-1}b\|_{2}\leq\alpha\cdot u\cdot\|\mm^{-1}\|_{2}\cdot\|b\|_{2}.
\]

We can bound $\|\mm^{-1}\|_{2}$ by
\begin{align*}
\|\mm^{-1}\|_{2} & =\left(\min_{v}\frac{\|\mm v\|_{2}}{\|v\|_{2}}\right)^{-1}=\left(\min_{v}\frac{\|v\|_{2}\cdot\|\mm v\|_{2}}{\|v\|_{2}^{2}}\right)^{-1}\\
 & \leq\left(\min_{v}\frac{v^{\top}[\alpha\md+\ma]v}{\|v\|_{2}^{2}}\right)^{-1}\leq\left(\alpha r+\lambda_{\text{min}}(\ma+\ma^{\top})\right)^{-1}\leq1/(\alpha r).
\end{align*}

Substituting this into Lemma~\ref{lem:comboUpper}, we get a bound
that doesn't depend on $\alpha$.

\[
\|\mm^{-1}\|_{2}\leq(1+\alpha\cdot\|\mm^{-1}\|_{2}\cdot\|\md\|_{2})\cdot\|\ma^{-1}\|_{2}\leq(1+\|\md\|_{2}/r)\cdot\|\ma^{-1}\|_{2}\leq2(u/r)\cdot\|\ma^{-1}\|_{2}
\]

Thus, by Lemma~\ref{lem:rcdd-inv-monotone}, 
\[
\|x-y\|_{2}\leq\alpha\cdot\|\ma^{-1}\|_{2}\cdot u\cdot\|\mm^{-1}\|_{2}\cdot\|b\|_{2}\leq\alpha\cdot\|\ma^{-1}\|_{2}^{2}\cdot(u^{2}/r)\cdot\|b\|_{2}=\frac{\epsilon}{2}\|b\|_{2}.
\]

By the triangle inequality, $\|x'-x\|_{2}\leq\epsilon\|b\|_{2}$.
Replacing $\epsilon$ with $\epsilon/\norm{\ma}_{2}$ completes the
proof.
\end{proof}
We now show that we can approximate $\mlap^{\dagger}b$ efficiently,
even when $b$ is not in the image of $\mlap$.
\begin{thm}
(Laplacian pseudoinverse solver)\label{thm:lap-pinv-solver}Let $\mlap=\md-\ma^{\top}$
be an $n\times n$ directed Laplacian matrix of the directed, strongly
connected graph $G$. Let $x=\mlap^{\dagger}b$ where $b\in\R^{n}$.
Then for any $0<\epsilon\leq1$, one can compute, with high probability
and in time 
\[
O\left(\tsolve\log^{2}\left(\frac{n\cdot\kappa(\md)\cdot\kappa(\mlap)}{\epsilon}\right)\right)
\]
a vector $x'$ satisfying $\|x'-x\|_{2}\leq\epsilon\norm x_{2}$.
\end{thm}
In order to prove this, we prove Theorem~\ref{thm:cond-lap-pinv-solver}
which says that one can do this if one has a reasonable upper bound
on $\kappa(\mlap)$ (or equivalently, $t_{\text{mix}}$). Theorem~\ref{thm:lap-pinv-solver}
then immediately follows from the fact that we can, in fact, compute
good upper bounds quickly via Theorem~\ref{thm:cond_and_stationary_est}.
\begin{thm}
(Conditional Laplacian solver)\label{thm:cond-lap-pinv-solver} Let
$\mlap=\md-\ma^{\top}$ be an $n\times n$ directed Laplacian associated
with a strongly connected directed graph $G$. Let 
\[
\mathcal{R}(\epsilon)\leq O\left(\tsolve\log^{2}\left(\frac{n\cdot M}{\epsilon}\right)\right)
\]
 be the cost of computing an $\epsilon$-approximation to the stationary
distribution of $\mlap$ in the sense of Theorem~\ref{thm:cond_and_stationary_est},
where $M$ is a known upper bound on $\norm{\left(\mI-\mw\right)^{\dagger}}_{2}$,
where $\mw=\ma^{\top}\md^{-1}$ is the associated random walk matrix.
Let $x=\mlap^{\dagger}b$, where $b\in\R^{n}$. Then for any $0<\epsilon\leq1$,
one can compute, with high probability and in time 
\[
O\left(\tsolve\log^{2}\left(\frac{n\cdot\kappa(\md)\cdot M}{\epsilon}\right)+\mathcal{R}\left(\frac{\epsilon}{n\cdot\kappa(\md)}\right)\right)\leq O\left(\tsolve\log^{2}\left(\frac{n\cdot\kappa(\md)\cdot M}{\epsilon}\right)\right)
\]
a vector $x'$ satisfying $\|x'-x\|_{2}\leq\epsilon\norm x_{2}$. 

Furthermore, all the intermediate Eulerian Laplacian solves required
to produce the approximate solution involve only $\mr$ for which
$\kappa(\mr+\mr^{\top}),\,\kappa(\mdiag(\mr))\leq(n\kappa(\md)M/\epsilon)^{O(1)}$.
\end{thm}
\begin{proof}
Note that since we wish to compute $x=\mlap^{\dagger}b$, we can
initially project $b$ onto $\im(\mlap)=\sspan(\vones)^{\perp}$.
Thus, we will assume for the rest of the proof that $b$ is in the
image of $\mlap$. 

We solve the equivalent linear system $(\mI-\mw)z=b$, where $\mI-\mw=\mlap\md^{-1}$,
to within accuracy $\epsilon'=\epsilon/\kappa(\md)$, in the sense
that the solution $z'$ we obtain satisfies $\norm{z'-z}_{2}\leq\epsilon'\cdot\norm z_{2}$.
At the end, we return the solution $x'=\md^{-1}z'$. The guarantee
we provide on $z'$ is equivalent to $\norm{\md^{-1}x'-\md^{-1}x}_{2}\leq\epsilon'\cdot\norm{\md^{-1}x}_{2}$,
and thus $\norm{\md^{-1}(x'-x)}_{2}\leq\epsilon'\cdot\norm{\md^{-1}}_{2}\norm x_{2}$.
We then know that $\norm{x'-x}_{2}\leq(\epsilon/\kappa(\md))\cdot\kappa(\md)\cdot\norm x_{2}\le\epsilon\norm x_{2}$.

In order to solve the system in $\mI-\mw$, we set $\mm=\alpha\mI+(\mI-\mw)$,
for $\alpha=\epsilon'/(6M\sqrt{n})$, and find an approximate solution
to the system involving $\mm$.

First we show that, having set $\alpha$ sufficiently small, the solution
to the system in $\mm$ is not very far from the true solution $z$.
Letting $\Pi$ denote the projection onto the subspace orthogonal
to the kernel of $\mI-\mw$, we upper bound $\norm{\Pi(z-\mm^{-1}b)}_{2}$.
First, using the fact that $\mm$ is invertible, we notice that $\norm{\mm^{-1}}_{2}\leq\max_{x:\norm x_{2}=1}\frac{1}{\norm{\mm x}_{2}}\leq\max_{x:\norm x_{2}=1}\frac{\sqrt{n}}{\norm{\mm x}_{1}}$.
Since $\mm$ is $\alpha$-RCDD we have $\norm{\mm x}_{1}\geq\alpha\norm x_{1}$.
Therefore $\norm{\mm^{-1}}_{2}\leq\max_{x:\norm x_{2}=1}\frac{\sqrt{n}}{\alpha\norm x_{1}}\leq\sqrt{n}/\alpha$.
By Lemma~\ref{lem:comboUpper}, we have $\|\mm^{-1}b\|_{2}\leq(1+\alpha\|\mm^{-1}\|_{2})\cdot\|(\mI-\mw)^{\dagger}b\|_{2}$,
and therefore 
\begin{eqnarray}
\norm{\mm^{-1}b}_{2} & \leq & 2\sqrt{n}\norm z_{2}\,.\label{eq:minvbd}
\end{eqnarray}
Thus we can bound
\begin{align*}
\norm{\Pi(z-\mm^{-1}b)}_{2} & =\norm{(\mI-\mw)^{\dagger}(\mI-\mw)(z-\mm^{-1}b)}_{2}\leq\norm{(\mI-\mw)^{\dagger}}_{2}\norm{b-(\mm-\alpha\mI)\mm^{-1}b}_{2}\\
 & =\norm{(\mI-\mw)^{\dagger}}_{2}\cdot\alpha\norm{\mm^{-1}b}_{2}\leq\norm{(\mI-\mw)^{\dagger}}\alpha\cdot2\sqrt{n}\cdot\norm z_{2}\leq\frac{\epsilon'}{3}\norm z_{2}\,.
\end{align*}

Next, we provide a bound on the required accuracy for solving $\mm^{-1}b$.
Using Theorem~\ref{thm:stat-from-dd}, we compute in time $O\left(\left(m+\tsolve\log\frac{n}{\alpha}\right)\cdot\log\frac{n}{\alpha}\right)=O\left(\tsolve\log^{2}\frac{nM\kappa(\md)}{\epsilon}\right)$,
a diagonal matrix $\ms\succeq\mzero$ such that $\tr(\ms)=1$, $\mm\ms$
is $\frac{\alpha}{3n}$-RCDD, and $\kappa(\ms)\leq\frac{60n^{3}}{\alpha^{2}}$.
Letting $\epsilon''=\epsilon'\cdot\alpha^{3}/(120n^{3.5})$ , and
using Corollary~\ref{cor:alpha-rcdd-solver}, we can compute, with
high probability and in $O\left(\tsolve\log\frac{1}{\epsilon''}\right)=O\left(\tsolve\log\frac{\kappa(\md)}{\epsilon}\right)$
time, a vector $x'$ satisfying
\begin{eqnarray*}
\norm{x'-(\mm\ms)^{-1}b}_{(1+\alpha)\ms} & \leq & \frac{\epsilon''}{\alpha}\Vert b\Vert_{(1+\alpha)^{-1}\ms^{-1}}
\end{eqnarray*}
 with high probability, and equivalently:
\begin{eqnarray*}
(1+\alpha)^{1/2}\norm{\ms x'-\mm^{-1}b}_{\ms^{-1}} & \leq & \frac{\epsilon''}{\alpha(1+\alpha)^{1/2}}\Vert b\Vert_{\ms^{-1}}\,.
\end{eqnarray*}
Plugging the bound on $\kappa(\ms)$, this yields:
\begin{eqnarray*}
\norm{\ms x'-\mm^{-1}b}_{2} & \leq & \frac{\epsilon''\kappa(\ms)}{\alpha(1+\alpha)}\Vert b\Vert_{2}\leq\frac{60n^{3}\epsilon''}{\alpha^{3}(1+\alpha)}\norm b_{2}\,.
\end{eqnarray*}
Since $b$ is in the image of $\mlap$, therefore in the image of
$\mI-\mw$, and using $\norm{\mI-\mw}_{2}\leq2\sqrt{n}$, we obtain:
\begin{eqnarray}
\norm{\ms x'-\mm^{-1}b}_{2} & \leq & \frac{60n^{3}\epsilon''}{\alpha^{3}(1+\alpha)}\norm{(\mI-\mw)(\mI-\mw)^{\dagger}b}_{2}\leq\frac{120n^{3.5}\epsilon''}{\alpha^{3}(1+\alpha)}\norm{(\mI-\mw)^{\dagger}b}_{2}\leq\frac{\epsilon'}{3}\norm z_{2}\,.\label{eq:firstbd}
\end{eqnarray}
Using triangle inequality we then obtain that
\begin{equation}
\norm{\Pi(\ms x'-z)}_{2}\leq\norm{\Pi(\ms x'-\mm^{-1}b)}_{2}+\norm{\Pi(z-\mm^{-1}b)}_{2}\leq\frac{2\epsilon'}{3}\norm z_{2}\,.\label{eq:goodbd}
\end{equation}
So, if we ignore the kernel of $\mI-\mw$, our solution $\ms x'$
is close to $z$. The rest of the proof is devoted to showing how
to project off something approximating the kernel of $\mI-\mw$ in
order to obtain a vector that unconditionally approximates $z=(\mI-\mw)^{\dagger}b$.
Using Theorem~\ref{thm:cond_and_stationary_est}, we compute in time
$\mathcal{R}\left(\epsilon'/(32n)\right)=\mathcal{R}\left(\epsilon/(32n\cdot\kappa(\md))\right)$,
an approximate stationary vector $s'$ satisfying $\|s'-s^{*}\|_{2}\leq\frac{\epsilon'}{32n}$,
where $s^{*}$ is the true stationary. We will prove that if we approximate
the true stationary sufficiently well, then the approximate projection
$\Pi'\ms x'$, where $\Pi'=\mI-\frac{s'}{\norm{s'}_{2}}\left(\frac{s'}{\norm{s'}_{2}}\right){}^{\top}$,
is unconditionally close to $z$.

First we bound
\begin{align*}
\norm{\Pi-\Pi'}_{2} & =\normFull{\left(\mI-\frac{s^{*}}{\norm{s^{*}}_{2}}\left(\frac{s^{*}}{\norm{s^{*}}_{2}}\right)^{\top}\right)-\left(\mI-\frac{s'}{\norm{s'}_{2}}\left(\frac{s'}{\norm{s'}_{2}}\right)^{\top}\right)}_{2}\\
 & =\normFull{\frac{s^{*}}{\norm{s^{*}}_{2}}\left(\frac{s^{*}}{\norm{s^{*}}_{2}}\right)^{\top}-\frac{s'}{\norm{s'}_{2}}\left(\frac{s'}{\norm{s'}_{2}}\right)^{\top}}_{2}\,.
\end{align*}

Using Lemma~\ref{lem:breakup-ineq}, we see that $\left\Vert \frac{s^{*}}{\norm{s^{*}}_{2}}-\frac{s'}{\norm{s'}_{2}}\right\Vert _{2}\leq2\cdot\frac{\norm{s'-s^{*}}_{2}}{\norm{s^{*}}_{2}}\leq\frac{\epsilon'}{16n}\cdot\sqrt{n}$.
Therefore we can use this, together with triangle inequality, to prove
that for any vector $y$:

\begin{align*}
\norm{(\Pi-\Pi')y}_{2} & =\normFull{\frac{s^{*}}{\norm{s^{*}}_{2}}\cdot\frac{\left(s^{*}\right)^{\top}y}{\norm{s^{*}}_{2}}-\frac{s'}{\norm{s'}_{2}}\cdot\frac{\left(s'\right)^{\top}y}{\norm{s'}_{2}}}_{2}\\
 & \leq\normFull{\frac{s'}{\norm{s'}_{2}}\cdot\left(\frac{s'^{\top}y}{\|s'\|_{2}}-\frac{\left(s^{*}\right)^{\top}y}{\|s^{*}\|_{2}}\right)}_{2}+\normFull{\left(\frac{s'}{\norm{s'}_{2}}-\frac{s^{*}}{\norm{s^{*}}_{2}}\right)\cdot\frac{\left(s^{*}\right)^{\top}y}{\norm{s^{*}}_{2}}}_{2}\\
 & \leq2\left\Vert \frac{s^{*}}{\norm{s^{*}}_{2}}-\frac{s'}{\norm{s'}_{2}}\right\Vert _{2}\norm y_{2}\leq\frac{\epsilon'}{8\sqrt{n}}\norm y_{2}\,.
\end{align*}

Therefore, using this together with (\ref{eq:firstbd}) and (\ref{eq:minvbd}),
we obtain
\begin{align*}
\norm{\Pi'\ms x'-\Pi\ms x'}_{2} & \leq\frac{\epsilon'}{8\sqrt{n}}\norm{\ms x'}_{2}\leq\frac{\epsilon'}{8\sqrt{n}}\left(\norm{\ms x'-\mm^{-1}b}_{2}+\norm{\mm^{-1}b}_{2}\right)\\
 & \leq\frac{\epsilon'}{8\sqrt{n}}\cdot\left(\frac{\epsilon'}{3}\norm z_{2}+2\sqrt{n}\norm z_{2}\right)\leq\frac{\epsilon'}{3}\norm z_{2}\,.
\end{align*}

Using this together with (\ref{eq:goodbd}), and applying the triangle
inequality, we obtain that $\norm{\Pi'\ms x'-z}_{2}\leq\epsilon'\norm z_{2}$.

Finally, as we can see from Theorem~\ref{thm:stat-from-dd}, all
the intermediate solves involve $(\alpha/3n)$-RCDD matrices for which
the condition number of their diagonals is bounded by $(n/\alpha)^{O(1)}$.
Using Lemma~\ref{lem:rcdd-to-eulerian-kappa}, we see that applying
the reduction from Theorem~\ref{thm:basic-rcdd} to these matrices
only outputs matrices whose symmetrizations have condition number
bounded by $4n^{3}\cdot3n/\alpha\cdot(n/\alpha)^{O(1)}=(n\kappa(\md)M/\epsilon)^{O(1)}$.
\end{proof}
We also note that the first part of the proof of Theorem~\ref{thm:cond-lap-pinv-solver}
also gives a solver for a more general class of linear systems\textendash arbitrary
row- or column-diagonally dominant systems. This includes directed
Laplacians for graphs that are not strongly connected. As the null
space of such matrices are more intricate, we give the following guarantee
in terms of the $\ell_{2}$ norm of the matrix-vector product.
\begin{thm}
(Diagonally dominant solver)\label{thm:dd-solver}Let $\mm$ be an
arbitrary $n\times n$ column-diagonally-dominant or row-diagonally-dominant
matrix with diagonal $\md$. Let $b\in\im(\mm)$. Then for any $0<\epsilon\leq1$,
one can compute, with high probability and in time 

\[
O\left(\tsolve\log^{2}\left(\frac{n\cdot\kappa(\md)\cdot\kappa(\mm)}{\epsilon}\right)\right)
\]
a vector $x'$ satisfying $\|\mm x'-b\|_{2}\leq\epsilon\norm b_{2}$.

Furthermore, all the intermediate Eulerian Laplacian solves required
to produce the approximate solution involve only matrices $\mr$ for
which $\kappa(\mr+\mr^{\top}),\:\kappa(\mdiag(\mr))\leq(n\kappa(\md)\kappa(\mm)/\epsilon)^{O(1)}$.
\end{thm}
\begin{proof}
The algorithm and proof are essentially the same as the first part
of that of Theorem~\ref{thm:cond-lap-pinv-solver}. 

First, let us analyze the CDD case. Suppose $K$ is a known upper
bound on $\kappa(\mm)$. We rescale $\mm$ by the diagonal, and solve
for $\widehat{\mm}=\mm\md^{-1}$ to within accuracy $\epsilon/\kappa(\md)$.
In order to do so, we consider the matrix $\tilde{\mm}$ with a slightly
increased diagonal
\[
\tilde{\mm}=\frac{1}{\alpha}\mI+\widehat{\mm},
\]
where $\alpha=\epsilon'/(6\kappa(\md)Kn)\leq\epsilon/(6\kappa^{2}(\md)Kn)$.
We notice that $\norm{(\mm\md^{-1})^{\dagger}}_{2}=\kappa(\mm\md^{-1})/\norm{\mm\md^{-1}}_{2}$,
and the term in the denominator can be lower bounded by $\norm{\mm\md^{-1}}_{1}/\sqrt{n}$.
But by definition the corresponding norm of the rescaled matrix is
at least $1$; thus $\norm{(\mm\md^{-1})^{\dagger}}_{2}\leq\kappa(\mm\md^{-1})\cdot\sqrt{n}\leq\kappa(\md)\kappa(\mm)\sqrt{n}\leq\kappa(\md)K\sqrt{n}$,
where for the last inequality we used Lemma~\ref{lem:cond-num-bd}.
So our setting of $\alpha$ corresponds to the what is required for
the analysis we have seen in Theorem~\ref{thm:cond-lap-pinv-solver}.
An approximate rescaling of it to the RCDD case is obtained by considering
the diagonal scaling $\ms$ of the matrix $\tilde{\mm}_{-}$, which
has all the off-diagonal entries replaced by negative entries with
the same magnitude. With this rescaling, solutions to the system $\tilde{\mm}x=b$
can be approximated to high accuracy using Theorem~\ref{thm:0-rcdd-solver}.
The $\ell_{2}$ norm error guarantees of the solution to $\tilde{\mm}^{\dagger}b$
follows via the exact same sequence of inequalities, namely Equations~\ref{eq:minvbd}~and~\ref{eq:firstbd}.
The second part involving the projections from the kernel is not necessary
since we measure only $\|\mm x'-b\|_{2}$. 

Finally, in order to eliminate the assumption on $K$, we can simply
run the algorithm with $K=2^{2^{0}},2^{2^{1}},\dots$ and test if
the returned solution satisfies the required bound. This adds at most
a constant overhead, since the running time of the algorithm doubles
between iterations, therefore it is dominated by the one involving
a value of $K$ that is within a polynomial factor from the true $\kappa(\mm)$.

The RDD case is treated similarly. The difference is that the initial
rescaling is done in order to produce $\widehat{\mm}=\md^{-1}\mm$,
and that the rescaling is computed with respect to $\tilde{\mm}^{\top}$.

Again, as seen in the statement of Theorem~\ref{thm:stat-from-dd},
all the intermediate solves involve $(\alpha/3n)$-RCDD matrices for
which the condition number of their diagonals is bounded by $(n/\alpha)^{O(1)}$.
Just like before, we see that by Lemma~\ref{lem:rcdd-to-eulerian-kappa},
plugging these into the reduction from Theorem~\ref{thm:basic-rcdd}
produces matrices with condition number bounded by $(n/\alpha)^{O(1)}=(n\kappa(\md)\kappa(\mm)/\epsilon)^{O(1)}$.
\end{proof}
We believe a guarantee more similar to Theorem~\ref{thm:cond-lap-pinv-solver}
can be obtained, but would require methods that better take the null
space of these matrices into account. One plausible approach for directed
Laplacians is to topologically sort the strongly connected components,
and then perform back-substitutions.

\subsection{\label{subsec:app-commute}Hitting and Commute Times}

In this section we show how to compute hitting times and commute times
for a random walk on a directed graph by solving only $O(1)$ linear
systems involving the associated directed Laplacian. We first formally
define hitting and commute times and then we provide Lemma~\ref{prop:hitting-time-rank-1}
which yields a short formula of these quantities in terms of the directed
Laplacian pseudoinverse. This suggests a natural algorithm which formally
analyze in the main result of this section, Theorem~\ref{thm:(Hitting-Time-Computation)}
to compute them, and then formally analyze the error required by this
Laplacian solve.
\begin{defn}
\label{def:hitting-time}The \textit{hitting time} $H_{uv}$ denotes
the expected number of steps a random walk starting at $u$ takes
until first reaching $v$.
\end{defn}

\begin{defn}
\label{def:commute-time}The \textit{commute time} $C_{uv}$ denotes
the expected number of steps for a random walk starting at $u$ to
reach $v$, then return to $u$. In other words, $C_{uv}=H_{uv}+H_{vu}$.
\end{defn}
These quantities can each be computed using the below identities,
which hold for Markov chains whose associated graph is strongly connected.
\begin{lem}
\label{prop:hitting-time-rank-1} \label{prop:hitting-time-inner-prod}
Let $\mw\in\R_{\geq0}^{n\times n}$ be a random walk matrix associated
to a strongly connected graph. Let $s$ be its stationary distribution.
The hitting time $H_{uv}$ satisfies
\[
H_{uv}=\left(\vec{1}-\vec{1}_{v}\cdot\frac{1}{s_{v}}\right)^{\top}(\mI-\mw)^{\dagger}\left(\vec{1}_{u}-\vec{1}_{v}\right)\,.
\]
\end{lem}
We prove this identity in Section~\ref{subsec:Hitting-Time-Formula}.
It suggests the following natural algorithm for computing hitting
times.

\begin{algorithm2e}

\caption{Hitting Time Computation}\label{alg:hitting-time}

\SetAlgoLined

\textbf{Input: }Random walk matrix $\mw\in\R_{\geq0}^{n\times n}$,
lower bound $\sigma\leq s_{\textnormal{min}}$, upper bound $M\geq t_{\textnormal{mix}}$,
vertices $u,v\in V$, accuracy $\epsilon\in(0,1)$.

Call Theorem~\ref{thm:lap-pinv-solver} with $\mlap=\mI-\mw$, $b=\vec{1}_{u}-\vec{1}_{v}$,
and $\epsilon'=\epsilon\cdot\frac{\sigma}{16\sqrt{2}Mn^{1.5}\log n}$.

Let $x'$ be the approximation to $\mlap^{\dagger}b$ that resulted
from invoking Theorem~\ref{thm:lap-pinv-solver}.

Call Theorem~\ref{thm:cond_and_stationary_est} to compute the stationary
of $\mw$ with multiplicative error $\epsilon''=\epsilon\cdot\frac{\sigma}{16\sqrt{2}Mn\log n}$.

Let $s'$ be the approximate stationary resulting from invoking Theorem~\ref{thm:cond_and_stationary_est}.

\textbf{Output}: $\tilde{H}_{uv}=\left\langle x',\vec{1}-e_{v}\cdot\frac{1}{s'_{v}}\right\rangle $
.

\end{algorithm2e}
\begin{thm}
\label{thm:(Hitting-Time-Computation)} (Hitting Time Computation)
Let $\mw\in\R_{\geq0}^{n\times n}$ be a random walk matrix with $m$
nonzero entries, associated with a strongly connected graph. Let $\sigma\leq s_{\textnormal{min}}$
be a lower bound on its minimum stationary probability, let $M\geq t_{\textnormal{mix}}$
be an upper bound on its mixing time, and let $u,v\in\left\{ 1,\dots,n\right\} $.
Then for any $\epsilon\in(0,1)$, Algorithm~\ref{alg:hitting-time}
computes $\tilde{H}_{uv}$ such that $\abs{\tilde{H}_{uv}-H_{uv}}\leq\epsilon$
in time $O(\tsolve\log^{2}(\frac{Mn}{\epsilon\sigma}))$, with high
probability in $n$.
\end{thm}
\begin{proof}
First we bound the running time. Theorem~\ref{thm:lap-pinv-solver}
gives that approximates $\mlap^{\dagger}b$ takes time $O(\tsolve\log^{2}(\frac{n\kappa(\ms)\kappa(\mI-\mw)}{\epsilon'}))=O(\tsolve\log^{2}(\frac{Mn}{\epsilon\sigma}))$.
The call to Theorem~\ref{thm:cond_and_stationary_est} produces an
approximate stationary in time $O(\tsolve\log^{2}(\frac{nM\kappa(\ms)}{\epsilon''}))=O(\tsolve\log^{2}(\frac{nM}{\epsilon\sigma}))$.
The desired bound then follows from all the other operations taking$O(n)$
time, which is sublinear in $\tsolve$.

We write the approximation error $\abs{\tilde{H}_{uv}-H_{uv}}$ explicitly,
and expand it by applying the triangle inequality. Letting $x=(\mI-\mw)^{\dagger}(\vec{1}_{u}-\vec{1}_{v})$
we obtain:
\begin{align*}
\abs{\tilde{H}_{uv}-H_{uv}}=\abs{\langle x',\vec{1}-\vec{1}_{v}\cdot\frac{1}{s'_{v}}\rangle-\langle x,\vec{1}-\vec{1}_{v}\cdot\frac{1}{s_{v}}\rangle}\leq\left(\sum_{i}\abs{x'_{i}-x_{i}}\right)+\abs{\frac{x'_{v}}{s'_{v}}-\frac{x_{v}}{s_{v}}}\\
\leq\norm{x'-x}_{1}+\abs{\frac{x'_{v}-x_{v}}{s_{v}'}}+\abs{\frac{x_{v}}{s_{v}'}-\frac{x_{v}}{s_{v}}}\leq\norm{x'-x}_{2}\sqrt{n}\cdot\left(1+\frac{1}{s_{v}'}\right)+\abs{x_{v}}\left(\frac{1}{s_{v}'}-\frac{1}{s_{v}}\right)\,.
\end{align*}

The call to Theorem~\ref{thm:lap-pinv-solver} returns $x'$ satisfying
$\Vert x-x'\Vert_{2}\leq\epsilon'\norm x_{2}$. Also, the call to
Theorem~\ref{thm:cond_and_stationary_est} returns $s'$ satisfying
$\abs{s'_{v}-s_{v}}\leq\epsilon''\cdot s_{v}$, so $\abs{\frac{1}{s_{v}'}-\frac{1}{s_{v}}}\leq\frac{2\epsilon}{s_{v}}$.
Therefore, since $\abs{x_{v}}\leq\norm x_{2}$,
\begin{align*}
\abs{\tilde{H}_{uv}-H_{uv}}\leq\norm x_{2}\cdot\left(\epsilon'\sqrt{n}\cdot\left(1+\frac{1}{\sigma}\right)+\frac{2\epsilon''}{\sigma}\right)\,.
\end{align*}
Plugging in $\norm x_{2}\leq\norm{(\mI-\mw)^{\dagger}}_{2}\norm{\vec{1}_{u}-\vec{1}_{v}}_{2}$,
we see that error satisfies
\begin{align*}
\abs{\tilde{H}_{uv}-H_{uv}} & \leq\norm{(\mI-\mw)^{\dagger}}_{2}\sqrt{2}\cdot\left(\epsilon'\sqrt{n}\cdot\left(1+\frac{1}{\sigma}\right)+\frac{2\epsilon''}{\sigma}\right)\\
 & \leq M\cdot4n\log n\cdot\sqrt{2}\cdot\left(\epsilon'\sqrt{n}\cdot\left(1+\frac{1}{\sigma}\right)+\epsilon''\frac{2}{\sigma}\right)\,,
\end{align*}
where, for the last inequality, we used the bound on $\norm{(\mI-\mw)^{\dagger}}_{2}$
from Theorem~\ref{thm:mixingTimesSingular}. Plugging in the values
for $\epsilon'$ and $\epsilon''$ yields the result.
\end{proof}
Combining the formulas for hitting times also leads to an expression
for commute time, $C_{uv}=H_{uv}+H_{vu}$. It can be computed by simply
invoking Theorem~\ref{thm:(Hitting-Time-Computation)} twice.

\subsection{Escape probabilities\label{subsec:Escape-probabilities}}

In this section we derive a simple formula for escape probabilities
in a Markov chain, and show how to compute them using our directed
Laplacian solver. We first define escape probabilities, then we show
how they can be computed using our directed Laplacian solver.
\begin{defn}
\label{def:escape-prob}Let $u$ and $v$ be two different vertices
in a Markov chain. The \textit{escape probability} $p_{w}$ denotes
the probability that starting at vertex $w$, a random walk reaches
$u$ before first reaching $v$.\footnote{One should note that escape probabilities can similarly be defined
with respect to sets of vertices, not just singleton vertices. In
particular, one can define escape probabilities just like above with
respect to two disjoint sets $S$ and $T$. However, this instance
reduces to the one above simply by contracting the sets into single
vertices.}
\end{defn}
\begin{lem}
\label{prop:escape-prob}Let $\mw\in\R_{\geq0}^{n\times n}$ be a
random walk matrix associated with a strongly connected graph. Let
$s$ be its stationary distribution, and let $u,v$ be two of its
vertices. Let $p$ be the vector of escape probabilities, where $p_{w}$
represents the probability that a random walk starting at $w$ reaches
$u$ before $v$. Then 
\[
p=\beta\left(\alpha s+(\mI-\mw)^{\dagger}(\vec{1}_{u}-\vec{1}_{v})\right)\,,
\]
where $\alpha=-e_{v}^{\top}\ms^{-1}(\mI-\mw)^{\dagger}(\vec{1}_{u}-\vec{1}_{v})$
and $\beta=\frac{1}{s_{u}(\vec{1}_{u}-\vec{1}_{v})^{\top}\ms^{-1}(\mI-\mw)^{\dagger}(\vec{1}_{u}-\vec{1}_{v})}$.
Furthermore,
\[
\abs{\alpha}\leq\frac{t_{\textnormal{mix}}}{s_{\textnormal{min}}}\cdot8\sqrt{n}\log n
\]
and
\[
\frac{1}{t_{\textnormal{mix}}\cdot8\sqrt{n}\log n}\leq\beta\leq\frac{1}{s_{\textnormal{min}}}\,.
\]
\end{lem}
The proof can be found in Section~\ref{subsec:Escape-Probabilities}.
Using it, we can compute escape probabilities in time $\tilde{O}(\tsolve)$
using the algorithm described below.

\begin{algorithm2e}

\caption{Escape Probability Computation}\label{alg:escape-prob-1}

\SetAlgoLined

\textbf{Input: }Random walk matrix $\mw\in\R_{\geq0}^{n\times n}$,
lower bound $\sigma\leq s_{\textnormal{min}}$, upper bound $M\geq t_{\textnormal{mix}}$,
vertices $u,v\in V$, accuracy $\epsilon\in(0,1)$. 

Let $\tilde{\epsilon}=\epsilon\sigma^{3}/\left(2000\cdot M^{2}\cdot n\log^{2}n\right)$.

Call Theorem~\ref{thm:lap-pinv-solver} with $\mlap=\mI-\mw$, $b=\vec{1}_{u}-\vec{1}_{v}$,
and $\epsilon'=\tilde{\epsilon}/(4\sqrt{2}\cdot n\log n\cdot M)$.

Let $\tilde{z}$ be the approximation to $\mlap^{\dagger}b$ that
resulted from invoking Theorem~\ref{thm:lap-pinv-solver}.

Call Theorem~\ref{thm:cond_and_stationary_est} to compute the stationary
of $\mw$ with additive error $\epsilon''=\sigma\tilde{\epsilon}$.

Let $\tilde{s}$ be the approximate stationary resulting from invoking
Theorem~\ref{thm:cond_and_stationary_est}.

Let $\tilde{\alpha}=-\tilde{z}_{v}/\tilde{s}_{v}$, $\tilde{\beta}=\frac{1}{\tilde{s}_{u}(\vec{1}_{u}-\vec{1}_{v})^{\top}\tilde{\ms}^{-1}\tilde{z}}$,
$\tilde{p}=\tilde{\alpha}\tilde{\beta}\tilde{s}+\tilde{\beta}\tilde{z}$.

\textbf{Output}: $\tilde{p}$

\end{algorithm2e}
\begin{thm}
(Escape Probability Computation) Let $\mw\in\R_{\geq0}^{n\times n}$
be a random walk matrix with $m$ non-zero entries, i.e. $\mI-\mw$
is a directed Laplacian, let $\sigma\leq s_{\textnormal{min}}$ be
a lower bound on the minimum stationary probability, $M\geq t_{\textnormal{mix}}$
an upper bound on the mixing time, and let $u,v\in\left\{ 1,\dots,n\right\} $.
Then for any $\epsilon\in(0,1)$, Algorithm~\ref{alg:escape-prob-1}
computes a vector $\tilde{p}$ such that $\abs{\tilde{p}_{i}-p_{i}}\leq\epsilon$
for all $i$ with high probability in $n$ in time $O(\tsolve\log^{2}(\frac{nM}{\epsilon\sigma}))$,
where $p$ is the vector of escape probabilities as defined in Definition~\ref{def:escape-prob}. 
\end{thm}
\begin{proof}
First we analyze the running time. From the description of the algorithm,
we see that the running time is dominated by the calls to Theorem~\ref{thm:lap-pinv-solver},
respectively Theorem~\ref{thm:cond_and_stationary_est}. The first
one computes $\tilde{z}$ in time time $O(\tsolve\log^{2}(\frac{n\cdot\kappa(s)\cdot\kappa(\mI-\mw)}{\epsilon'}))=O(\tsolve\log^{2}(\frac{nM}{\epsilon\sigma}))$.
The second one computes $\tilde{s}$ in time $O(\tsolve\log^{2}(\frac{nM}{\epsilon''}))=O(\tsolve\log^{2}(\frac{nM}{\epsilon\sigma}))$.
Hence the bound on running time. 

Next, we analyze the error. As seen in Lemma \ref{prop:escape-prob},
we can write the true escape probability vector as $p=\beta\left(\alpha s+z\right)$,
where $z=\left(\mI-\mw\right)^{\dagger}(e_{u}-e_{v})$. Our algorithm
uses the approximations to $s$ and $z$ that it has access to in
order to obtain $\tilde{\alpha}$ and $\tilde{\beta}$. We will show
that these approximate the original values sufficiently well to provide
additive guarantees on each element of $\tilde{p}$.

By definition, $\tilde{\epsilon}\leq\epsilon s_{\textnormal{min}}^{3}/\left(2000\cdot t_{\textnormal{mix}}^{2}\cdot n\log^{2}n\right)$.
In order to provide bounds on the error of the approximation, we will
prove that:

\begin{enumerate}
\item $\norm{\tilde{z}-z}_{\infty}\leq\tilde{\epsilon}$
\item $\abs{\tilde{s}_{i}-s_{i}}\leq\tilde{\epsilon}s_{i}$ for all $i$
\item $\abs{\alpha-\tilde{\alpha}}\leq\tilde{\epsilon}\cdot20\frac{t_{\textnormal{mix}}}{s_{\textnormal{min}}}\cdot\sqrt{n}\log n=\tilde{\epsilon}_{\alpha}$
\item $\abs{\tilde{\beta}-\beta}\leq\frac{18}{s_{\textnormal{min}}}\tilde{\epsilon}\cdot\beta=\tilde{\epsilon}_{\beta}\cdot\beta\,.$
\end{enumerate}
Indeed, combining them we obtain
\begin{align*}
\norm{\tilde{p}-p}_{\infty} & =\norm{\tilde{\alpha}\tilde{\beta}\tilde{s}+\tilde{\beta}\tilde{z}-(\alpha\beta s+\beta z)}_{\infty}\\
 & \leq\norm{\tilde{\alpha}\tilde{\beta}\tilde{s}-\alpha\beta s}_{\infty}+\norm{\tilde{\beta}\tilde{z}-\beta z}_{\infty}\\
 & \leq\norm{\tilde{\beta}\tilde{s}(\tilde{\alpha}-\alpha)}_{\infty}+\norm{(\tilde{\beta}\tilde{s}-\beta s)\alpha}_{\infty}+\norm{\tilde{\beta}(\tilde{z}-z)}_{\infty}+\norm{(\tilde{\beta}-\beta)z}_{\infty}\\
 & \leq\abs{\tilde{\beta}}\abs{\tilde{\alpha}-\alpha}\norm{\tilde{s}}_{\infty}+\abs{\alpha}\norm{\tilde{\beta}\tilde{s}-\beta s}_{\infty}+\abs{\tilde{\beta}}\norm{\tilde{z}-z}_{\infty}+(\tilde{\beta}-\beta)\norm z_{\infty}\\
 & \leq(1+\tilde{\epsilon}_{\beta})\beta\cdot\tilde{\epsilon}_{\alpha}\cdot(1+\tilde{\epsilon})+\abs{\alpha}\cdot2(\tilde{\epsilon}+\tilde{\epsilon}_{\beta})\norm{\beta s}_{\infty}+(1+\tilde{\epsilon}_{\beta})\beta\tilde{\epsilon}+\tilde{\epsilon}_{\beta}\beta\norm z_{\infty}\,.
\end{align*}

Next, we use the fact that $z=\frac{1}{\beta}p-\alpha s$, hence $\norm z_{\infty}\leq\frac{1}{\beta}\norm p_{\infty}+\abs{\alpha}\norm s_{\infty}\leq\frac{1}{\beta}+\abs{\alpha}$.
Combining with the previous bound we obtain
\begin{align*}
\norm{\tilde{p}-p}_{\infty} & \leq\tilde{\epsilon}_{\alpha}(1+\tilde{\epsilon})(1+\tilde{\epsilon}_{\beta})\beta+2(\tilde{\epsilon}+\tilde{\epsilon}_{\beta})\abs{\alpha}\beta+\tilde{\epsilon}(1+\tilde{\epsilon}_{\beta})\beta+\tilde{\epsilon}_{\beta}\beta\left(\frac{1}{\beta}+\abs{\alpha}\right)\\
 & \leq\tilde{\epsilon}_{\alpha}\cdot12\abs{\alpha}\beta\\
 & \leq\tilde{\epsilon}_{\alpha}\cdot12\cdot\left(\frac{t_{\textnormal{mix}}}{s_{\textnormal{min}}}\cdot8\sqrt{n}\log_{2}n\right)\cdot\frac{1}{s_{\textnormal{min}}}\\
 & \leq\epsilon\,.
\end{align*}

Therefore it suffices to prove statements 1\textendash 4 from before.

\begin{enumerate}
\item Theorem~\ref{thm:lap-pinv-solver} gives us that $\norm{\tilde{z}-(\mI-\mw)^{\dagger}(\vec{1}_{u}-\vec{1}_{v})}_{2}\leq\epsilon'\norm x_{2}\leq\epsilon'\cdot\frac{1}{4\sqrt{2}\cdot n\log_{2}(n)\cdot M}\cdot\norm x_{2}$.
Therefore $\norm{\tilde{z}-z}_{2}\leq\tilde{\epsilon}$.
\item Theorem~\ref{thm:cond_and_stationary_est} gives us that $\norm{\tilde{s}-s}_{\infty}\leq\norm{\tilde{s}-s}_{2}\leq\epsilon''$.
Hence $\abs{\tilde{s}_{i}-s_{i}}\leq\frac{\epsilon''}{s_{\textnormal{min}}}\cdot s_{i}\leq\tilde{\epsilon}s_{i}$
for all $i$.
\item We write $\abs{\tilde{\alpha}-\alpha}=\abs{\frac{\tilde{z}_{v}}{\tilde{s}_{v}}-\frac{z_{v}}{s_{v}}}=\abs{\frac{s_{v}}{\tilde{s}_{v}}\cdot\frac{\tilde{z}_{v}}{s_{v}}-\frac{z_{v}}{s_{v}}}\leq\abs{\frac{\tilde{z}_{v}}{s_{v}}-\frac{z_{v}}{s_{v}}}+\abs{\left(\frac{s_{v}}{\tilde{s}_{v}}-1\right)\cdot\frac{\tilde{z}_{v}}{s_{v}}}$.
Next, using the additive bound on $\tilde{z}$, respectively the multiplicative
bound on $\tilde{s}_{v}$ we see that this is further upper bounded
by $\frac{1}{s_{v}}\abs{\tilde{z}_{v}-z_{v}}+\abs{2\tilde{\epsilon}s_{v}\cdot\frac{\tilde{z}_{v}}{s_{v}}}\leq\frac{1}{s_{\textnormal{min}}}\cdot\abs{\tilde{z}_{v}-z_{v}}+2\tilde{\epsilon}\cdot\tilde{z}_{v}$.
Applying triangle inequality, we see that this is at most $\abs{\tilde{z}_{v}-z_{v}}\cdot\left(\frac{1}{s_{\textnormal{min}}}+2\tilde{\epsilon}\right)+2\tilde{\epsilon}\cdot z_{v}.$
Similarly to the proof of Lemma~\ref{prop:escape-prob}, we bound
$z_{v}\leq\norm{(\mI-\mw)^{\dagger}(\vec{1}_{u}-\vec{1}_{v})}_{1}\leq t_{\textnormal{mix}}\cdot8\sqrt{n}\log n$.
Therefore $\abs{\tilde{\alpha}-\alpha}\leq\tilde{\epsilon}\cdot\frac{3}{s_{\textnormal{min}}}+\tilde{\epsilon}\cdot t_{\textnormal{mix}}\cdot16\sqrt{n}\log n$,
which gives our bound.
\item We have $\abs{\tilde{\beta}-\beta}=\abs{\frac{1}{\tilde{s}_{u}(\vec{1}_{u}-\vec{1}_{v})^{\top}\tilde{\ms}^{-1}\tilde{z}}-\frac{1}{s_{u}(\vec{1}_{u}-\vec{1}_{v})^{\top}\ms^{-1}z}}$.
By triangle inequality, we see that this is at most $\abs{\left(\frac{1}{\tilde{s}_{u}}-\frac{1}{s_{u}}\right)\cdot\frac{1}{(\vec{1}_{u}-\vec{1}_{v})^{\top}\ms^{-1}z}}+\abs{\frac{1}{\tilde{s}_{u}}\cdot\left(\frac{1}{(\vec{1}_{u}-\vec{1}_{v})^{\top}\ms^{-1}z}-\frac{1}{(\vec{1}_{u}-\vec{1}_{v})^{\top}\tilde{\ms}^{-1}\tilde{z}}\right)}$.
Using the multiplicative bound on $\tilde{s}$ we bound it with
\[
\frac{2\tilde{\epsilon}}{s_{u}}\cdot\abs{\frac{1}{(\vec{1}_{u}-\vec{1}_{v})^{\top}\ms^{-1}z}}+\frac{1}{(1-\tilde{\epsilon})s_{u}}\cdot\abs{\frac{1}{(\vec{1}_{u}-\vec{1}_{v})^{\top}\ms^{-1}z}-\frac{1}{(\vec{1}_{u}-\vec{1}_{v})^{\top}\tilde{\ms}^{-1}\tilde{z}}}\,.
\]
Now let us bound $\abs{\frac{1}{(\vec{1}_{u}-\vec{1}_{v})^{\top}\ms^{-1}z}-\frac{1}{(\vec{1}_{u}-\vec{1}_{v})^{\top}\tilde{\ms}^{-1}\tilde{z}}}$.
By triangle inequality, it is at most
\[
\abs{\frac{1}{(\vec{1}_{u}-\vec{1}_{v})^{\top}\ms^{-1}z}-\frac{1}{(\vec{1}_{u}-\vec{1}_{v})^{\top}\ms^{-1}\tilde{z}}}+\abs{\frac{1}{(\vec{1}_{u}-\vec{1}_{v})^{\top}\ms^{-1}\tilde{z}}-\frac{1}{(\vec{1}_{u}-\vec{1}_{v})^{\top}\tilde{\ms}^{-1}\tilde{z}}}\,.
\]
Using the multiplicative bound on $\tilde{s}$ we see that this is
at most
\[
\abs{\frac{1}{(\vec{1}_{u}-\vec{1}_{v})^{\top}\ms^{-1}z}-\frac{1}{(\vec{1}_{u}-\vec{1}_{v})^{\top}\ms^{-1}\tilde{z}}}+\abs{\frac{2\tilde{\epsilon}}{(\vec{1}_{u}-\vec{1}_{v})^{\top}\ms^{-1}\tilde{z}}}\,,
\]
which, again, by triangle inequality we bound with
\[
\abs{\frac{1}{(\vec{1}_{u}-\vec{1}_{v})^{\top}\ms^{-1}z}-\frac{1}{(\vec{1}_{u}-\vec{1}_{v})^{\top}\ms^{-1}\tilde{z}}}+2\tilde{\epsilon}\left(\abs{\frac{1}{(\vec{1}_{u}-\vec{1}_{v})^{\top}\ms^{-1}z}}+\abs{\frac{1}{(\vec{1}_{u}-\vec{1}_{v})^{\top}\ms^{-1}z}-\frac{1}{(\vec{1}_{u}-\vec{1}_{v})^{\top}\ms^{-1}\tilde{z}}}\right)\,.
\]
Using $\abs{\frac{1}{a}-\frac{1}{a+b}}=\abs{\frac{b}{a(a+b)}}\leq2\abs{\frac{b}{a}}$,
if $a+b\geq1$, with $a=(\vec{1}_{u}-\vec{1}_{v})^{\top}\ms^{-1}z$
and $b=(\vec{1}_{u}-\vec{1}_{v})^{\top}\ms^{-1}(\tilde{z}-z)$ (which
satisfies the conditions since $a\geq1$ and $\abs b\leq\norm{\ms^{-1}(\vec{1}_{u}-\vec{1}_{v})}_{2}\norm{\tilde{z}-z}_{2}\leq\frac{\sqrt{2}}{s_{\textnormal{min}}}\cdot\tilde{\epsilon}\leq\frac{1}{2}$)
we obtain that the quantity above is at most
\begin{align*}
2\abs{\frac{(\vec{1}_{u}-\vec{1}_{v})^{\top}\ms^{-1}(\tilde{z}-z)}{(\vec{1}_{u}-\vec{1}_{v})^{\top}\ms^{-1}z}}(1+2\tilde{\epsilon})+\abs{\frac{2\tilde{\epsilon}}{(\vec{1}_{u}-\vec{1}_{v})^{\top}\ms^{-1}z}}\\
=\frac{1}{(\vec{1}_{u}-\vec{1}_{v})^{\top}\ms^{-1}z}\cdot\left(2\abs{(\vec{1}_{u}-\vec{1}_{v})^{\top}\ms^{-1}(\tilde{z}-z)}(1+2\tilde{\epsilon})+2\tilde{\epsilon}\right)\,.
\end{align*}
 Applying the Cauchy-Schwarz inequality: $\abs{(\vec{1}_{u}-\vec{1}_{v})^{\top}\ms^{-1}(\tilde{z}-z)}\leq\norm{\ms^{-1}(\vec{1}_{u}-\vec{1}_{v})^{\top}}_{2}\norm{\tilde{z}-z}_{2}$,
we see that the quantity is at most
\begin{align*}
\frac{1}{(\vec{1}_{u}-\vec{1}_{v})^{\top}\ms^{-1}z}\left(2\norm{\ms^{-1}(\vec{1}_{u}-\vec{1}_{v})^{\top}}_{2}\norm{\tilde{z}-z}_{2}\cdot(1+2\tilde{\epsilon})+2\tilde{\epsilon}\right) & \leq\\
\frac{1}{(\vec{1}_{u}-\vec{1}_{v})^{\top}\ms^{-1}z}\left(\frac{2\sqrt{2}}{s_{\textnormal{min}}}\cdot\tilde{\epsilon}(1+2\tilde{\epsilon})+2\tilde{\epsilon}\right) & \leq\frac{1}{(\vec{1}_{u}-\vec{1}_{v})^{\top}\ms^{-1}z}\cdot\frac{8}{s_{\textnormal{min}}}\tilde{\epsilon}\,.
\end{align*}
Plugging back into the bound for $\abs{\tilde{\beta}-\beta}$ we obtain
that
\[
\abs{\tilde{\beta}-\beta}\leq\frac{1}{s_{u}(\vec{1}_{u}-\vec{1}_{v})^{\top}\ms^{-1}z}\left(2\tilde{\epsilon}+\frac{1}{(1-\tilde{\epsilon})}\cdot\frac{8}{s_{\textnormal{min}}}\tilde{\epsilon}\right)\leq\beta\cdot\frac{18}{s_{\textnormal{min}}}\tilde{\epsilon}\,.
\]
\end{enumerate}
\end{proof}

\subsection{\label{subsec:Sketching-Commute-Times}Sketching Commute Times}

In this section, we show that we can efficiently compute an $O\left(\epsilon^{-2}\log n\right)\times n$
matrix that allows us to compute a $(1+\epsilon)$ multiplicative
approximation to the commute time between any two vertices in $O\left(\epsilon^{-2}\log n\right)$
time. This data structure is the natural generalization of a powerful
result for undirected graphs: with $O(\epsilon^{-2}\log n)$ undirected
Laplacian solves one can build an $O\left(\epsilon^{-2}\log n\right)\times n$
sized matrix that allows one to compute the effective resistance between
any pair of vertices in $O(\epsilon^{-2}\log n)$ time \cite{SpielmanS08}.

This follows from applying the technique used in \cite{SpielmanS08}
for sketching effective resistances.
\begin{lem}
\label{prop:commute-time}Let $\mw\in\R_{\geq0}^{n\times n}$ be a
random walk matrix corresponding to an irreducible Markov chain with
stationary distribution $s$. Then the commute time $C_{uv}$ satisfies
the following three identities:
\begin{eqnarray*}
C_{uv} & = & \left(\vec{1}_{u}-\vec{1}_{v}\right)^{\top}\left(\left(\mI-\mw\right)\ms\right)^{\dagger}\left(\vec{1}_{u}-\vec{1}_{v}\right)\\
 & = & \left(\vec{1}_{u}-\vec{1}_{v}\right)^{\top}\ms^{-1}\left(\mI-\mw\right)^{\dagger}\left(\vec{1}_{u}-\vec{1}_{v}\right)\\
 & = & \left(\vec{1}_{u}-\vec{1}_{v}\right)^{\top}\left(\mlap_{b}^{\dagger}\right)^{\top}\mU_{b}\mlap_{b}^{\dagger}\left(\vec{1}_{u}-\vec{1}_{v}\right)\,,
\end{eqnarray*}

where $\mlap_{b}:=(\mI-\mw)\ms$ is the Laplacian rescaled to an Eulerian
graph using the stationary $s$, and $\mU_{b}=\frac{1}{2}\left(\mlap_{b}+\mlap_{b}^{\top}\right)$.
\end{lem}
We defer the proof to Section~\ref{subsec:commute-formula} of the
Appendix. The connection to effective resistances is more apparent
by noticing that for undirected graphs\textemdash where $s=\md/2m$,
and $\mlap=(\mI-\mw)\md$\textemdash one obtains $C_{uv}=(\vec{1_{u}}-\vec{1}_{v})^{\top}((\mI-\mw)\md\cdot\frac{1}{2m})^{\dagger}(\vec{1}_{u}-\vec{1}_{v})=2m\cdot(\vec{1}_{u}-\vec{1}_{v})^{\top}((\mI-\mw)\md)^{\dagger}(\vec{1}_{u}-\vec{1}_{v})=2m\cdot\reff{uv}$,
which matches a well known identity.

Having shown that $C_{uv}$ can be defined as a quadratic form involving
a symmetric positive semi-definite matrix, we can sketch the quadratic
form using the technique from \cite{SpielmanS11}. In the remainder
of this section we show how to carry out this approach.

One should note that applying the operators $\mlap_{b}$ and $\mU_{b}$
appearing in the symmetric formula requires knowing the stationary
distribution. Since our solver provides an approximation $s'$ to
the stationary, the scaling $\mlap_{b}'=(\mI-\mw)\ms'$ will not be
Eulerian in general, and therefore its corresponding symmetrization
may not even be positive semi-definite. This is addressed by adding
extra arcs to $\mlap_{b}$ so that it becomes Eulerian. If the stationary
is computed to high enough accuracy, the additional arcs will not
change the commute time by more than a factor of $1+\epsilon/2$,
and the result follows.

For completeness, we first state the Johnson-Lindenstrauss lemma:
\begin{lem}
(Johnson-Lindenstrauss) Let $v_{1,}\dots,v_{m}\in\R^{n}$, $\epsilon>0$
and $\mq$ a random $k\times n$ matrix where each entry is random
$\pm1/\sqrt{k}$, with $k\geq24\log m/\epsilon^{2}$. Then with high
probability, $(1-\epsilon)\left\Vert v_{i}\right\Vert _{2}^{2}\leq\left\Vert \mq v_{i}\right\Vert _{2}^{2}\leq(1+\epsilon)\left\Vert v_{i}\right\Vert _{2}^{2}$,
for all $i\in\left\{ 1,\dots,m\right\} $.
\end{lem}
Then we proceed with describing the implementation of the sketch for
walk matrices with known stationary distribution. Extending it to
the case where the stationary is not known a priori will be done later,
by rescaling with an approximate stationary distribution obtained
via Theorem~\ref{thm:cond_and_stationary_est}, and perturbing the
problem so that the rescaled matrix becomes exactly balanced.

\begin{algorithm2e}

\caption{Commute Time Sketch Computation with Known Stationary}\label{alg:commute-time-sketch}

\SetAlgoLined

\textbf{Input: }Random walk matrix $\mw\in\R_{\geq0}^{n\times n}$
with stationary distribution $s$, accuracy $\epsilon\in(0,1)$ 

Let $\mlap=(\mI-\mw)\ms$ be the Eulerian directed Laplacian corresponding
to $\mw$; factor its symmetrization $\mU=\mb^{\top}\mc\mb$ where
$\mb$ is the edge-vertex incidence matrix, and $\mc$ is a diagonal
matrix.

Let $k=2400\cdot\log m/\epsilon^{2}$, and let $\mq$ be a random
$k\times n$ matrix, where each entry is random $\pm1/\sqrt{k}$. 

For $i\in\left\{ 1,\dots,k\right\} $ call Lemma~\ref{lem:weirdnormsolver}
on $\mlap$ with $b=\mb^{\top}\mc^{1/2}\left(\mq^{\top}\right)_{:,i}$
and $\epsilon'\leq\frac{\epsilon^{2}s_{\textnormal{min}}^{3}w_{\textnormal{min}}^{3}}{50n^{4}}$.

Let $\tilde{\my}$ be the $k\times n$ matrix where $\tilde{\my}_{i,:}$
is the approximation to $\mlap^{\dagger}\mb^{\top}\mc^{1/2}\left(\mq^{\top}\right)_{:,i}$
that resulted from invoking Lemma~\ref{lem:weirdnormsolver}.

\textbf{Output}: $\tilde{\my}$.

\end{algorithm2e}

Obtaining an approximation to the commute time $C_{uv}$ is done simply
by computing $\Vert\tilde{\my}(\vec{1}_{u}-\vec{1}_{v})\Vert_{2}^{2}$,
as can be seen from the following lemma.
\begin{lem}
\label{lem:sketch-exact-stationary} Let $\mw\in\R_{\geq0}^{n\times n}$
be a random walk matrix with $m$ non-zero entries and a given stationary
distribution $s$, i.e. $(\mI-\mw)\ms$ is an Eulerian directed Laplacian.
Let $s_{\textnormal{min}}$ and $w_{\textnormal{min}}$ be the minimum
stationary probability and the smallest transition probability, respectively.
Then, for any $\epsilon\in(0,1)$, Algorithm~\ref{alg:commute-time-sketch}
computes in time $O(\tsolve\cdot\log(\frac{n}{\epsilon\cdot s_{\textnormal{min}}w_{\textnormal{min}}})\cdot\frac{1}{\epsilon^{2}})$
a $O(\log n/\epsilon^{2})\times n$ matrix $\tilde{\my}$ such that,
$(1-\epsilon)C_{uv}\leq\left\Vert \tilde{\my}(\vec{1}_{u}-\vec{1}_{v})\right\Vert _{2}^{2}\leq(1+\epsilon)C_{uv}$,
with high probability in $n$.
\end{lem}
\begin{proof}
The proof involves carefully analyzing the composition of errors.
Let $\mx=\mlap^{\top}\mU^{\dagger}\mlap$, $\my=\mq\mc^{1/2}\mb\left(\mlap^{\top}\right)^{\dagger}$
and recall that, by Lemma~\ref{lem:weirdnormsolver}, $\tilde{\my}$
is an approximation to $\my$ in the sense that $\Vert\tilde{\my}_{i,:}-\my_{i,:}\Vert_{\mx}\leq\epsilon'\cdot\Vert\mb^{\top}\mc^{1/2}\left(\mq^{\top}\right)_{:,i}\Vert_{\mU^{\dagger}}$.

Using the fact that the commute time $C_{uv}$ can be written as $C_{uv}=\Vert\mc^{1/2}\mb\left(\mlap^{\top}\right)^{\dagger}\left(\vec{1}_{u}-\vec{1}_{v}\right)\Vert^{2}$
(since $\vec{1}_{u}-\vec{1}_{v}\perp\vec{1}$), and applying the JL
lemma, we have that
\[
(1-\epsilon/10)C_{uv}\leq\Vert\mq\mc^{1/2}\mb\left(\mlap^{\top}\right)^{\dagger}\left(\vec{1}_{u}-\vec{1}_{v}\right)\Vert^{2}=\norm{\my(\vec{1}_{u}-\vec{1}_{v})}_{2}^{2}\leq(1+\epsilon/10)C_{uv}\,.
\]

We can bound the error introduced by the Laplacian solver using triangle
inequality as follows:
\begin{eqnarray*}
\abs{\Vert\tilde{\my}(\vec{1}_{u}-\vec{1}_{v})\Vert_{2}-\Vert\my(\vec{1}_{u}-\vec{1}_{v})\Vert_{2}} & \leq & \Vert(\tilde{\my}-\my)(\vec{1}_{u}-\vec{1}_{v})\Vert_{2}=\sqrt{\sum_{i=1}^{k}\left(\left(\tilde{\my}_{i,:}-\my_{i,:}\right)(\vec{1}_{u}-\vec{1}_{v})\right)^{2}}\,.
\end{eqnarray*}

Next, let $\Pi=\mI-\frac{1}{n}\vec{1}\vec{1}^{\top}$ be the projection
onto the space perpendicular to $\vec{1}$. Using the fact that both
$\vec{1}_{u}-\vec{1}_{v}$ and $\left(\tilde{\my}_{i,:}-\my_{i,:}\right)\Pi$
are orthogonal to $\ker(\mx)=\vec{1}$ (this identity is immediate,
since $\mx$ corresponds to a doubly stochastic random walk matrix),
we can write:
\begin{align*}
\left(\tilde{\my}_{i,:}-\my_{i,:}\right)(\vec{1}_{u}-\vec{1}_{v})=\left(\tilde{\my}_{i,:}-\my_{i,:}\right)\Pi(\vec{1}_{u}-\vec{1}_{v}) & =\left(\tilde{\my}_{i,:}-\my_{i,:}\right)\Pi\mx^{1/2}\cdot\left(\mx^{\dagger}\right)^{1/2}\Pi(\vec{1}_{u}-\vec{1}_{v})\,,
\end{align*}

then apply the Cauchy-Schwarz inequality:
\begin{align*}
\left(\tilde{\my}_{i,:}-\my_{i,:}\right)\Pi\mx^{1/2}\cdot\left(\mx^{\dagger}\right)^{1/2}\Pi(\vec{1}_{u}-\vec{1}_{v}) & \leq\norm{\left(\tilde{\my}_{i,:}-\my_{i,:}\right)\Pi}_{\mx}\cdot\norm{\Pi\left(\vec{1}_{u}-\vec{1}_{v}\right)}_{\mx^{\dagger}}\\
 & =\norm{\tilde{\my}_{i,:}-\my_{i,:}}_{\mx}\cdot\norm{\vec{1}_{u}-\vec{1}_{v}}_{\mx^{\dagger}}\,.
\end{align*}
Combining this with the solver guarantee we obtain
\begin{eqnarray*}
\abs{\Vert\tilde{\my}(\vec{1}_{u}-\vec{1}_{v})\Vert_{2}-\Vert\my(\vec{1}_{u}-\vec{1}_{v})\Vert_{2}} & \leq & \sqrt{\sum_{i=1}^{k}\epsilon'\cdot\Vert\mb^{\top}\mc^{1/2}\left(\mq^{\top}\right)_{:,i}\Vert_{\mU^{\dagger}}^{2}\cdot\norm{\vec{1}_{u}-\vec{1}_{v}}_{\mx^{\dagger}}^{2}}\\
 & = & \norm{\vec{1}_{u}-\vec{1}_{v}}_{\mx^{\dagger}}\cdot\sqrt{\epsilon'}\cdot\sqrt{\sum_{i=1}^{k}\Vert\mc^{1/2}\mb\mU^{\dagger}\mb^{\top}\mc^{1/2}\left(\mq^{\top}\right)_{:,i}\Vert^{2}}\\
 & = & C_{uv}\cdot\sqrt{\epsilon'}\cdot\sqrt{\sum_{i=1}^{k}\Vert\mc^{1/2}\mb\mU^{\dagger}\mb^{\top}\mc^{1/2}\left(\mq^{\top}\right)_{:,i}\Vert^{2}}\,.
\end{eqnarray*}

The factor containing the square root is equal to the Frobenius norm
$\Vert\mc^{1/2}\mb\mU^{\dagger}\mb^{\top}\mc^{1/2}\mq^{\top}\Vert_{F}$,
which can be upper bounded by
\begin{eqnarray*}
\Vert\mc^{1/2}\mb\mU^{\dagger}\mb^{\top}\mc^{1/2}\mq^{\top}\Vert_{F} & \leq & \norm{\mc^{1/2}\mb}_{F}\cdot\norm{\mU^{\dagger}\mb^{\top}\mc^{1/2}}_{F}\cdot\norm{\mq^{\top}}_{F}\\
 & = & \sqrt{\tr\left(\mU\right)\cdot\tr\left(\mU^{\dagger}\right)\cdot\tr\left(\mq^{\top}\mq\right)}\\
 & \leq & \sqrt{n\cdot\left(n/\lambda_{\textnormal{min}}\left(\mU\right)\right)\cdot\tr\left(\mq^{\top}\mq\right)}\\
 & = & \sqrt{n^{3}/\lambda_{\textnormal{min}}\left(\mU\right)}\,.
\end{eqnarray*}

Then, we obtain that
\begin{align*}
\abs{\Vert\tilde{\my}(\vec{1}_{u}-\vec{1}_{v})\Vert_{2}-\Vert\my(\vec{1}_{u}-\vec{1}_{v})\Vert_{2}} & \leq C_{uv}\cdot\sqrt{\epsilon'}\cdot\sqrt{n^{3}/\lambda_{\textnormal{min}}\left(\mU\right)}\\
 & =C_{uv}\cdot\sqrt{\frac{\epsilon^{2}}{100n^{3}}\cdot2s_{\textnormal{min}}^{2}w_{\textnormal{min}}^{2}\cdot\frac{s_{\textnormal{min}}w_{\textnormal{min}}}{n}\cdot\frac{n^{3}}{\lambda_{\textnormal{min}}\left(\mU\right)}}\leq\sqrt{C_{uv}}\cdot\epsilon/10\,.
\end{align*}

For the last part we used Cheeger's inequality to bound $\lambda_{\textnormal{min}}\left(\mU\right)\geq2(s_{\textnormal{min}}w_{\textnormal{min}})^{2}$,
and Lemma~\ref{prop:(Bounding-commute-times)} to bound $C_{uv}\leq n/(s_{\textnormal{min}}w_{\textnormal{min}})$.

Therefore
\begin{align*}
\abs{\left\Vert \tilde{\my}(\vec{1}_{u}-\vec{1}_{v})\right\Vert _{2}^{2}-C_{uv}} & \leq\abs{\left(\norm{\my(\vec{1}_{u}-\vec{1}_{v})}_{2}+\Vert\tilde{\my}(\vec{1}_{u}-\vec{1}_{v})\Vert_{2}-\Vert\my(\vec{1}_{u}-\vec{1}_{v})\Vert_{2}^{2}-C_{uv}\right)}\\
 & =\left(\norm{\my(\vec{1}_{u}-\vec{1}_{v})}_{2}^{2}-C_{uv}\right)+\left(\Vert\tilde{\my}(\vec{1}_{u}-\vec{1}_{v})\Vert_{2}-\Vert\my(\vec{1}_{u}-\vec{1}_{v})\Vert_{2}\right)^{2}\\
 & +2\cdot\sqrt{C_{uv}}\cdot\left(\Vert\tilde{\my}(\vec{1}_{u}-\vec{1}_{v})\Vert_{2}-\Vert\my(\vec{1}_{u}-\vec{1}_{v})\Vert_{2}\right)\\
 & \leq C_{uv}\cdot\epsilon/10+C_{uv}\cdot\epsilon^{2}/100+2\cdot C_{uv}\cdot\epsilon/10\\
 & \leq C_{uv}\cdot\epsilon\,,
\end{align*}

and our desired guarantee follows.

The running time is dominated by the $k$ calls to Lemma~\ref{lem:weirdnormsolver},
each of them requiring time $O(\tsolve\log(1/\epsilon'))=O(\tsolve\log\frac{n}{\epsilon\cdot s_{\textnormal{min}}w_{\textnormal{min}}})$.
Hence the bound on running time.
\end{proof}
Next we show how to extend this sketching technique to general walk
matrices.

The difficulty lies in the fact that Algorithm~\ref{alg:commute-time-sketch}
requires knowing the exact stationary distribution of $\mw$. However,
our solver produces only an approximate solution. In order to alleviate
this issue, we use the approximate stationary to rescale the Laplacian,
then patch the graph by adding extra arcs in order to make it Eulerian.
The small error in the stationary approximation guarantees that patching
does not significantly change the commute times or the mixing time
of the walk matrix. Therefore we can use the patched Laplacian inside
Algorithm~\ref{alg:commute-time-sketch} to produce a sketch. The
method is underlined in the algorithm described below.

\begin{algorithm2e}

\caption{Walk Matrix with Exact Stationary}\label{alg:commute-time-approx}

\SetAlgoLined

\textbf{Input: }Random walk matrix $\mw\in\R_{\geq0}^{n\times n}$,
upper bound $M\geq t_{\textnormal{mix}}$, lower bound $\sigma\leq s_{\textnormal{min}}$,
accuracy $\epsilon\in(0,1)$.

Call Theorem \ref{thm:cond_and_stationary_est} to compute the stationary
of $\mw$ with additive error $\epsilon'=\frac{\epsilon\sigma}{4n}$.

Let $s'$ be the approximate stationary obtained from invoking Theorem
\ref{thm:cond_and_stationary_est}, and set $s''=s'+\epsilon'\cdot\vec{1}$.

Let $d=\left(\mw-\mI\right)s''$. Compute a matching $\mm$ routing
demand $d$.

Let $\tilde{s}=\left(s''+d_{>0}\right)/\norm{s''+d_{>0}}_{1}$.

Let $\tilde{\mw}=\left(\mw\ms''+\mm\right)\tilde{\ms}^{-1}$

\textbf{Output}: $\left(\tilde{\mw},\tilde{s}\right)$

\end{algorithm2e}
\begin{lem}
\label{prop:produce-exact-stationary}Let $\mw\in\R_{\geq0}^{n\times n}$
be the random walk matrix of a Markov chain associated with a strongly
connected graph. Let $M\geq t_{\textnormal{mix}}$ and $\sigma\leq s_{\textnormal{min}}$.
Then Algorithm~\ref{alg:commute-time-approx} produces a walk matrix
$\tilde{\mw}$ such that $\abs{\mw_{ij}-\tilde{\mw}_{ij}}\leq\epsilon$
for all $i,j$, along with its exact stationary distribution $\tilde{s}$
in time $O\left(\tsolve\left(\log\frac{nM}{\epsilon\sigma}\right)^{2}\right)$.
Furthermore, $s_{i}+\epsilon\sigma/2\geq\tilde{s}_{i}\geq s_{i}/(1+\epsilon\sigma)$
for all $i$.
\end{lem}
\begin{proof}
We start by bounding the $\ell_{1}$ norm of the routed demand $d$.
By construction we have $\norm d_{1}=\norm{(\mI-\mw)s''}_{1}=\norm{(\mI-\mw)(s+s''-s)}_{1}=\norm{(\mI-\mw)(s''-s)}_{1}$.
Since applying $\mw$ does not increase the $\ell_{1}$ norm of the
vector it is applied to, we see that the previous quantity can be
bounded by $2\norm{s''-s}_{1}=2\norm{\epsilon'\vec{1}+s'-s}_{1}\leq2n\epsilon'+2\norm{s'-s}_{1}$.
Plugging in the error guarantee $\norm{s'-s}_{2}\leq\epsilon'$ we
obtain $\norm{s'-s}_{1}\leq\epsilon'\sqrt{n}$. Hence, putting everything
together we get $\norm d_{1}\leq2n\epsilon'+2\epsilon'\sqrt{n}\leq4n\epsilon'$. 

Also note that $s_{i}''\tilde{s}_{i}^{-1}=\norm{s''+d_{>0}}_{1}\cdot\frac{s_{i}''}{s_{i}''+(d_{>0})_{i}}$.
Since by construction $s''\geq s$, we obtain
\begin{eqnarray*}
\frac{s''_{i}}{s''_{i}+\norm d_{1}/2} & \leq & s_{i}''\tilde{s}_{i}^{-1}\leq\norm{s+(s''-s)+d_{>0}}_{1}\,.
\end{eqnarray*}
Using the bound on $\norm d_{1}$ for the first inequality, plus the
triangle inequality for the second, we get
\begin{equation}
\frac{s''_{i}}{s''_{i}+2n\epsilon'}\leq s_{i}''\tilde{s}_{i}^{-1}\leq1+4n\epsilon'\,,\label{eq:commute_bd1}
\end{equation}
 for all $i$. Plugging in $\epsilon'=\frac{\epsilon\sigma}{4n}$
we obtain
\begin{equation}
1-\epsilon\leq s_{i}''\tilde{s}_{i}^{-1}\leq1+\epsilon/2\,.\label{eq:commute_bd2}
\end{equation}

Also, note that
\begin{eqnarray}
(\mm\tilde{\ms}^{-1})_{ij} & \leq & \norm{d_{>0}}_{1}\cdot\max_{k}\left(\norm{s''+d_{>0}}_{1}\cdot\frac{1}{s_{k}''+(d_{>0})_{k}}\right)\label{eq:commute_bd3}\\
 & \leq & 2n\epsilon'\cdot\left(\left(1+4n\epsilon'\right)\cdot\frac{1}{s_{\textnormal{min}}}\right)\leq\epsilon/2\,.\label{eq:eq:commute_bd4}
\end{eqnarray}

Combining Equations~\ref{eq:commute_bd2} and \ref{eq:commute_bd3},
we obtain that $(1-\epsilon)\mw_{ij}\leq\tilde{\mw}_{ij}\leq(1+\epsilon/2)\mw_{ij}+\epsilon/2$,
and the desired bound on $\tilde{\mw}_{ij}$ follows.

Also, by construction, we notice that $\tilde{s}_{i}\geq s''_{i}/\norm{s''+d_{>0}}_{1}\geq s_{i}/(1+4n\epsilon')=s_{i}/(1+\epsilon\sigma)$.
Similarly, since $s''\geq s$, and therefore $\tilde{s}_{i}\leq s_{i}+\norm{d_{>0}}_{1}$,
we have $\tilde{s_{i}}\leq s_{i}+2n\epsilon'\leq s_{i}+\epsilon\sigma/2$.

\begin{sloppypar}The running time is dominated by the invocation
of Theorem \ref{thm:cond_and_stationary_est}, which requires time
$O\left(\tsolve\left(\log\frac{nM}{\epsilon'}\right)^{2}\right)=O\left(\tsolve\left(\log\frac{nM}{\epsilon\sigma}\right)^{2}\right)$.\end{sloppypar}
\end{proof}
The lemma below shows a sufficient accuracy to which we need to invoke
Algorithm~\ref{alg:commute-time-approx} in order to approximately
preserve commute times. 
\begin{lem}
\label{prop:approx-commute}Let $\mw,\tilde{\mw}\in\R_{\geq0}^{n\times n}$
be random walk matrices of Markov chains associated with strongly
connected graphs, having stationaries $s$ and $\tilde{s}$, respectively.
Let $C$ and $\tilde{C}$, respectively, denote commute times for
the two walk matrices. Let $s_{\textnormal{min}}$ represent the minimum
stationary probability of $\mw$, and let $\epsilon_{0}\leq\epsilon\cdot s_{\textnormal{min}}^{2}/(256\cdot n^{3}\log^{2}n\cdot t_{\textnormal{mix}}^{2})$.
If $\abs{\mw_{ij}-\tilde{\mw}_{ij}}\leq\epsilon_{0}$ for all $i,j$,
and $\norm{s-\tilde{s}}_{\infty}\leq\epsilon_{0}$, then $\abs{\tilde{C}_{uv}-C_{uv}}\leq\epsilon$,
for all $u,v$.
\end{lem}
\begin{proof}
Recall that, by Lemma~\ref{prop:commute-time}, we can write $C_{uv}=(\vec{1}_{u}-\vec{1}_{v})^{\top}\ms^{-1}(\mI-\mw)^{\dagger}(\vec{1}_{u}-\vec{1}_{v})$,
and similarly for $\tilde{C}_{uv}$. We proceed to bounding the difference
between them as follows:
\begin{align}
\abs{C_{uv}-\tilde{C}_{uv}} & =\abs{(\vec{1}_{u}-\vec{1}_{v})^{\top}(((\mI-\mw)\ms)^{\dagger}-((\mI-\tilde{\mw})\tilde{\ms})^{\dagger})(\vec{1}_{u}-\vec{1}_{v})}\nonumber \\
 & \leq\norm{\vec{1}_{u}-\vec{1}_{v}}\norm{(((\mI-\mw)\ms)^{\dagger}-((\mI-\tilde{\mw})\tilde{\ms})^{\dagger})(\vec{1}_{u}-\vec{1}_{v})}\,.\label{eq:commute-difference}
\end{align}

Our goal will be to apply Lemma~\ref{lem:sol-perturb} with $\ma=(\mI-\mw)\ms$,
$\Delta\ma=(\mI-\tilde{\mw})\tilde{\ms}-(\mI-\mw)\ms$, $b=e_{u}-e_{v}$.
First let us bound $\norm{\ma^{\dagger}}$ and $\norm{\Delta\ma}$.

First, we bound
\begin{eqnarray*}
\norm{\ma^{\dagger}} & = & \norm{((\mI-\mw)\ms)^{\dagger}}=\max_{x:(\mI-\mw)\ms x\neq0}\frac{\norm x}{\norm{(\mI-\mw)\ms x}}\\
 & \leq & \max_{x:(\mI-\mw)\ms x\neq0}\frac{\norm{\ms^{-1}}\norm{\ms x}}{\norm{(\mI-\mw)\ms x}}=\norm{\ms^{-1}}\norm{(\mI-\mw)^{\dagger}}\\
 & \leq & 4\cdot n\log n\cdot t_{\textnormal{mix}}/s_{\textnormal{min}}\,.
\end{eqnarray*}

For the last inequality, we used the bound on $\norm{(\mI-\mw)^{\dagger}}$
from Theorem~\ref{thm:mixingTimesSingular}.

Second, we bound $\norm{\Delta\ma}$. Applying triangle inequality,
we obtain
\begin{eqnarray*}
\norm{\Delta\ma} & \leq & \norm{((\mI-\tilde{\mw})-(\mI-\mw))\tilde{\ms}}+\norm{(\mI-\mw)(\tilde{\ms}-\ms)}\\
 & \leq & \norm{\tilde{\mw}-\mw}\norm{\tilde{\ms}}+\norm{\mI-\mw}\norm{\tilde{\ms}-\ms}\,.
\end{eqnarray*}

Since we have an entry-wise bound on the difference between $\tilde{\mw}$
and $\mw$, we can upper bound $\norm{\tilde{\mw}-\mw}$ with the
Frobenius norm $\norm{\tilde{\mw}-\mw}_{F}\leq n\epsilon_{0}$. Also
$\tilde{\ms}$ is a stationary, so $\norm{\tilde{\ms}}\leq1$. Additionally,
the additive error on $\tilde{s}$, bounds $\norm{\tilde{\ms}-\ms}\leq\epsilon_{0}$.
Also $\norm{\mI-\mw}_{2}\leq\sqrt{\norm{\mI-\mw}_{1}\cdot\norm{\mI-\mw}_{\infty}}\leq\sqrt{2\cdot n}$.
Combining these, we obtain
\[
\norm{\Delta\ma}\leq n\epsilon_{0}+\sqrt{2n}\cdot\epsilon_{0}\leq\frac{1}{2\norm{\ma^{\dagger}}}\,.
\]

Therefore we can plug our bounds into in our bounds into Lemma~\ref{lem:sol-perturb},
and obtain $\frac{\norm{(\ma'^{\dagger}-\ma^{\dagger})b}}{\norm{\ma^{\dagger}b}}\leq\frac{\norm{\ma^{\dagger}}\norm{\Delta\ma}}{1-\norm{\ma^{\dagger}}\norm{\Delta\ma}}\leq2\norm{\ma^{\dagger}}\norm{\Delta\ma}$,
and equivalently $\norm{(\ma'^{\dagger}-\ma^{\dagger})b}\leq2\norm{\ma^{\dagger}}^{2}\norm{\Delta\ma}\norm b$.

Finally, combining with Equation~\ref{eq:commute-difference}, we
obtain $\abs{\tilde{C}_{uv}-C_{uv}}\leq\sqrt{2}\cdot2\norm{\ma^{\dagger}}^{2}\norm{\Delta\ma}\norm b\leq4\cdot\left(\frac{4n\log n\cdot t_{\textnormal{mix}}}{s_{\textnormal{min}}}\right)^{2}\cdot\left(n\epsilon_{0}+\sqrt{2n}\cdot\epsilon_{0}\right)\leq\epsilon$.
\end{proof}
\begin{lem}
\label{prop:(Bounding-commute-times)}(Bounding commute times) Let
$\mw\in\R_{\geq0}^{n\times n}$ be a random walk matrix corresponding
to an irreducible Markov chain, and let $s_{\textnormal{min}}$ and
$w_{\textnormal{min}}$ be its smallest stationary probability, and
smallest transition probability, respectively. Then any commute time
$C_{uv}$ satisfies $C_{uv}\leq n/(s_{\textnormal{min}}w_{\textnormal{min}})$.
\end{lem}
\begin{proof}
We will use the fact that commute times form a metric. Then any commute
time $C_{uv}$ is upper bounded by the sum of commute times over arcs
in $\mw$ with nonzero probability. Using Lemma~\ref{prop:commute-time},
we can therefore bound
\begin{eqnarray*}
C_{uv} & \leq & \sum_{i\neq j:\mw_{ij}>0}(\vec{1}_{i}-\vec{1}_{j})^{\top}\left(\mlap_{b}\mU_{b}^{\dagger}\mlap_{b}^{\top}\right)^{\dagger}(\vec{1}_{i}-\vec{1}_{j})\\
 & \leq & \frac{1}{\min_{i\neq j:\mw_{ij}>0}\abs{\mU_{b}}_{ij}}\cdot\sum_{i\neq j:\mw_{ij}>0}\left(\mlap_{b}\mU_{b}^{\dagger}\mlap_{b}^{\top}\right)^{\dagger}\bullet\left((\vec{1}_{i}-\vec{1}_{j})(\vec{1}_{i}-\vec{1}_{j})^{\top}\cdot\abs{\mU_{b}}_{ij}\right)\\
 & \leq & \frac{1}{\min_{i\neq j:\mw_{ij}>0}\abs{\mU_{b}}_{ij}}\cdot\tr\left(\left(\mlap_{b}\mU_{b}^{\dagger}\mlap_{b}^{\top}\right)^{\dagger}\mU_{b}\right)\,.
\end{eqnarray*}

From Lemma~\ref{lem:lapbounds} we have that $\mlap_{b}\mU_{b}^{\dagger}\mlap_{b}^{\top}\succeq\mU_{b}$,
and therefore $\left(\mlap_{b}\mU_{b}^{\dagger}\mlap_{b}^{\top}\right)^{\dagger}\preceq\mU_{b}^{\dagger}$.
Hence $\tr\left(\left(\mlap_{b}\mU_{b}^{\dagger}\mlap_{b}^{\top}\right)^{\dagger}\mU_{b}\right)\leq\tr\left(\mU_{b}^{\dagger}\mU_{b}\right)=n-1$.
So
\[
C_{uv}\leq\frac{n-1}{\min_{i\neq j:\mw_{ij}>0}\abs{\mU_{b}}_{ij}}\,.
\]

Finally, since $\left(\mU_{b}\right)_{ij}\geq\left(\mlap_{b}\right)_{ij}=\left(\left(\mI-\mw^{\top}\right)\ms\right)_{ij}$,
we have that
\begin{eqnarray*}
\min_{i\neq j:\mw_{ij}>0}\abs{\mU_{b}}_{ij} & \geq & \min_{i\neq j:\mw_{ij}>0}\left(\mw^{\top}\ms\right)_{ij}\geq s_{\textnormal{\ensuremath{\min}}}w_{\textnormal{min}}\,.
\end{eqnarray*}
Therefore
\[
C_{uv}\leq\frac{n-1}{s_{\textnormal{min}}w_{\textnormal{min}}}\,.
\]
\end{proof}
Finally, we can combine Lemmas~\ref{lem:sketch-exact-stationary},
\ref{prop:produce-exact-stationary}, and \ref{prop:approx-commute}
in order to obtain the following result:
\begin{thm}
Let $\mw\in\R_{\geq0}^{n\times n}$ be the random walk matrix of a
Markov chain associated with a strongly connected graph. Let $M\geq t_{\textnormal{mix}}$,
$\sigma\leq s_{\textnormal{min}}$, where $s_{\textnormal{min}}$
represents the smallest stationary probability, and let $w_{\textnormal{min}}$
be the smallest transition probability. Then Algorithm~\ref{alg:commute-time-final}
computes in time $O(\tsolve((\log\frac{nM}{\epsilon\sigma})^{2}+\frac{1}{\epsilon^{2}}\log\frac{n}{\epsilon\sigma w_{\textnormal{min}}}))$
an $O(\log n/\epsilon^{2})\times n$ matrix $\tilde{\my}$ such that,
$(1-\epsilon)C_{uv}\leq\left\Vert \tilde{\my}(e_{u}-e_{v})\right\Vert _{2}^{2}\leq(1+\epsilon)C_{uv}$,
with high probability in $n$.
\end{thm}
\begin{proof}
We start by proving correctness. First, we see that the commute times
in $\tilde{\mw}$ approximate the original ones within $1+\epsilon/2$.
Indeed, as seen in Lemma~\ref{prop:produce-exact-stationary}, running
Algorithm~\ref{alg:commute-time-approx} on $\mw$ with accuracy
$\epsilon_{0}=\frac{\sigma^{2}}{256\cdot n^{3}\log^{2}n\cdot M^{2}}\cdot\frac{\epsilon}{2}$
produces a walk matrix $\tilde{\mw}$ with stationary $\tilde{s}$,
whose entries approximate those of $\mw$ within additive error $\epsilon_{0}$.
Furthermore $s_{i}+\epsilon_{0}\sigma/2\geq\tilde{s}_{i}\geq s_{i}/(1+\epsilon_{0}\sigma)$,
which immediately implies $\norm{\tilde{s}-s}_{\infty}\leq\epsilon_{0}$.
Hence, by Lemma~\ref{prop:approx-commute}, $\abs{\tilde{C}_{uv}-C_{uv}}\leq\epsilon/2$
for all $u,v$. Then, running Algorithm~\ref{alg:commute-time-sketch}
on $(\tilde{\mw},\tilde{s})$ we obtain a matrix $\tilde{\my}$, which
we can use to sketch all commute times in $\tilde{\mw}$ with multiplicative
error $1+\epsilon/3$. Therefore $\tilde{\my}$ sketches commute times
in $\mw$ with multiplicative error $1+\epsilon$.

Next, we bound the total running time. The call to Algorithm~\ref{alg:commute-time-approx}
requires time $O(\tsolve(\log\frac{nM}{\epsilon_{0}\sigma})^{2})=O(\tsolve(\log\frac{nM}{\epsilon\sigma})^{2})$.
The running time of the call to Algorithm~\ref{alg:commute-time-sketch}
depends on the minimum stationary probability, and the smallest transition
probability of $\tilde{\mw}$:$\tilde{s}_{\textnormal{min}}$ and
$\tilde{w}_{\textnormal{min}}$, respectively. From Lemma~\ref{prop:produce-exact-stationary},
we see that these are within constant factors from the originals.
Therefore, the running time of this call is $O(\tsolve\cdot\log\frac{n}{\epsilon\sigma w_{\textnormal{min}}}\cdot\frac{1}{\epsilon^{2}})$.
Putting these bounds together, we obtain our bound.
\end{proof}
\begin{algorithm2e}[H]

\caption{Commute Time Sketch Computation}\label{alg:commute-time-final}

\SetAlgoLined

\textbf{Input: }Random walk matrix $\mw\in\R_{\geq0}^{n\times n}$,
upper bound $M\geq t_{\textnormal{mix}}$, lower bound $\sigma\leq s_{\textnormal{min}}$,
accuracy $\epsilon\in(0,1)$.

Run Algorithm~\ref{alg:commute-time-approx} on input $\mw$, $M$,
$\sigma$, $\epsilon_{0}=\frac{\sigma^{2}}{256\cdot n^{3}\log^{2}n\cdot M^{2}}\cdot\frac{\epsilon}{2}$

Let $(\tilde{\mw},\tilde{s})$ be the output produced by Algorithm~\ref{alg:commute-time-approx}.

Run Algorithm~\ref{alg:commute-time-sketch} on input $\tilde{\mw}$,
$\tilde{s}$, $\epsilon/3$

Let $\tilde{\my}$ be the matrix produced by Algorithm~\ref{alg:commute-time-sketch}.

\textbf{Output}: $\tilde{\my}$

\end{algorithm2e}


\appendix

\section{Numerical Stability \label{sec:numerical_stability}}

In this section we give brief justifications of the numerical stability
of our routine. For simplicity we will use fixed point arithmetic,
and show that $O(1)$ words are sufficient when the input fits in
word size. However, we do assume the numerical errors are adversarial:
specifically we assume that at each step, any value can be perturbed
adversarially by machine epsilon ($\epsilon_{mac})$.

Formally our assumptions include:
\begin{enumerate}
\item The numerical stability of Gaussian elimination / LU factorization.
Specifically given a matrix $M$ with condition number $\kappa$,
we can access its inverse $M^{-1}$ with additive error $\epsilon$
in $O(n^{3}\log^{O(1)}(\kappa\epsilon^{-1}))$ time by manipulating
words of length $O(\log(\kappa\epsilon^{-1}))$: see i.e. \cite{higham02},
Chapters 9 and 10, in particular Theorems 9.5. and 10.8.
\item \begin{sloppypar}Numerically stable solvers for undirected Laplacians:
given a graph Laplacian where all weights are in the range $[1,$U{]},
solves in it can be made to $\epsilon$ additive error in time $O(m\log^{O(1)}(n)\log(U\epsilon^{-1}))$
by manipulating words of length $O(\log(U\epsilon^{-1}))$. \end{sloppypar}
\item The stability of iterative refinement: once we obtain a constant factor
approximation, we can obtain an $\epsilon$-approximation in a numerically
stable manner by working with words of the same length.
\end{enumerate}
Our analyses center around bounding the condition numbers of the matrices
involved. This is because our algorithm makes crucial use of the Woodbury
formula, which preserves condition number bounds, but may severely
change the matrix entry-wise. Recall that $\kappa(\cdot)$ denotes
the condition number of a matrix, which in case of symmetric matrices/operators
is the maximum over minimum eigenvalue.

\subsection{Reduction to Eulerian Case}

Here it suffices to check that the iterative process outlined in Theorem
\ref{thm:stat-from-dd} can tolerate perturbations to its intermediate
vectors. Its proof relies on showing the 1-norm of a vector, $\norm{\md^{-1}(e^{(t-1)}-\alpha d)}_{1}$
is geometrically decreasing until up to $32\cdot\alpha\cdot|V|$.
Such a bound is friendly to perturbations by vectors that are entrywise
less than $\alpha$. Furthermore, if the entries of $D$ are at least
$U^{-1},$ perturbations less than $U^{-1}\alpha$ are sufficient.
The most important condition is that the intermediate matrices generated
by this reduction do not become ill-conditioned.
\begin{lem}
\label{lem:reductionCondNice}If Algorithm \ref{alg:stationary_computation}
takes a directed Laplacian $\mlap\in\R^{n\times n}$ where all non-zero
entries are within a factor of $\kappa$ of each other, and restart
parameter $\alpha$, then all intermediate matrices ($\mm^{(t)}$)
that it performs linear-system solve on satisfy $\kappa(\frac{1}{2}(\mm+\mm^{\top}))\leq O(\kappa^{2}\alpha^{-2}n^{5})$.
\end{lem}
\begin{proof}
The proof of Theorem~\ref{thm:stat-from-dd} in Section \ref{sec:computing_stationary:analysis:main_theorem}
gives that all entries in $\mathbf{DX}^{(t)}$ are within a ratio
of $20\alpha^{-2}n$ of each other. As $\mathbf{D}$ is the diagonal
of $\mathcal{\mlap}$, all entries in $\md$ are within a factor of
$\kappa n$ of each other, and therefore all entries of $\mx^{(t)}$are
within a factor of $20\kappa\alpha^{-2}n$ of each other.

Therefore all non-zero entries in the matrix $\mathcal{\mlap}\mx^{(t)}$
are within a factor of $20\kappa^{2}\alpha^{-2}n$ of each other.
Since the matrix $\mE^{(t)}$ is chosen so that $\mE^{(t)}\mx^{(t)}+\mathcal{\mlap}\mx^{(t)}$
is $\alpha$-RCDD and $\alpha\leq\frac{1}{2}$, the additional entries
can equal to at most the row/column sums of $\mathcal{\mlap}\mx^{(t)}$.
Therefore all non-zero entries in $\mathbf{M}$ are within a factor
of $20\kappa^{2}\alpha^{-2}n^{2}$ of each other, and the choice of
$\mU=\frac{1}{2}(\mm+\mm^{\top})$ means all non-zero entries of the
SDD matrix $\mU$ are within a factor of $40\kappa^{2}\alpha^{-2}n^{2}$
of each other. As condition numbers are invariant under scaling, we
may assume without loss of generality that they're in the range $[1,$$40\kappa^{2}\alpha^{-2}n^{2}]$
by rescaling.

The assumption of $\mlap$ being strongly connected means $\mU$ is
connected. Therefore the minimum non-zero eigenvalue of $\mU$ is
at least $0.1n^{-2}$, which occurs when $\mU$ is a path, and the
diagonal is either exactly equal to row / column sums, or larger by
$1$ on a single coordinate. The maximum on the other hand is bounded
by $Trace(\mU)\leq40\kappa^{2}\alpha^{-2}n^{3}$ , which gives that
the condition number of $\mU$ is bounded by $O(\kappa^{2}\alpha^{-2}n^{5})$.
\end{proof}

\subsection{Solving Eulerian Systems}

Lemma~\ref{lem:reductionCondNice} implies that the condition number
of $\mathbf{U}$ is polynomially bounded. Lemma~\ref{lem:lapbounds}
then implies that the condition number of $\mx$ is also bounded by
$O(\kappa^{2}\alpha^{-2}n^{5})$. Therefore, the routines in this
section are performing one-step preconditioning operations on polynomially
conditioned matrices, and intuitively are stable. Formally, there
are three main sources of numerical errors in this process:
\begin{enumerate}
\item The Woodbury matrix identity from Lemma \ref{lem:ShermanMorrison}.
\item Gaussian elimination involving the matrix $\mw_{S}^{-1}+\mr\widetilde{\mU}^{\dagger}\mr^{\top}$.
This matrix is polynomially conditioned, so LU/Cholesky factorizations
on it are stable by assumption.
\item The calls to preconditioned iterative methods through Lemma \ref{lem:complicated-matrix-solver}.
Since the relative condition number is bounded by $O(n^{2})$, this
is stable by stability results of Chebyshev iterations.
\end{enumerate}
To show these bounds we need to bound the condition numbers of the
input matrices.
\begin{lem}
When Theorem \ref{thm:euleriansolver} takes a matrix with condition
number $\kappa$ and entries between $[U^{-1},U]$ as input, all the
intermediate matrices generated have condition number $\leq\kappa n^{10}$,
and entries bounded by \textup{$[U^{-1}\kappa^{-1}n^{-10},U\kappa n^{10}]$}
\end{lem}
\begin{proof}
These distortions happen in several places where we manipulate the
matrices:
\end{proof}
\begin{enumerate}
\item Lemma \ref{lem:ultrasparsify}: here weights can be restricted to
the range $[1,m]$ by lower bounding the probabilities at $1/m$.
\item The Sherman-Morrison formula:

\begin{enumerate}
\item $\mathbf{W_{S}}$ has condition number $O(n^{2})$ as above
\item Entries of $\mathbf{U^{\dagger}}$are pairwise effective resistances,
and therefore bounded between $[U^{-1}n^{-2},Un^{2}]$, so the matrix
$\mr\widetilde{\mU}^{+}\mr^{\top}$ is well conditioned, and have
bounded entries.
\item Compounding this with multiplying by $\mr\widetilde{\mU}^{+}$gives
the bound on overall condition number and entry values
\end{enumerate}
\end{enumerate}
This allows us to control the overall accumulation of error through
direct propagation of vector error bounds.
\begin{lem}
If all the routines invoked in Theorem \ref{thm:euleriansolver} are
run to precision $U^{-1}\kappa^{-1}n^{-10}$, then the guarantees
of Theorem \ref{thm:euleriansolver} holds with $2\epsilon$.
\end{lem}
\begin{proof}
The overall proof scheme is to view all numerical errors as perturbations
with small $\ell_{2}$-norms after each step. Then the fact that the
matrix are well conditioned means constant number of these steps can
only increase error by a factor of $\mathrm{poly}(U,\kappa,n)$.

There are two main numerical routines invoked in the course of the
algorithm.:

1. All the calls to $\mathbf{U^{\dagger}}$ are stable by the stability
analyses of operator based solvers, since $\mathbf{U}$ is well-conditioned.

2. The calls to Chebyshev iteration invoked in Lemma\ref{lem:complicated-matrix-solver}
has a condition number of $O(n^{2})$, so all the operators are polynomially
conditioned from the condition number bound on $\mathbf{Z}$. Therefore
we can invoke known analyses of stability of Chebyshev iteration,
such as Lemma 2.6.3. of \cite{Peng13:thesis}.

Aside from that, the only other operations are matrix-vector multiplications. 
\end{proof}

\section{Matrix Facts \label{sec:matrix_facts}}
\begin{lem}
\label{lem:rcdd-inv-monotone} Let $\mm\in\R^{n\times n}$ be an invertible
RCDD Z-matrix. Then for any diagonal matrix $\md\in\R_{\geq0}^{n\times n}$
and $\norm{(\md+\mm)^{-1}}_{2}\leq\norm{\mm^{-1}}_{2}$.
\end{lem}
\begin{proof}
For all $\alpha\in\R_{\geq0}$ let $\mm_{\alpha}\defeq\alpha\md+\mm$.
Let $\md_{\alpha}=\mdiag(\alpha\md+\mm)$ and let $\ma_{\alpha}=\mdiag(\mm)-\mm$.
Note that $\mm_{\alpha}=\md-\ma_{\alpha}$ for all $\alpha\geq0$.
Now since $\mm$ is invertible and RCDD we know that $\mm_{\alpha}$
is invertible and RCDD and 
\[
\mm_{\alpha}^{-1}=(\md_{\alpha}-\ma)^{-1}=\md_{\alpha}^{-1}(\mI-\ma\md_{\alpha}^{-1})=\sum_{i=0}^{\infty}\md_{\alpha}^{-1}(\ma\md_{\alpha}^{-1})^{i}\,.
\]
Consequently, we see that $\mm_{\alpha}^{-1}$ is a all positive matrix
and $(\mm_{\alpha}^{-1})^{\top}\mm_{\alpha}^{-1}$ is a all positive
matrix both of which decrease entrywise as $\alpha$ increases. Since
the maximum eigenvector of $(\mm_{\alpha}^{-1})^{\top}\mm_{\alpha}^{-1}$
clearly has entries of all the same sign we see that $\norm{\mm_{\alpha}^{-1}}_{2}$
decreases monotonically as $\alpha$ increases. Considering $\alpha\in\{0,1\}$
yields the result.

\end{proof}
\begin{lem}
\label{lem:lap-onenorm-bound} Let $\mm\in\R^{n\times n}$ be either
CDD or RDD. Then $\norm{\mm x}_{1}\leq2\cdot\tr(\mm)\cdot\norm x_{\infty}$
for all $x\in\R^{n}$. Consequently, $\norm{\mm}_{2}\leq2\cdot\tr(\mm)$. 
\end{lem}
\begin{proof}
Since $\mm$ is CDD or RDD we know $\mm_{ii}\geq0$ for all $i\in[n]$
and $\tr(\mm)\geq\sum_{i\neq j\in[n]}|\mm_{ij}|$. Consequently, 
\[
\norm{\mm x}_{1}\leq\sum_{i,j\in[n]}|\mm_{ij}|\cdot\norm x_{\infty}=\left(\sum_{i\in[n]}\mm_{ii}+\sum_{i\neq j\in[n]}|\mm_{ij}|\right)\cdot\norm x_{\infty}\leq2\cdot\tr(\mm)\cdot\norm x_{\infty}\,.
\]
Therefore, $\norm{\mm x}_{2}\leq\norm{\mm x}_{1}\leq2\cdot\tr(\mm)\cdot\norm x_{\infty}\leq2\cdot\tr(\mm)\cdot\norm x_{2}$.
\end{proof}

\begin{lem}
\label{lem:inverse-alphacdd_bound} Let $\mm\in\R^{n\times n}$ be
an $\alpha$-CDD Z-matrix. For all $x\in\R_{\geq0}^{n}$ we have $\mdiag(\mm)^{-1}x\leq\mm^{-1}x\leq\mdiag(\mm)^{-1}x+\frac{\norm x_{1}}{\alpha}\mdiag(\mm)^{-1}1$
\end{lem}
\begin{proof}
Let $\md\defeq\mdiag(\mm)$ and let $\ma\defeq\md-\mm$ so $\mm=\md-\ma$.
We see that 
\[
\mm^{-1}=(\md-\ma)^{-1}=\md^{-1}(\mI-\ma\md^{-1})=\sum_{i=0}^{\infty}\md^{-1}(\ma\md^{-1})^{i}
\]
Consequently
\[
\mm^{-1}x-\md^{-1}x=\md^{-1}\sum_{i=1}^{\infty}(\ma\md^{-1})^{i}x
\]
Now since $\mm$ is a $\alpha$-CDD Z-matrix we see that $(\ma\md^{-1})x$
is a non-negative vector with $\norm{\ma\md^{-1}x}_{1}\leq\frac{1}{1+\alpha}\norm x_{1}$.
Consequently, $\sum_{i=1}^{\infty}(\ma\md^{-1})^{i}x$ is a non-negative
vector with $\ellOne$ norm at most 
\[
\sum_{i=1}^{\infty}\left(\frac{1}{1+\alpha}\right)^{i}=\left(1-\frac{1}{1+\alpha}\right)^{-1}-1=\frac{1+\alpha}{\alpha}-1=\frac{1}{\alpha}\,.
\]
\end{proof}
\begin{lem}
\label{lem:cond-num-bd}Let $\mm\in\R^{n\times n}$, and let a diagonal
matrix $\md\succeq0$. Then $\norm{(\mm\md^{-1})^{\dagger}}_{2}\leq\norm{\mm^{\dagger}}_{2}\cdot\norm{\md}_{2}$.
Furthermore, $\kappa(\mm\md^{-1})=\kappa((\mm\md^{-1})^{\dagger})\leq\kappa(\mm)\kappa(\md)$.
\end{lem}
\begin{proof}
By definition, and using the fact that $\ker((\mm\md^{-1})^{\dagger})=\ker(\md^{-1}\mm^{\top})\perp\im(\mm\md^{-1})$,
we have 
\begin{eqnarray*}
\norm{(\mm\md^{-1})^{\dagger}}_{2} & = & \max_{x\perp\ker((\mm\md^{-1})^{\dagger}),x\neq0}\frac{\norm{(\mm\md^{-1})^{\dagger}x}_{2}}{\norm x_{2}}=\max_{y:\mm\md^{-1}y\neq0}\frac{\norm{(\mm\md^{-1})^{\dagger}\mm\md^{-1}y}_{2}}{\norm{\mm\md^{-1}y}_{2}}\\
 & = & \max_{z:\mm z\neq0}\frac{\norm{\md z}_{2}}{\norm{\mm z}_{2}}=\max_{z:\mm z\neq0}\left(\frac{\norm z_{2}}{\norm{\mm z}_{2}}\cdot\frac{\norm{\md z}_{2}}{\norm z_{2}}\right)\\
 & \leq & \max_{z:\mm z\neq0}\frac{\norm z_{2}}{\norm{\mm z}_{2}}\cdot\max_{z\neq0}\frac{\norm{\md z}_{2}}{\norm z_{2}}=\norm{\mm^{\dagger}}_{2}\cdot\norm{\md}_{2}\,.
\end{eqnarray*}

Now we can plug in the definition for the condition number: $\kappa(\mm\md^{-1})=\kappa((\mm\md^{-1})^{\dagger})=\norm{(\mm\md^{-1})^{\dagger}}_{2}\norm{\mm\md^{-1}}_{2}$.
For the first part we use the upper bound we just proved. For the
second one, we bound it by $\norm{\mm\md^{-1}}_{2}\leq\norm{\mm}_{2}\norm{\md^{-1}}_{2}$.
Combining these bounds we see that $\kappa(\mm\md^{-1})\leq\norm{\mm^{\dagger}}_{2}\norm{\md}_{2}\cdot\norm{\mm}_{2}\norm{\md^{-1}}_{2}=\kappa(\mm)\kappa(\md)$.
\end{proof}
\begin{lem}
\label{lem:rcdd-to-eulerian-kappa}Let $\mm\in\R^{n\times n}$ be
an $\alpha$-RCDD matrix with diagonal $\md$, let $\mU_{\mm}=(\mm+\mm^{\top})/2$,
and let the $(n+1)\times n$ matrix $\mc=\begin{pmatrix}\mI_{n}\\
-\vones_{n}^{\top}
\end{pmatrix}$. Then the $(n+1)\times(n+1)$ matrix $\mc\mU_{\mm}\mc^{\top}$ satisfies
$\kappa(\mc\mU_{\mm}\mc^{\top})\leq2(1+\alpha)/\alpha\cdot n^{3}\cdot\kappa(\md)$.
\end{lem}
\begin{proof}
First, notice that if $\mm$ is $\alpha$-RCDD, then it is also the
case that $\mU_{\mm}$ is $\alpha$-RCDD. Then $\frac{\alpha}{1+\alpha}\md\preceq\mU_{\mm}$.
Similarly, we have $\mU_{\mm}\preceq2\md$. Therefore $\kappa(\mc\mU_{\mm}\mc^{\top})\leq\kappa(\mc\md\mc^{\top})\cdot2(1+\alpha)/\alpha$.
We see that $\mc\md\mc^{\top}$ is an undirected Laplacian with trace
at most $2\cdot\vec{1}^{\top}\md\vec{1}\leq2n\max_{i}\md_{i,i}$.
Also, its minimum eigenvalue is therefore at least $\min_{i}\md_{i,i}\cdot n^{2}$.
Hence $\kappa(\mc\md\mc^{\top})\leq n^{3}\kappa(\md)$, and our result
follows.
\end{proof}
\begin{lem}
\label{lem:norm-inverse-perturb}Let $\ma$ and $\ma'=\ma+\Delta\ma$
be matrices with the same null space. Then, if $\norm{\Delta\ma}<\norm{\ma^{\dagger}}^{-1}$,
then $\Vert(\ma+\Delta\ma)^{\dagger}\Vert\leq\Vert\ma^{\dagger}\Vert\cdot\frac{1}{1-\norm{\Delta\ma}\norm{\ma^{\dagger}}}$.
\end{lem}
\begin{proof}
Let $x$ be the largest singular vector of $\ma'^{\dagger}$, i.e.
it is perpendicular to the kernel, and maximizes $\frac{\norm x}{\norm{\ma'x}}$.
Then 
\begin{align*}
\frac{\norm x}{\norm{\ma'x}} & \leq\frac{\norm x}{\abs{\norm{\ma x}-\norm{\Delta\ma x}}}\,.
\end{align*}

Since $\norm{\Delta\ma}<\norm{\ma^{\dagger}}^{-1}$, we have that
$\norm{\Delta\ma x}\leq\norm{\Delta\ma}\norm x<\norm{\ma^{\dagger}}^{-1}\cdot\norm x\leq\norm{\ma x}$.
Therefore we can drop the absolute value, and obtain

\begin{align*}
\frac{\norm x}{\norm{\ma'x}} & \leq\frac{\norm x}{\norm{\ma x}-\norm{\Delta\ma x}}=\frac{\norm x}{\norm{\ma x}}\cdot\frac{\norm{\ma x}}{\norm{\ma x}-\norm{\Delta\ma x}}\\
 & =\frac{\norm x}{\norm{\ma x}}\cdot\frac{1}{1-\norm{\Delta\ma x}/\norm{\ma x}}\leq\frac{\norm x}{\norm{\ma x}}\cdot\frac{1}{1-\norm{\Delta\ma}\norm{\ma^{\dagger}}}\,.
\end{align*}

Since $\frac{\norm x}{\norm{\ma x}}$ is bounded by $\norm{\ma^{\dagger}}$,
our inequality follows.
\end{proof}
\begin{lem}
\label{lem:sol-perturb}Let $\ma$ and $\ma'=\ma+\Delta\ma$ be matrices
with the same left and right null space, such that $\norm{\Delta\ma}<\norm{\ma^{\dagger}}^{-1}$.
Let $b$ be a vector in the image, and let $x=\ma^{\dagger}b$, and
$\Delta x$ such that $x+\Delta x=(\ma+\Delta\ma)^{\dagger}b$. Then
\[
\frac{\norm{\Delta x}}{\norm x}\leq\frac{\norm{\ma^{\dagger}}\norm{\Delta\ma}}{1-\norm{\ma^{\dagger}}\norm{\Delta\ma}}\,.
\]
\end{lem}
\begin{proof}
Since $(\ma+\Delta\ma)(x+\Delta x)=b$, we have $\Delta x=\ma^{\dagger}\left(-\Delta\ma x-\Delta\ma\Delta x\right)$.
Applying triangle inequality we see that $\norm{\Delta x}\leq\norm{\ma^{\dagger}}\norm{\Delta\ma}\left(\norm x+\norm{\Delta x}\right)$.
Therefore, after rearranging terms we obtain 
\begin{eqnarray*}
\norm{\Delta x}\left(1-\norm{\ma^{\dagger}}\norm{\Delta\ma}\right) & \leq & \norm x\cdot\norm{\ma^{\dagger}}\norm{\Delta\ma}\,.
\end{eqnarray*}

Finally, since $\norm{\ma^{\dagger}}\norm{\Delta\ma}<1$, we get
\[
\frac{\norm{\Delta x}}{\norm x}\leq\frac{\norm{\ma^{\dagger}}\norm{\Delta\ma}}{1-\norm{\ma^{\dagger}}\norm{\Delta\ma}}\,.
\]
\end{proof}

\section{\label{sec:Deriving-the-Identities}Deriving the Identities from
Section \ref{subsec:app-commute}}

\subsection{\label{subsec:Hitting-Time-Formula}Formula for Hitting Times (Proof
of Lemma \ref{prop:hitting-time-rank-1})}

We derive the formula for $H_{uv}$ from first principles. Let $h^{v}$
be a vector where $h_{u}^{v}=H_{uv}$. Then, applying the definition
we have that 
\begin{eqnarray*}
h_{u}^{v} & = & \begin{cases}
0, & \mbox{if }u=v\\
1+\sum_{w}\mbox{\ensuremath{\mw}}_{wu}\cdot h_{w}^{v}, & \mbox{if }u\neq v
\end{cases}=\begin{cases}
0, & \mbox{if }u=v\\
1+\left(\mbox{\ensuremath{\mw}}^{\top}h^{v}\right)_{u}, & \mbox{if }u\neq v
\end{cases}\,.
\end{eqnarray*}
Rearranging terms, we see that $h^{v}$ is a solution to the linear
system 
\[
\left[\begin{array}{cc}
\left[\mI-\mw^{\top}\right]_{B,B} & \left[\mI-\mw^{\top}\right]_{B,v}\\
\vzero & 1
\end{array}\right]\left[\begin{array}{c}
h_{B}^{v}\\
h_{v}^{v}
\end{array}\right]=\left[\begin{array}{c}
\vec{1}\\
0
\end{array}\right]\,,
\]
 where $B=V\setminus\left\{ v\right\} $. Since every row is diagonally
dominant, with the last row being strictly so, we see that a solution
to this system of equation is unique. Consequently, it suffices to
find a solution to this system of equations. Now, let $x$ be any
solution to $(\mI-\mw)^{\top}x=d$ where $d_{u}=1$ for $u\neq v$
and $d_{v}=1-\frac{1}{s_{v}}$. Since the graph associated to $\mw$
is strongly connected, and $d\perp s$, there exists such an $x$
which is unique up to adding multiples of $\vec{1}$. Consequently,
if we let $y=x-x_{v}\vec{1}$ then we see that $[(\mI-\mw^{\top})y]_{u}=[(\mI-\mw^{\top})x]_{u}=d_{u}=1$
for $u\neq v$, and $y_{v}=0$. Therefore, $y=h^{v}$ and
\[
H_{uv}=y_{u}=\vec{1}_{u}^{\top}(\mI-\mw^{\top})^{\dagger}d-\vec{1}_{v}^{\top}(\mI-\mw^{\top})^{\dagger}d=(\vec{1}_{u}-\vec{1}_{v})^{\top}(\mI-\mw^{\top})^{\dagger}d\,.
\]
Plugging in $d=\vec{1}-\vec{1}_{v}\cdot\frac{1}{s_{v}}$, we obtain 

\[
H_{uv}=(\vec{1}-\vec{1}_{v}\cdot\frac{1}{s_{v}})^{\top}(\mI-\mw)^{\dagger}(\vec{1}_{u}-\vec{1}_{v})\,.
\]

\subsection{\label{subsec:commute-formula}Formula for Commute Times (Proof of
Lemma \ref{prop:commute-time})}

Using the identity from Lemma \ref{prop:hitting-time-rank-1}, we
can prove a nice formula for commute times.
\begin{align*}
C_{uv} & =H_{uv}+H_{vu}\\
 & =(\vec{1}-\vec{1}_{v}\cdot\frac{1}{s_{v}})^{\top}(\mI-\mw)^{\dagger}(\vec{1}_{u}-\vec{1}_{v})+(\vec{1}-\vec{1}_{u}\cdot\frac{1}{s_{u}})^{\top}(\mI-\mw)^{\dagger}(\vec{1}_{v}-\vec{1}_{u})\\
 & =(\vec{1}-\vec{1}_{v}\cdot\frac{1}{s_{v}})^{\top}(\mI-\mw)^{\dagger}(\vec{1}_{u}-\vec{1}_{v})-(\vec{1}-\vec{1}_{u}\cdot\frac{1}{s_{u}})^{\top}(\mI-\mw)^{\dagger}(\vec{1}_{u}-\vec{1}_{v})\\
 & =(\vec{1}_{u}\cdot\frac{1}{s_{u}}-\vec{1}_{v}\cdot\frac{1}{s_{v}})^{\top}(\mI-\mw)^{\dagger}(\vec{1}_{u}-\vec{1}_{v})\\
 & =(\vec{1}_{u}-\vec{1}_{v})^{\top}\ms^{-1}(\mI-\mw)^{\dagger}(\vec{1}_{u}-\vec{1}_{v})\,.
\end{align*}

Furthermore, this also gives 
\begin{eqnarray*}
C_{uv} & = & (\vec{1}_{u}-\vec{1}_{v})((\mI-\mw)\ms)^{\dagger}(\vec{1}_{u}-\vec{1}_{v})\,.
\end{eqnarray*}
While in general $((\mI-\mw)\ms)^{\dagger}\neq\ms^{-1}(\mI-\mw)^{\dagger},$
in this case we notice that 
\[
(\vec{1}_{u}-\vec{1}_{v})^{\top}\left(\ms^{-1}(\mI-\mw)^{\dagger}-((\mI-\mw)\ms)^{\dagger}\right)(\vec{1}_{u}-\vec{1}_{v})=0\,,
\]
and the result follows. 

\begin{sloppypar}Indeed, we can immediately see that $\left((\mI-\mw)^{\dagger}-\ms((\mI-\mw)\ms)^{\dagger}\right)(\vec{1}_{u}-\vec{1}_{v})\in\ker(\mI-\mw)$:
using the fact that $\vec{1}_{u}-\vec{1}_{v}\perp\ker(\mI-\mw^{\top})$,
so $\vec{1}_{u}-\vec{1}_{v}\in Im(\mI-\mw)$, we obtain $(\mI-\mw)\left((\mI-\mw)^{\dagger}-\ms((\mI-\mw)\ms)^{\dagger}\right)(\vec{1}_{u}-\vec{1}_{v})=(\mI_{Im(\mI-\mw)}-\mI_{Im((\mI-\mw)\ms)})(\vec{1}_{u}-\vec{1}_{v})=0$.
Therefore, $(\vec{1}_{u}-\vec{1}_{v})^{\top}\left(\ms^{-1}(\mI-\mw)^{\dagger}-((\mI-\mw)\ms)^{\dagger}\right)(\vec{1}_{u}-\vec{1}_{v})=(\vec{1}_{u}-\vec{1}_{v})^{\top}\ms^{-1}\cdot k$
where $k\in\ker(\mI-\mw)$. Hence, writing $k=\alpha s$ for some
constant $\alpha$, we have $(\vec{1}_{u}-\vec{1}_{v})^{\top}\ms^{-1}k=(\vec{1}_{u}-\vec{1}_{v})^{\top}\ms^{-1}\cdot\alpha s=(\vec{1}_{u}-\vec{1}_{v})^{\top}\cdot\alpha\vec{1}=0$,
which yields the desired result.\end{sloppypar}

Finally, for the third identity, since the vectors on both sides are
symmetric, we can write 
\[
C_{uv}=\left(\vec{1}_{u}-\vec{1}_{v}\right)^{\top}\mlap_{b}^{\dagger}\left(\vec{1}_{u}-\vec{1}_{v}\right)=\frac{1}{2}\left(\vec{1}_{u}-\vec{1}_{v}\right)^{\top}\left(\mlap_{b}^{\dagger}+\left(\mlap_{b}^{\dagger}\right)^{\top}\right)\left(\vec{1}_{u}-\vec{1}_{v}\right).
\]

We will show that the matrix that appears in the middle is, as a matter
of fact, symmetric positive semi-definite. In order to do so, we first
use the fact that
\[
\mlap_{b}^{\dagger}+\left(\mlap_{b}^{\dagger}\right)^{\top}=\mlap_{b}^{\dagger}\left(\mlap_{b}+\mlap_{b}^{\top}\right)\left(\mlap_{b}^{\dagger}\right)^{\top}\,,
\]
 which follows immediately by expanding the right hand side. This
identity critically uses the fact that $\mlap_{b}$ is Eulerian, and
therefore the right and left null spaces coincide. Plugging this into
the previous formula for commute times yields our desired identity.

\subsection{\label{subsec:Escape-Probabilities}Formula for Escape Probabilities
(Proof of Lemma \ref{prop:escape-prob})}

Just like we did for hitting times, we can derive the formula for
escape probabilities from first principles. In particular, the definition
gives us 
\begin{eqnarray*}
p_{i} & = & \begin{cases}
1, & \mbox{if }i=u\\
0, & \mbox{if }i=v\\
\sum_{j}\mw_{ij}\cdot p_{j}, & \mbox{otherwise}
\end{cases}=\begin{cases}
1, & \mbox{if }i=u\\
0, & \mbox{if }i=v\\
\mw\cdot p, & \mbox{otherwise}
\end{cases}\,.
\end{eqnarray*}
Equivalently, this yields the linear system:

\[
\left[\begin{array}{ccc}
[\mI-\mw]_{B,B} & [\mI-\mw]_{B,s} & [\mI-\mw]_{B,t}\\
0 & 1 & 0\\
0 & 0 & 1
\end{array}\right]\left[\begin{array}{c}
p_{B}\\
p_{u}\\
p_{v}
\end{array}\right]=\left[\begin{array}{c}
\vec{0}\\
1\\
0
\end{array}\right]\,,
\]

where $B=V\setminus\left\{ u,v\right\} $. Again, since every row
is diagonally dominant, and the last two are strictly so, we see that
a solution to this system is unique. So we only need to find a solution
to this system.

Now let $x$ be a solution to $(\mI-\mw)x=\vec{1}_{u}-\vec{1}_{v}$.
Since $\vec{1}_{u}-\vec{1}_{v}\perp\vec{1}$, the solution exists.
Also, since the graph associated to $\mw$ is strongly connected,
it is unique up to adding multiples of the stationary. Therefore letting
$y=x-s\cdot\frac{x_{v}}{s_{v}}$ all the constraints of the system
are satisfied, except $y_{u}=1$. Setting $y'=\frac{1}{y_{u}}\cdot y=\left(x_{u}-s_{u}\cdot\frac{x_{v}}{s_{v}}\right)^{-1}\left(x-s\cdot\frac{x_{v}}{s_{v}}\right)=\left(s_{u}\left(\frac{x_{u}}{s_{u}}-\frac{x_{v}}{s_{v}}\right)\right)^{-1}\cdot\left(x-s\cdot\frac{x_{v}}{s_{v}}\right)$
yields a solution to the system. Consequently, $p=y'$ and the result
follows.

Next, we bound the coefficients.

We have $\abs{\alpha}=\abs{\frac{x_{v}}{s_{v}}}\leq\frac{1}{s_{\textnormal{min}}}\cdot\norm{(\mI-\mw)^{\dagger}(\vec{1}_{u}-\vec{1}_{v})}_{\infty}\leq\frac{1}{s_{\textnormal{min}}}\cdot\norm{(\mI-\mw)^{\dagger}(\vec{1}_{u}-\vec{1}_{v})}_{\infty}\leq\frac{1}{s_{\textnormal{min}}}\cdot\norm{(\mI-\mw)^{\dagger}(\vec{1}_{u}-\vec{1}_{v})}_{1}\leq\frac{1}{s_{\textnormal{min}}}\cdot\norm{(\mI-\mw)^{\dagger}}_{1}\norm{\vec{1}_{u}-\vec{1}_{v}}_{1}$.
Plugging in Theorem \ref{thm:mixingTimesSingular}, we obtain that
$\abs{\alpha}\leq\frac{t_{\textnormal{mix}}}{s_{\textnormal{min}}}\cdot8\sqrt{n}\log_{2}n$.

Also, $\beta=\frac{1}{s_{u}\left(\frac{x_{u}}{s_{u}}-\frac{x_{v}}{s_{v}}\right)}=\frac{1}{s_{u}(\vec{1}_{u}-\vec{1}_{v})^{\top}\ms^{-1}(\mI-\mw)^{\dagger}(\vec{1}_{u}-\vec{1}_{v})}$.
As seen from Lemma \ref{prop:commute-time}, the term $(\vec{1}_{u}-\vec{1}_{v})^{\top}\ms^{-1}(\mI-\mw)^{\dagger}(\vec{1}_{u}-\vec{1}_{v})$
is precisely the commute time $C_{uv}$ which satisfies $1\leq C_{uv}$
by definition. Also, we have that $C_{uv}=(\vec{1}_{u}-\vec{1}_{v})^{\top}\ms^{-1}(\mI-\mw)^{\dagger}(\vec{1}_{u}-\vec{1}_{v})\leq\norm{\ms^{-1}(\vec{1}_{u}-\vec{1}_{v})}_{\infty}\cdot\norm{(\mI-\mw)^{\dagger}(\vec{1}_{u}-\vec{1}_{v})}_{1}\leq\frac{1}{s_{\textnormal{min}}}\cdot\norm{(\mI-\mw)^{\dagger}(\vec{1}_{u}-\vec{1}_{v})}_{1}\leq\frac{t_{\textnormal{mix}}}{s_{\textnormal{min}}}\cdot8\sqrt{n}\log_{2}n$
. Therefore $\frac{1}{t_{\textnormal{mix}}\cdot8\sqrt{n}\log_{2}n}\leq\beta\leq\frac{1}{s_{\textnormal{\ensuremath{\min}}}}$.

\end{document}